%% file: arxiv.tex
\documentclass{svjour4}                     
\smartqed  
\input{preambel}
\journalname{Algorithmica}
\begin{document}

\title{Online Square-into-Square Packing
	\thanks{A preliminary extended abstract appears in the Proceedings of the 16th International Workshop APPROX 2013~\cite{fh-ossp-13}. Partially funded by DFG grant FE407/17-1 within the DFG Research Unit 1800, ``Controlling Concurrent Change''.}
	}
\author{ S\'andor~P.~Fekete \and
	    Hella-Franziska~Hoffmann
}

\institute{
	 S\'andor P.~Fekete\at
              Department of Computer Science\\
              TU Braunschweig\\
              38106 Braunschweig\\
              Germany\\
              \email{s.fekete@tu-bs.de}
           \and
           Hella-Franziska Hoffmann \at
           David R. Cheriton School of Computer Science\\
	University of Waterloo\\
	200 University Avenue West\\
	Waterloo, ON Canada N2L 3G1\\
	Tel.: +1-519-888-4567 ext. 33440\\
	\email{hrhoffmann@uwaterloo.ca}
}

\date{}

%
\maketitle
\begin{abstract}
In 1967, Moon and Moser proved a tight bound on the critical
density of squares in squares: any set of squares with a 
total area of at most 1/2 can be packed into a unit square, which is tight. 
The proof requires full knowledge of the set,
as the algorithmic solution consists in sorting the objects by decreasing size,
and packing them greedily into shelves. Since then, the online version
of the problem has remained open; the best upper bound is still 1/2, while the currently best lower bound is 1/3, due to Han et al. (2008).
In this paper, we present a new lower bound of 11/32, based on a 
dynamic shelf allocation scheme, which may be interesting in itself.

We also give results for the closely related problem in which the size of the
square container is not fixed, but must be dynamically increased in order
to accommodate online sequences of objects. For this variant, we establish an upper
bound of 3/7 for the critical density, and a lower bound of 1/8.
When aiming for accommodating an online sequence of squares, this 
corresponds to a $2.82\ldots$-competitive
method for minimizing the required container size, 
and a lower bound of $1.33\ldots$ for the achievable factor.
\end{abstract}


\input{introduction}

\input{fixedContainer}
\input{dynamicContainer}

\input{conclusion}

%

\section*{Acknowledgement}
We thank the anonymous reviewers for many helpful comments that improved the overall manuscript.


\bibliographystyle{spmpsci}      

 \bibliography{lit}



\end{document}

%% file: preambel.tex
\usepackage{amssymb,amsmath}
\usepackage{xspace}
\usepackage{todonotes}
\usepackage{refcount}
\usepackage{xcolor}

\usepackage{enumerate}
\usepackage{import}
\usepackage{subfigure}

\spnewtheorem{lem}{Lemma}[lemma]{\bfseries}{\itshape}

\newcommand{\old}[1]{{}}
\newcommand{\raus}[1]{{}}

\newcommand{\packedseq}{\ensuremath{\mathcal{P}}\xspace} 
\newcommand{\packedseqOf}[1]{\ensuremath{\mathcal{P}(#1)}\xspace}
\newcommand{\maxArea}{\ensuremath{\delta}\xspace}
\newcommand{\seq}{\ensuremath{\sigma}\xspace} 
\newcommand{\unitsq}{\ensuremath{U}\xspace} 
\newcommand{\headOf}[1]{\textit{head(#1)}\xspace}

\newcommand{\shelfEnd}[1]{\textit{end}(#1)\xspace}

\newcommand{\nextFitShelf}[1]{NFS(#1)}

\newcommand{\norm}[1]{\ensuremath{\lVert#1\rVert}\xspace}

\newcommand{\newshelf}{\ensuremath{\mathcal{S}}\xspace}
\newcommand{\newbuffer}{\ensuremath{\mathcal{B}}\xspace}
\newcommand{\newendbuffer}{\ensuremath{\mathcal{E}}\xspace}
\newcommand{\newinitial}{\ensuremath{\mathcal{I}}\xspace}
\newcommand{\newvertical}{\ensuremath{\mathcal{V}}\xspace}
\newcommand{\overlapshelf}{\ensuremath{\mathcal{L}}\xspace}
\newcommand{\Ebuffershelf}{\ensuremath{\bar{E}}\xspace}

\newcommand{\initialbuffer}[1]{\ensuremath{\mathcal{I}_{#1}}\xspace}

\newcommand{\eps}{\ensuremath{\varepsilon}\xspace}

\newcommand{\assigned}[1]{\textit{assigned}(#1)\xspace}
\newcommand{\occupied}[1]{\textit{occupied}(#1)\xspace}
\newcommand{\totalAreaOf}[1]{\ensuremath{\widetilde{\mathcal{A}}(#1)}\xspace}
\newcommand{\usedSection}[1]{\textit{usedSection}(\ensuremath{#1})\xspace}
\newcommand{\usLength}[1]{\ensuremath{{\ell_{#1}}}\xspace}
\newcommand{\extraOf}[1]{\textit{extra}(\ensuremath{#1})\xspace}
\newcommand{\mainPackingArea}{\ensuremath{M}\xspace}

\newcommand{\vBase}[2]{\ensuremath{V^{#1}_{#2}}\xspace}

\newcommand{\vOpen}{\ensuremath{V_{\text{open}}}\xspace}
\newcommand{\vClosed}{\ensuremath{V_{\text{closed}}}\xspace}

\newcommand{\vHeadB}{\ensuremath{V^{\text{B}}_{\text{head}}}\xspace}
\newcommand{\vHeadE}{\ensuremath{V^{\text{E}}_{\text{head}}}\xspace}
\newcommand{\vHeadEi}{\ensuremath{V^{\text{E}_i}	_{\text{head}}}\xspace}
\newcommand{\vHead}{\ensuremath{V_{\text{head}}}\xspace}

\newcommand{\upperhalf}{\ensuremath{\mathcal{H}_u}\xspace}
\newcommand{\lowerhalf}{\ensuremath{\mathcal{H}_\ell}\xspace}

\newcommand{\maxBrick}{\ensuremath{B_{max}}\xspace}

\newcounter{mycount}


%% file: introduction.tex
\section{Introduction}
Packing is one of the most natural and common optimization problems.
Given a set $\cal O$ of objects and a container $E$, find a placement
of all objects into $E$, such that no two overlap.
Packing problems are highly relevant in many practical applications,
both in geometric and abstract settings. Simple one-dimensional
variants (such as the {\sc Partition} case with two containers, or
the {\sc Knapsack} problem of a largest packable subset)
are NP-hard. Additional difficulties occur
in higher dimensions: as Leung et al.~\cite{ltwyc90} showed,
it is NP-hard even to check whether a given set of squares fits into a unit-square container.

When dealing with an important, but difficult optimization problem, it is crucial
to develop a wide array of efficient methods for distinguishing
feasible instances from the infeasible ones.
In one dimension, a trivial necessary and sufficient criterion is the
total size of the objects in comparison to the container. This makes it 
natural to consider a similar approach for the two-dimensional version:
{\em What is the largest number $\delta$, such that any family of squares
with area at most $\delta$ can be packed into a unit square?}
An upper bound of $\delta\leq 1/2$ is trivial: two squares
of size $1/2+\varepsilon$ cannot be packed.
As Moon and Moser showed in 1967~\cite{mm67}, $\delta=1/2$ is the correct critical
bound: sort the objects by decreasing size, and greedily pack them into a vertical
stack of one-dimensional ``shelves'',
i.e., horizontal subpackings whose height is defined by the largest object.

This approach cannot be used when the set of objects is not known a priori,
i.e., in an online setting. It is not hard to see that a pure shelf-packing approach
can be arbitrarily bad. However, other, more sophisticated approaches were
able to prove lower bounds for $\delta$: the current best bound
(established by Han et al.~\cite{hiz08}) is based on a relatively natural
recursive approach and shows that $\delta\geq 1/3$.

Furthermore, it may not always be desirable (or possible) to assume
a fixed container: 
the total area of objects may remain small, so 
a fixed large, square container may be wasteful.
Thus, it is logical to consider the size of the container itself 
as an optimization parameter. Moreover, considering a possibly {\em larger} container
reflects the natural optimization scenario 
in which the full set of objects {\em must} be accommodated, possibly by
paying a price in the container size. From this perspective, 
$1/\sqrt{\delta}$ yields a competitive factor for the 
minimum size of the container,
which is maintained at any stage of the process.
This perspective has been studied extensively for the case 
of an infinite strip, but not for an adjustable square.

\subsection{Our Results}
We establish a new best lower bound of $\delta\geq 11/32$ 
for packing an online sequence of squares into a fixed square container,
breaking through the threshold of $1/3$ that
is natural for simple recursive approaches based on brick-like structures.
Our result is based on a two-dimensional system of multi-directional
shelves and buffers, which are dynamically allocated and updated.
We believe that this approach is interesting in itself, as it may 
not only yield worst-case estimates, but also provide a
possible avenue for further improvements, and be useful as an algorithmic
method. 

As a second set of results, we establish the first upper and 
lower bounds for a square container, which is dynamically enlarged, but must maintain
its quadratic shape. In particular, we show 
that there is an upper bound of $\delta\leq 3/7<1/2$ for the critical density,
and a lower bound of $1/8\leq\delta$; when focusing on the minimum
size of a square container, these results correspond to a $2.82\ldots$-competitive
factor, and a lower bound of $1.33\ldots$ for the achievable factor by any 
deterministic online algorithm.

\subsection{Related Work}
Two- and higher-dimensional problems of packing rectangular objects into rectangular containers
have received a considerable amount of attention; see Harren's Ph.D.\ thesis~\cite{phdharren} for a relatively recent
survey. Many of the involved ideas are closely or loosely related to some of the ideas
of our paper. We summarize many of the related papers, with particular attention dedicated to
those that are of direct significance for our approach.

\paragraph{Offline Packing of Squares.}
One of the earliest considered packing variants is the problem of finding
a dense square packing for a rectangular container.
In 1966 Moser~\cite{m66} first stated the question as follows:
\begin{center}
``What is the smallest number $A$ such that any family of objects with total
area at most $1$ can be packed into a rectangle of area $A$?''
\end{center}
The offline case has been widely studied since 1966; 
there is a long list of results for packing squares into a rectangle. 
Already in 1967, Moon and Moser~\cite{mm67} gave the first bounds for $A$:
any set of squares with total area at most
$1$ can be packed into a square with side lengths $\sqrt2$, which shows
$A \leq 2$, and thus $\delta\geq 1/2$;
they also proved $A \geq 1.2$.
Meir and Moser \cite{mm68} showed that
any family of squares each with side lengths $\leq x$ and total area $A$ can
be packed into a rectangle of width $w$ and height $h$, if $w,h \geq x$ and
$x^2 + (w-x)(h-x) \geq A$; 
they also proved that any family of $k$-dimensional cubes with side
lengths $\leq x$ and total volume $V$ can be packed into a rectangular
parallelepiped with edge lengths $a_1, \ldots, a_k$ if $a_i \geq x$ for
$i=1,\ldots ,k$ and $x^k + \prod_{i=1}^k{(a_i-x)} \geq V$.
Kleitman and Krieger improved the upper bound on $A$ to $\sqrt3 \approx
1.733$~\cite{kk70} and to $4/\sqrt6 \approx 1.633$~\cite{kk75} by
showing that any finite family of squares with total area $1$ can be packed
into a rectangle of size $\sqrt{2}\times 2/\sqrt{3}$.
Novotn{\'y} further improved the bounds to $1.244 \approx (2 + \sqrt{3})/3 \leq A < 1.53$
in 1995~\cite{n95} and 1996~\cite{n96}. The current best known upper bound
of $1.3999$ is due to Hougardy~\cite{h11}.
There is also a considerable number of other related work on offline packing squares, cubes, or hypercubes;
see~\cite{ck-soda04,js-ptas08,h09} for prominent examples.

\paragraph{Online Packing of Squares into a Square.}
In 1997, Januszewski and Lassak~\cite{jl97} studied the online version of the
dense packing problem. In particular, they proved that for $d\geq5$, every online
sequence
of $d$-dimensional cubes of total volume $2(\frac{1}{2})^d$ can be packed
into the unit cube. For lower dimensions, they established online methods
for packing (hyper-) cubes and squares with a total volume of at most $\frac{3}{2}(\frac{1}{2})^d$
and $\frac{5}{16}$ for $d\in\{3,4\}$ and $d=2$, respectively. The results are
achieved by performing
an online algorithm that subsequently divides the unit square into rectangles with
aspect ratio $\sqrt2$. In the following, we call these
rectangles {\em bricks}.
The best known lower bound of $2(\frac{1}{2})^d$ for any $d \geq 1$ was presented
by Meir and Moser \cite{mm68}.

Using a variant of the brick algorithm, Han et al.~\cite{hiz08} extended the
result to packing a 2-dimensional sequence with total area $\leq 1/3$ into
the unit square.

A different kind of online square packing was considered by 
Fekete et al.~\cite{fks-osp-09,fks-ospg-14}. The container is an unbounded strip,
into which objects enter from above in a Tetris-like fashion; any new
object must come to rest on a previously placed object, and the 
path to its final destination must be collision-free. Their best 
competitive factor is $34/13\approx 2.6154\ldots$, which 
corresponds to an (asymptotic) packing density of $13/34\approx 0.38\ldots$.

\paragraph{Other Online Packing of Squares.}
There are various ways to generalize online packing of squares; see Epstein and van Stee~\cite{es-soda04,es05,es07} for online bin packing variants
in two and higher dimensions. In this context, also see parts of Zhang et al.~\cite{zcchtt10}.

\paragraph{Online Packing of Rectangles.}
A natural generalization of online packing of squares is online packing of rectangles,
which have also received a serious amount of attention. Most notably, online strip packing
has been considered; for prominent examples, see Azar and Epstein~\cite{ae-strip97}, who employ 
shelf packing, and Epstein and van Stee~\cite{es-soda04}.

\paragraph{Packing into One Container.}
Offline packing of rectangles into a unit square or rectangle has also been considered
in different variants; for examples, see \cite{fgjs05}, as well as \cite{jz-profit07}.
Particularly interesting for methods for online packing into a single container may be the work by Bansal et
al.~\cite{bcj-struct-09}, who show that for any complicated packing of rectangular items into a rectangular container,
there is a simpler packing with almost the same value of items.

\paragraph{Two-Dimensional Bin Packing.}
Packing squares or rectangles into a minimum number of square boxes amounts to two-dimensional
bin packing, which is closely related to packing into a single container. Arguably, bin packing is the 
two-dimensional packing problem that has received the most attention from an algorithmic perspective. 
See~\cite{harmony,bs-soda04,ck-soda04,clm-bin05,bls-focs05,bck-bpmd-06,bcs-bin06,c-stages08,bcs-namsc-09,zcchtt10,jp-soda13,hjp-bin13,bk-soda14}
for particularly relevant work. Most of these papers consider offline problems, with notable exceptions already cited above.

\paragraph{Resource Augmentation.}
Our study of online packing into a dynamic square container can be interpreted as a variant of 
resource augmentation, which has been studied in the context of two-dimensional packing by several other authors, including
\cite{clm-small06,c-aug-06,fgj-aug08,js-aug09}.

\paragraph{Strip Packing.}
Dynamically expanding a square container (as presented in Section~\ref{sec:dynamicContainer}) 
can be seen as a variation of increasing a container along only one dimension, i.e., packing into a strip.
Two- and higher-dimensional offline strip packing has been studied intensively, see~\cite{kr-focs96,js-strip-05,js-mfcs07,bhi-sicomp13,hjp-strip14}
for prominent examples.

%% file: fixedContainer.tex
\section{Packing into a Fixed Container}\label{sec:recShelfAlg}
As noted in the introduction, it is relatively easy to achieve a dense 
packing of squares in an offline setting: sorting the items by decreasing size makes sure
that a shelf-packing approach places squares of similar size together, so the
loss of density remains relatively small. This line of attack is not available in an
online setting; indeed, it is not hard to see that a brute-force shelf-packing
method can be arbitrarily bad if the sequence of items consists of a limited number
of medium-sized squares, followed by a large number of small ones.
Allocating different size classes to different horizontal shelves is not a
remedy, as we may end up saving space for squares that never appear, and 
run out of space for smaller squares in the process; on the other hand, fragmenting
the space for large squares by placing small ones into it may be fatal when a large
one does appear after all.

Previous approaches (in particular, the brick-packing algorithm) have 
side-stepped these difficulties by using a recursive subdivision scheme.
While this leads to relatively good performance guarantees (such as the previous 
record of 1/3 for a competitive ratio), it seems impossible to tighten
the lower bound; in particular, 1/3 seems to be a natural upper bound for this relatively
direct approach.
Thus, making progress on this natural and classical algorithmic problem
requires less elegant, but more powerful tools.

In the following we present a different approach for overcoming the crucial impediment of mixed
square sizes, and breaking through the barrier of 1/3. 
Our \emph{Recursive Shelf Algorithm} aims at subdividing the set of squares into
different size classes called {\em large}, {\em medium} and {\em small}, which are
packed into pre-reserved shelves. The crucial challenge is to dynamically update 
regions when one of them gets filled up before the other ones do;
in particular, we have to protect against the arrival of one large square,
several medium-sized squares, or many small ones.
To this end, we combine a number of new techniques:

\begin{itemize}
\item Initially, we assign carefully chosen horizontal strips for shelf-packing each size class.
\item We provide rules for dynamically updating shelf space when required by the sequence
of items. In particular, we accommodate a larger set of smaller squares by 
inserting additional {\it vertical} shelves into the space for larger squares whenever necessary.
\item In order to achieve the desired overall density, we maintain a set of buffers for 
overflowing strips. These buffers can be used for different size classes, 
depending on the sequence of squares.
\end{itemize}

With the help of these techniques, and a careful analysis, we are able to
establish $\maxArea \geq 11/32$. It should be noted that the development
of this new technique may be more significant than the numerical improvement
of the density bound: we are convinced that tightening the remaining gap
towards the elusive 1/2 will be possible by an extended (but more complicated) 
case analysis. 

The remainder of this section is organized as follows. In Section~\ref{sec:algOverview} we give
an overview of the algorithm. Section~\ref{sec:large} sketches
the placement of large objects, while
Section~\ref{sec:ceilingPacking} describes
the packing created with medium-sized squares. 
In Section~\ref{sec:generalShelfConcept} we describe the general
concept of shelf-packing that is used for the packing of
small squares discussed in Section~\ref{sec:packSmall}.
The overall performance is analyzed in Section~\ref{sec:algAnalysis}.

\subsection{Algorithm Overview}\label{sec:algOverview}
We construct a shelf-based packing in the unit square by
packing \emph{small}, \emph{medium} and \emph{large squares} separately.
We stop the Recursive Shelf Algorithm when
the packings of two different subalgorithms would overlap.
As it turns out, this can only happen when the total area 
of the given squares is greater than $11/32$; details are provided in the
``Combined Analysis'' of Section~\ref{sec:algAnalysis}, after describing the approach for
individual size classes.

In the following, we will subdivide the set of possible squares into
subsets, according to their size:
We let $H_k$ denote the height class belonging
to the interval $(2^{-(k+1)},2^{-k}]$. In particular, we call
all squares in $H_0$ {\em large}, all squares in $H_1$ {\em medium},
and all other squares (in $H_{\geq 2}$) {\em small}.

\subsection{Packing Large Squares}\label{sec:large}
The simplest packing subroutine is applied to large squares, i.e.,
of size greater than $1/2$.
We pack a square $Q_0 \in H_0$ into the top right corner of the unit square \unitsq.
Clearly, only one large 
square can be part of a sequence with total area $\leq 11/32$.
Hence, this single location for the squares in $H_0$ is sufficient.

\subsection{Packing Medium Squares}\label{sec:ceilingPacking}
We pack all medium squares (those with side lengths in $(1/4,1/2]$) separately;
note that there can be at most five of these squares, otherwise their total area is
already bigger than $3/8>11/32$.

\begin{figure}[t]
  \centering
  \subfigure[]{
    \includegraphics[width = 0.45\textwidth]{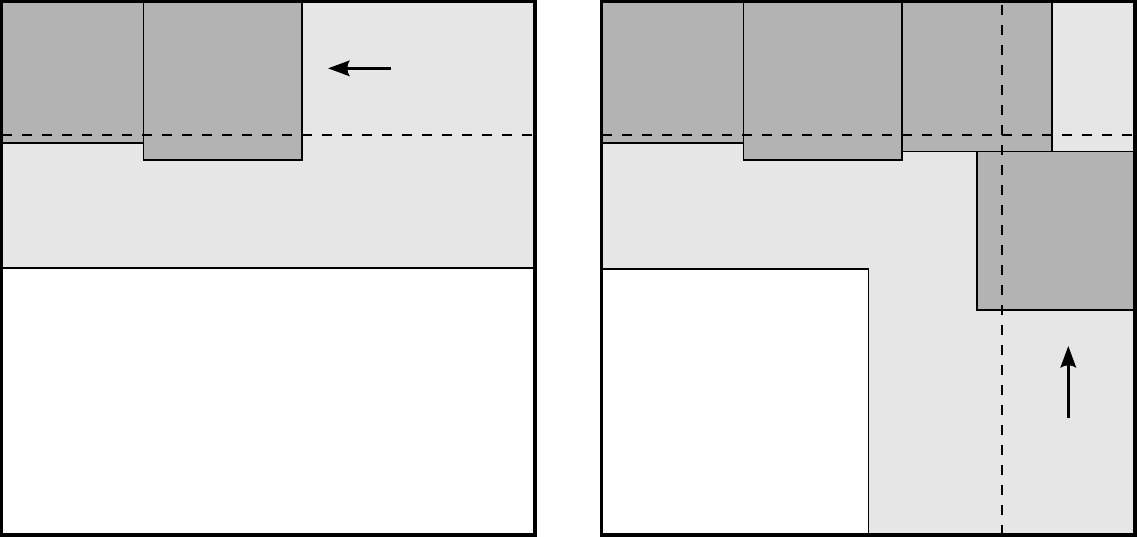}
    \label{fig:ceiling1}
  }
  \hskip2mm
  \subfigure[]{
    \def\svgwidth{.45\textwidth}
    \import{./}{ceiling2.tex}
  }
\caption{Packing medium squares (Subsection~\ref{sec:ceilingPacking}).
	(a): The L-shaped packing created with medium squares.
	 (b) Density consideration: The Ceiling Packing Algorithm
    packs at least as much as the area of the gray region ($R$) shown on the left.
	If a portion of $R$ remains uncovered by squares, a larger portion of $\unitsq\setminus R$ must be covered.
	}  
	\label{fig:ceiling2}
\end{figure}

We start with packing the $H_1$-squares from left to right coinciding with
the top of the unit square \unitsq. If a square would cross the right
boundary of \unitsq, we continue by placing the following squares from top
to bottom coinciding with the right boundary; see Fig.~\ref{fig:ceiling1}.
 
We call the corresponding subroutine the \emph{Ceiling Packing Algorithm}.
Without interference of other height classes, the algorithm succeeds in packing
any sequence of $H_1$-squares with total area $\leq 3/8$.
\begin{theorem}\label{thm:ceiling}
    The Ceiling Packing Subroutine packs any sequence of medium squares
    with total area at most $3/8$ into the unit square. 
\end{theorem}
    \begin{proof}
	  Assume that the Ceiling Packing subroutine fails to pack a square $Q$.
	 By construction, the algorithm successfully packs squares aligned with the top of \unitsq and the squares aligned with the right boundary of \unitsq until the space left at the bottom of \unitsq is too small to fit square $Q$.
	  We prove that in this case the total area of the given sequence \seq is
	  greater than the area $\norm{R} = 3/8$ of the gray region $R$ depicted in Fig.~\ref{fig:ceiling2}(b).
	The idea is that all of $R$ is covered by packed squares except for potentially a small portion of it in the top right that 
	can only be left uncovered as a result of receiving a large square that covers parts of $\unitsq\setminus R$.
	  Let $Q_2$ with side length $x_2$ be the first square that was not packed aligned with the top boundary of \unitsq and 
	  let $Q_1$ with side length $x_1$ be the square packed aligned with the top of \unitsq that touches the top boundary of
	 $Q_2$. Let $d_1$ be the distance of $Q_1$ to the right boundary of \unitsq and $d_2$ the distance of $Q_2$ to the top
	 boundary of \unitsq. Then we have $d_1 < x_2$ and $d_2 = x_1$.
	Because all medium squares have a side length of at least $1/4$, we have $x_1^2 = 1/4x_1 + (x_1 - 1/4)x_1 \geq 1/4x_1+ (d_2 - 1/4) \cdot 1/4$ and $x_2^2 = 1/4x_2 + (x_2 - 1/4)x_2 > 1/4x_2+ (d_1 - 1/4) \cdot 1/4$.
	Furthermore, we get that the set $\seq_1$ of all squares packed before $Q_1$  in \seq has a total area of at least $1/4 \cdot (1 - d_1 - x_1)$, and that the set $\seq_2$ of all squares that appeared after $Q_2$  in \seq has a total area of at least $1/4 \cdot (1 - d_2 - x_2)$. 
	Hence, we conclude
	 \begin{align*}
		\norm{\seq} & \geq  \norm{\seq_1} + \norm{Q_1} + \norm{Q_2} + \norm{\seq_2}		\\
						& >  \frac{1}{4} (1 - d_1 - x_1) + \frac{1}{4} x_1 + (d_2-\frac{1}{4}) \frac{1}{4}
						  + \frac{1}{4} x_2 + (d_1 - \frac{1}{4}) \frac{1}{4} + \frac{1}{4} (1 - d_2 - x_2) \\
						& = \frac{1}{4}\left( 1 - x_1 - d_1 + x_1 + d_2 - \frac{1}{4} + x_2 + d_1 - \frac{1}{4} + 1 - x_2 - d_2\right) =  3/8 .
	\end{align*}
    \end{proof}
    
\subsection{Shelf Packing}\label{sec:generalShelfConcept}
In this section we revisit the well-known shelf-packing algorithm that is used for packing small squares into the unit square.
Given a set of squares with maximum size $h$, a {\em shelf} \newshelf
is a subrectangle of the container that has height $h$; the {\em Next Fit Shelf Algorithm \nextFitShelf{\newshelf}} places incoming squares into \newshelf next to each other, until
some object no longer fits; see Fig.~\ref{fig:shelf}. When that happens,
the shelf is closed, and a new shelf gets opened. Before we analyze the density of the resulting packing, we introduce some notation.

\begin{figure}[t]
  \centering
  \hspace*{-2mm}
  \subfigure[]{
    \def\svgwidth{.44\textwidth}
    \import{./}{shelf1.tex}
    \label{fig:shelf}
    \hskip4mm
  }
  \subfigure[]{
    \def\svgwidth{.48\textwidth}
    \import{./}{shelf2.tex}
    \label{fig:shelfNotation}
  }
  \subfigure[]{
    \def\svgwidth{.5\textwidth}
    \import{./}{shelfend.tex}
    \label{fig:shelfEnd}
  }
  \caption{
    (a) A shelf \newshelf packed by \nextFitShelf{\newshelf} with squares of one height class.
    (b) Different areas of a shelf \newshelf.
       \occupied{\newshelf}: total area of squares in \packedseqOf{\newshelf} (dark gray),
       \usedSection{\newshelf}: region with light gray background
				(incl. \occupied{\newshelf}) to the left,
	\headOf{\newshelf}: region with light gray background to the right, and
	\shelfEnd{\newshelf}: hatched region to the right.
    (c) Assignment of $\extraOf{Q}$ (hatched) to \newshelf
	when square $Q$ causes an overflow of shelf \newshelf.
	\label{fig:shelfPacking}
    }
\end{figure}
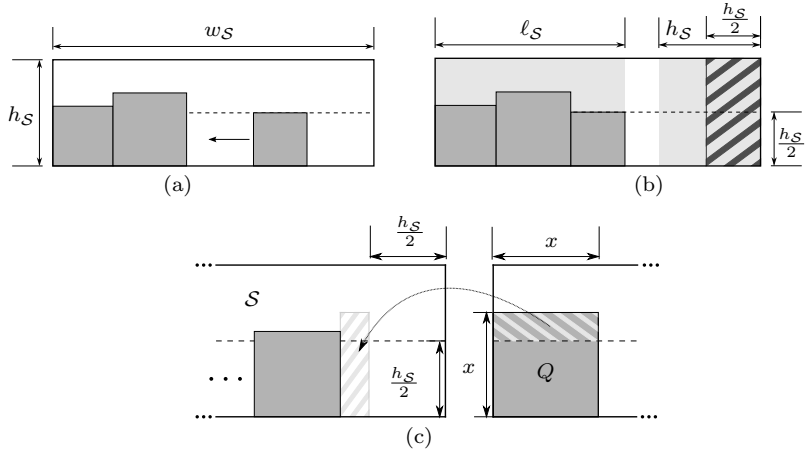

\paragraph{Notation.}
In the following we call a shelf with height $2^{-k}$ designed to accommodate squares of height class $H_k$ an \emph{$H_k$-shelf}.
We let $w_\newshelf$ denote the width of a shelf \newshelf,
$h_\newshelf$ denote its height and \packedseqOf{\newshelf} denote the set of squares packed into it.
We define \usedSection{\newshelf} as the horizontal section of \newshelf that
contains \packedseqOf{\newshelf} and \usLength{\newshelf} as its length; see Fig.~\ref{fig:shelfNotation}.
We denote the last $h_\newshelf$-wide section at the end of \newshelf by
\headOf{\newshelf} and the last $h_\newshelf/2$-wide slice by \shelfEnd{\newshelf}.
The total area of the squares packed into a shelf \newshelf is \occupied{\newshelf}.
The part of the square $Q$ packed in the upper half of \newshelf 
is \extraOf{Q}.

A useful property of the shelf-packing algorithm is that \usedSection{\newshelf}
has a packing-density of $1/2$ if we pack $\newshelf$ with squares of the same
height class only. The gap remaining at the end of a closed shelf may vary depending
on the sequence of squares. However, the following density property described in
the following lemma (due to Moon and Moser~\cite{mm67}).
\begin{lemma}\label{lem:shelfend}
    Let \newshelf be an $H_k$-shelf with width $w_\newshelf$ and
    height $h_\newshelf$ that is packed by $\nextFitShelf{\newshelf}$ with a set \packedseqOf{\newshelf} of  $H_k$-squares.
    Let $Q$ be an additional square of $H_k$ with side length $x$ that does not fit into \newshelf.
    Then the total area $\norm{\packedseqOf{\newshelf}}$ of all squares packed into \newshelf plus the area $\norm{Q}$ of $Q$ is greater than $\norm{\newshelf}/2 - (h_\newshelf/2)^2 + \frac{1}{2}h_\newshelf\cdot x$.
\end{lemma}
In other words: If we count the extra area of
the overflowing square $Q$ towards the density of a closed shelf \newshelf,
we can, w.l.o.g., assume that \newshelf has a packing
density of $1/2$, except for at its end \shelfEnd{\newshelf}.
We formalize this charging scheme as follows.
When a square $Q$ causes a shelf \newshelf to be closed, we assign
\extraOf{Q} to \newshelf; see Fig.~\ref{fig:shelfEnd}.
The total area assigned to \newshelf this way is referred to as \assigned{\newshelf}.
Further, define \totalAreaOf{\newshelf} as \occupied{\newshelf} plus \assigned{\newshelf} minus \extraOf{Q} of all squares $Q$ in \newshelf.
   
  \begin{corollary} 
  \label{cor:PackingAreaOfShelf}
    Let \newshelf be a closed shelf packed by the shelf-packing algorithm.
    Then $\totalAreaOf{\newshelf} \geq \norm{\newshelf\setminus\shelfEnd{\newshelf}}/2$.
  \end{corollary}

	\begin{proof}
		If the packing of \packedseq intersects with \shelfEnd{\newshelf}, then
		\begin{equation*}
			\totalAreaOf{\newshelf} \geq \occupied{\newshelf} - \sum_{Q'\in\packedseqOf{\newshelf}}\extraOf{Q'}> h_\newshelf/2 \cdot (w_\newshelf-h_\newshelf/2).
		\end{equation*}		
		Otherwise, square $Q$ with side length $x$ caused shelf \newshelf to be closed and we have:
		\begin{align*}
			\totalAreaOf{\newshelf}&=  \occupied{\newshelf} - \sum_{Q'\in\packedseqOf{\newshelf}}\extraOf{Q'} + \extraOf{Q} \\
										 &= \qquad\quad \frac{h_\newshelf}{2} \cdot (w_\newshelf-x) \qquad\quad + x(x-\frac{h_\newshelf}{2})\\
										 &\geq \frac{h_\newshelf}{2} \cdot (w_\newshelf-x) + \frac{h_\newshelf}{2}(x-\frac{h_\newshelf}{2}) & \\
										& = \frac{h_\newshelf (w_\newshelf-\frac{h_\newshelf}{2})}{2}.&	
		\end{align*}
	\end{proof}
   
\subsection{The \textit{packSmall} Subroutine}\label{sec:packSmall}
As noted above, the presence of one large or few medium squares already assigns 
a majority of the required area, without causing too much fragmentation. Thus, the
critical question is how to deal with small squares in a way that leaves space for 
larger ones, but allows us to find extra space for a continuing sequence of small squares.

We describe an algorithm for packing any family of $H_k$-squares with $k\geq 2$ and total area up to $11/32$ in Sections~\ref{sec:packSmallAlgo_overview}~to~\ref{sec:packSmallAlgo_H4} and discuss the resulting packing density in Sections~\ref{sec:packSmallAnalysis_overview}~to~\ref{sec:packSmallAnalysis_sep4}.
In Section~\ref{sec:packSmallAlgo_mixed} we describe mixed packing of small squares and analyze the corresponding density in Sections~\ref{sec:packSmallAnalysis_mixed}~and~\ref{sec:packSmallAnalysis_additional}.

\subsubsection{The \textit{packSmall} Algorithm: Overview and Notation}\label{sec:packSmallAlgo_overview}
In the Recursive Shelf Algorithm we pack all small squares according to the
\emph{packSmall} subroutine, independent of the large and medium square
packings. The method is based on the Next Fit Shelf (NFS) packing scheme described above.
We first give a brief overview of the general distribution of the shelves and the order in which we allocate the shelves for the respective height classes.

\paragraph{Notation and Distribution of the Shelves.}
The general partition of the unit square we use is depicted in Fig.~\ref{fig:smallPartition1}.
The regions $M_1, \dots, M_4$ (in that order) act as shelves for height class $H_2$.
We call the union $M$ of the $M_i$ the \emph{main packing area};
this is the part of \unitsq that will definitely be packed with squares by our packSmall subroutine.
The other regions may stay empty, depending on the sequence of incoming small squares.
The regions $B_1, \dots, B_4$ provide shelves for $H_3$. We call the union $B$
of the $B_j$ the \emph{buffer region}.
In the region $A$ we reserve $H_k$-shelf space for every $k \geq 4$.
We call $A$ the \emph{initial buffer region}.
The ends $E_1$, $E_2$ and $E_3$ of the main packing regions $M_1$, $M_2$ and $M_3$ serve as both
parts of the main packing region and additional buffer areas.
We use $\Ebuffershelf_i$ to refer to the vertical section of $M_i$ that does not intersect with \usedSection{M_i}.
\begin{figure}[t]
  \centering
  \subfigure[]{
	  \def\svgwidth{.333\textwidth}
	  \import{./}{smallPartition.tex}
      \label{fig:smallPartition1}
      \hskip 2cm
  }
  \subfigure[]{
      \includegraphics[width = 0.275\textwidth]{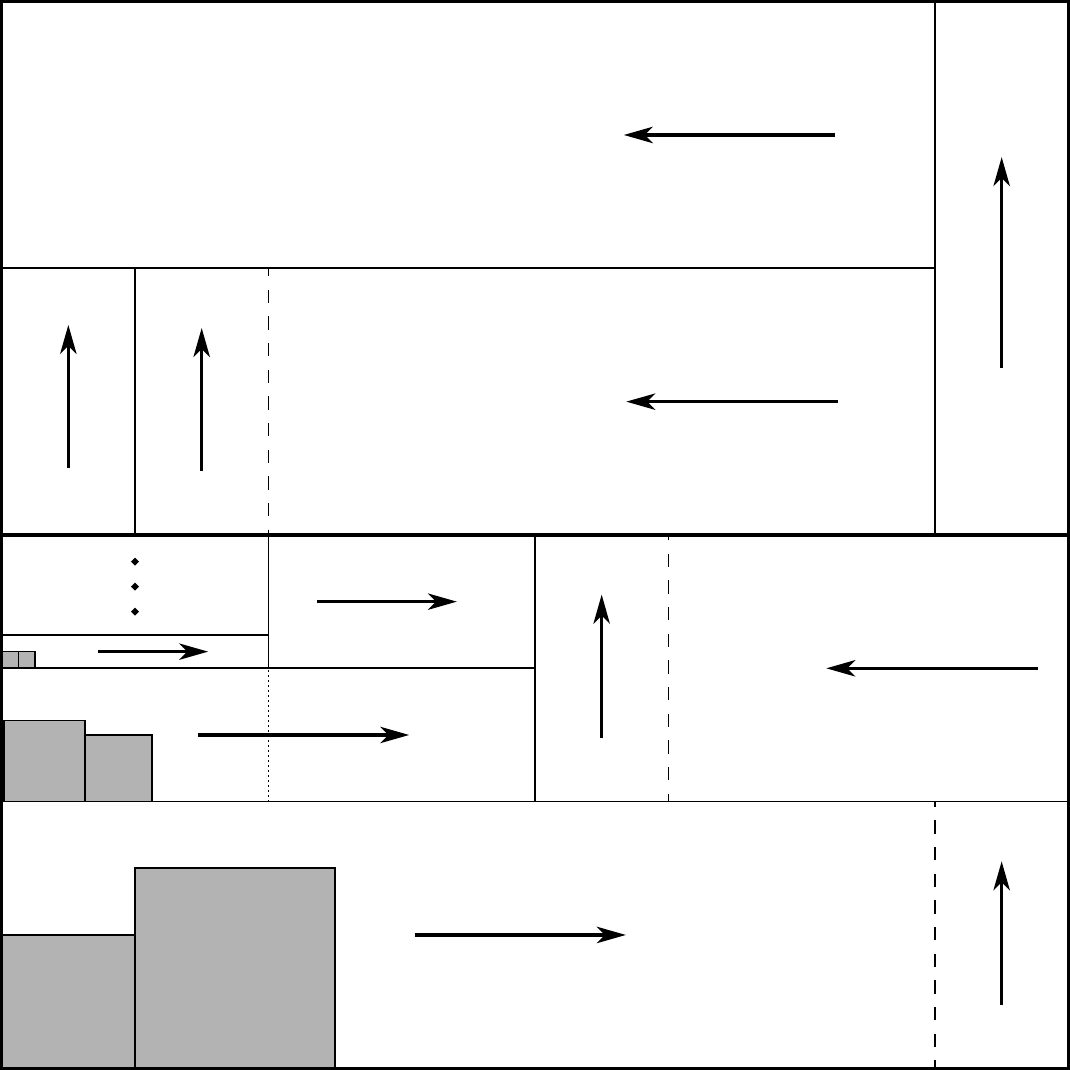}
      \label{fig:smallPartition2}
      \hskip 2mm
  }
\caption{
	(a) Distribution of the shelves for the \emph{smallPack} Algorithm.
	(b) Initital shelf packing and packing directions.
}
  \label{fig:smallPartition}
\end{figure}
\paragraph{Shelf Allocation Order.} During the packing process, we maintain open shelves for all the height classes
for which we already received at least one square as input and pack
each of them according to \emph{NFS}. The order and location for the shelf allocation are chosen as follows.
\begin{itemize}
	\item We start packing small squares into shelves that we open on the left side
		of the lower half \lowerhalf of \unitsq; see Fig.~\ref{fig:smallPartition2}.
		The region $M_1$ serves as the first $H_2$-shelf, the left half (width $1/4$)
		of $B_1$ serves as the first shelf for $H_3$ and region $A$ is reserved for
		first shelves for any $H_k$ with $k\geq 4$; see details below.
	\item Once an overflow occurs in a main packing region $M_i$, we close the
		corresponding $H_2$-shelf and continue packing $H_2$-squares into $M_{i+1}$. 
	\item Once the packing in the initial shelf for $H_{k}$ with $k\geq3$ reaches a certain length,
		we cut a vertical slice $\newvertical_{k}$ out of the currently open $H_2$-shelf (one of the $M_i$ regions)
		and use $\newvertical_{k}$ for the packing of subsequent $H_{k}$-squares.
	\item Once the packing in $\newvertical_{k}$ reaches a certain height,
		we allocate space in the buffer region $B\cup E$ to accommodate $H_{k}$-squares before returning to pack $\newvertical_{k}$.
	\item Once $\newvertical_{k}$ is full, we cut another vertical slice out of the main
		packing region and repeat the process.
\end{itemize}
We claim that we can accommodate any family of small squares with total area up to $11/32$ this way.
In the following, we describe the packings for the different small height
classes in more detail.

\subsubsection{The \textit{packSmall} Algorithm: Separate Packing of $H_2$-squares}\label{sec:packSmallAlgo_H2}
In the main packing area, we always maintain an open shelf $M_i$ for height class $H_2$,
which is packed with $H_2$-squares according to $\nextFitShelf{M_i}$.
In order to avoid early collisions with large and medium squares,
we start with packing $M_1$ from left to right, continuing
with packing $M_2$ from right to left.
Then we alternately treat $M_3$ and $M_4$ as the current main packing region,
placing $H_2$-squares into the region whose \textit{usedSection}
is smaller. When the length of \usedSection{M_4}
becomes larger than $3/8$,
we prefer $M_3$ over $M_4$ until $M_3$ is full.

\subsubsection{The \textit{packSmall} Algorithm: Separate Packing of $H_3$-squares}\label{sec:packSmallAlgo_H3}
For the packing of $H_3$-squares we alternate between using the buffer regions $B_1$, \dots, $B_4$, and vertical slices of width $1/8$ cut out of the main packing region as the currently open $H_3$-shelf; see details below and  Fig.~\ref{fig:h3_packing} for an example.

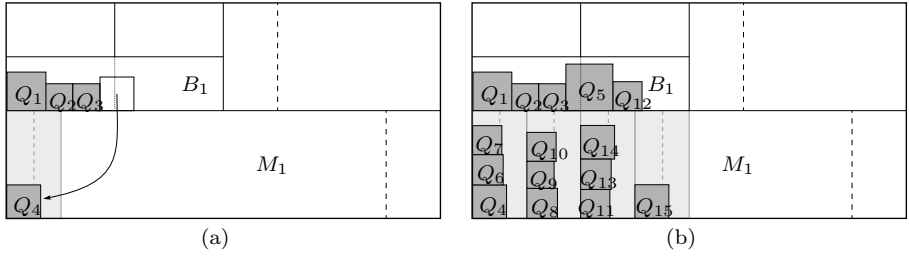
\begin{figure}[t]
  \centering
  \subfigure[]{
	  \def\svgwidth{.47\textwidth}
	  \import{./}{h3_packing_new1.tex}
      \label{fig:h3_packing1}
  }
  \subfigure[]{
	  \def\svgwidth{.47\textwidth}
	  \import{./}{h3_packing_new2.tex}
      \label{fig:h3_packing2}
  }
\caption{
	A sample packing of $H_3$-squares in the lower half of \unitsq.
	(a) Initial packing and first vertical shelf.
	(b) Packing after three iterations of step 2. and step 3.
}
  \label{fig:h3_packing}
\end{figure}
The algorithm uses variables $\mu$, $\beta$, $\eps_1$, $\eps_2$ and $\eps_3$, which are used to quantify the growth of the packings in regions $M$, $B$, $E_1$ ,$E_2$, and $E_3$, respectively. In the algorithm we use a comparison of $\mu$ and $\beta + \sum_i\eps_i$ to decide whether to place the next incoming square into the main packing region $M$ or the buffer region $B \cup E$. Intuitively, we do this to ensure approximately proportional growth of the two regions (see Lemma \ref{lem:invariant}), which in turn helps avoiding early collisions with large and medium squares.
In addition, we use $\newvertical_3$ to denote the (only) currently open vertical $H_3$-shelf in $M$. We define the algorithm \textit{packSmall(3)} for $H_3$ as follows.
\begin{enumerate}
	\item[0.] Set $\mu := 0$, $\beta := 0$, $\varepsilon_i := 0 \;\forall i$, $\newvertical_3 := \emptyset$.
	\item Open an $H_3$-shelf in $B_1$. Use \nextFitShelf{$B_1$} to pack incoming $H_3$-squares $Q$
		   and increase $\beta$ by $x_Q$ each time. 
		   Once $\beta +  x_Q \geq \mu + 1/4$, for the next incoming square $Q$, set $\beta := \beta + 1/16-x_Q$ and continue with step 2.
	\item Open a new vertical shelf $\newvertical$ of width $\frac{1}{8}$ and height $\frac{1}{4}$ at the end of the packing
			in $M$. Set $\newvertical_3 := \newvertical$.
			Use \nextFitShelf{$\newvertical_3$} (from bottom to top) to pack $H_3$-squares
			until the packing of the next square $Q$ in $\newvertical_3$ would
			intersect with \headOf{$\newvertical_3$}.
	\item  Increase $\mu$ by $1/16$ and:
			\begin{enumerate}
			\item  If $\beta + \sum_i\varepsilon_i + x_Q \geq \mu + 1/4$, pack $Q$ into $\newvertical_3$ and increase $\beta$ by $x_Q - 1/16$.
			\item Otherwise:
				\begin{enumerate}
					\item If there is an open end buffer shelf $\Ebuffershelf_i$ for which $M_i$ is closed, then 
					\begin{itemize}
				 		\item either pack $Q$ into $\newvertical_3$ and set $\varepsilon_i := \varepsilon_i + 1/16$ if $x_Q > 1/8 - \ell_i$ or
						\item use \nextFitShelf{$\Ebuffershelf_i$} to pack $Q$ and set $\varepsilon_i := \varepsilon_i + x_Q$, otherwise.
					\end{itemize}
					\item Otherwise, use \nextFitShelf{$B_1$} to pack $Q$ and increase $\beta$ by $x_Q$.
				\end{enumerate}
			\end{enumerate}
	\item Use \nextFitShelf{$\newvertical_3$} to pack all following $H_3$-squares until $\newvertical_3$ is full.
	\item Repeat Steps 2 to 4 using regions $M_1$, \dots, $M_4$ (in the same order and direction as for the $H_2$-square packing) for the placement of $\newvertical_3$ in Step~2 and regions $B_1$, \dots, $B_4$, $\newshelf_3$ (in this order) for the placement of $Q$ in Step~3(b)ii.
	If the algorithm closes region $M_i$, set $\varepsilon_i := 2\usLength{E_i}$.
If at any point in time $\varepsilon_i \geq 2/16$ or $\usLength{\Ebuffershelf_i}\geq 2/16$,
close $\Ebuffershelf_i$ and set $\varepsilon_i := \max\{\varepsilon_i,2/16\}$.
\end{enumerate}

\subsubsection{The \textit{packSmall} Algorithm: Separate Packing of $H_k$-squares with $k\geq4$}\label{sec:packSmallAlgo_H4}
For each $H_k$ with $k\geq4$, the packing algorithm \textit{packSmall(k)} is defined as follows.
\begin{enumerate}
	\item[0.] Set $\mu := 0$, $\beta := 0$, $\varepsilon_i := 0 \;\forall i$, $\newvertical_k := \emptyset$, $\newbuffer_k := \emptyset$.
	\item Open an $H_k$-shelf of length $1/4$ (and height $2^{-k}$) on top of the existing shelves in $A$. Call this shelf \initialbuffer{k}.
			Use \nextFitShelf{\initialbuffer{k}} (from left to right) to pack incoming $H_k$-squares until \initialbuffer{k} is 					full.
	\item Open a vertical shelf $\newvertical$ of width $2^{-k}$ and height $1/4$ at the end of the packing
			in $M$. 
			Set $\newvertical_k := \newvertical$  and use \nextFitShelf{$\newvertical_k$} (from bottom to top)
			to pack $H_k$-squares
			until the packing of the next square $Q$ in $\newvertical_k$ would
			intersect with \headOf{$\newvertical_k$}.
	\item  Increase $\mu$ by $2^{-k}/2$ and:
			\begin{enumerate}
				\item If $\beta + \sum_i\varepsilon_i \geq \mu + 3/16$, pack square $Q$ into $\newvertical_k$.
				\item Otherwise:
				\begin{enumerate}
					\item If there is an open end buffer shelf $\Ebuffershelf_i$ for which $M_i$ is closed,
						then open a horizontal $H_k$-shelf $\newendbuffer$ with width $1/8-\usLength{\Ebuffershelf_i}$ and
						height $2^{-k}$ on top of the current packing in $\Ebuffershelf_i$, set
						$\newendbuffer_k := \newendbuffer$ and $\varepsilon_i := \varepsilon_i + 2^{-k}/2$.
						Use \nextFitShelf{$\newendbuffer_k$} to pack incoming $H_k$-squares until $\newendbuffer_k$ is full.
					\item Otherwise, open a vertical $H_k$-shelf $\newbuffer_k$ (with height $1/8$ and width $2^{-k}$) at
  						the end of the current packing in $B_1$, set $\beta := \beta + 2^{-k}$ and 
						use \nextFitShelf{$\newbuffer_k$} to pack incoming $H_k$-squares until $\newbuffer_k$ is full.
				\end{enumerate}
			\end{enumerate}
	\item Use \nextFitShelf{$\newvertical_k$} to pack all following $H_k$-squares until $\newvertical_k$ is full.
	\item Repeat Steps 2 to 4 using regions $M_1$, \dots, $M_4$ (in the same order and direction as for the $H_2$-square packing) for the placement of $\newvertical_k$ in Step~2 and regions $B_1$, \dots, $B_4$, $\newshelf_4$ (in this order) for the placement of $\newbuffer_k$ in Step~3(b)ii.
	If the algorithm closes region $M_i$, set $\varepsilon_i := 2\usLength{E_i}$.
If at any point in time $\varepsilon_i \geq 2/16$ or $\usLength{\Ebuffershelf_i}\geq 2/16$,
close $\Ebuffershelf_i$ and set $\varepsilon_i := \max\{\varepsilon_i,2/16\}$.
\end{enumerate}

\subsubsection{\textit{packSmall} Analysis: Overview}\label{sec:packSmallAnalysis_overview}
In following sections we prove that the \emph{packSmall} subroutine successfully packs
any set of \emph{small} squares with total area at most $11/32$.

In order to quantify the overall density achieved by the
\emph{packSmall} Algorithm, we make some simplifying assumptions on the density
reached in the respective shelves. We argue that low-density shelves only appear
along with high-density regions and define a charging scheme that assigns extra
areas from dense regions to sparse regions in order to estimate the overall density.
More precisely, we prove the following important invariant for our algorithm, which is essential for 
the overall density analysis in the case of mixed packings; see Section~\ref{sec:packSmallAnalysis_mixed}. 

\begin{property}\label{prop:mainPackingDensity}
  In any step of the algorithm, the total area of the small squares packed
  into \unitsq is at least $\norm{\usedSection{\mainPackingArea}\setminus E}/2$.
\end{property}

We start the density analysis by introducing notation, simplifying assumptions and general packing properties 
in Section~\ref{sec:packSmallAnalysis_general}.
We proceed with analyzing the case of separately packing a set of only $H_k$-squares. 
In Sections~\ref{sec:packSmallAnalysis_sep2}, Sections~\ref{sec:packSmallAnalysis_sep3}~and~\ref{sec:packSmallAnalysis_sep4}, we discuss the cases $k=2$, $k = 3$ and $k\geq 4$, respectively.
We describe and analyze the case of packing a mixed sequence of small squares in Sections~\ref{sec:packSmallAlgo_mixed} ~\ref{sec:packSmallAnalysis_mixed} and conclude with a presentation of additional density properties in Section~\ref{sec:packSmallAnalysis_additional}.

\subsubsection{\textit{packSmall} Analysis: Preliminaries}\label{sec:packSmallAnalysis_general}
For the analysis of the density achieved with the packing of small squares, we use the following notation:
\begin{align*}
	\vBase{k}{X} &:= \text{set of all vertical $H_k$-shelves in region $X$} \\
	\vOpen &:= \text{set of all open vertical shelves} \\
	\vClosed &:= \text{set of all closed vertical shelves} \\
V^{k}_\text{head} &:= \text{set of all $\mathcal{V} \in V^k_M$ for which we executed Step 3} \\
                & \text{in the \textit{smallPack}($k$) subroutine after opening $\mathcal{V}$} \\
V_{\text{head}} &:= \bigcup_{k\in K}{V^{k}_\text{head}}\\
	\vHeadB &:=  \text{set of all $\newvertical \in \vHead$ for which we placed a square in $B$ (Step 3(b)ii)} \\
	\vHeadE &:=  \text{set of all $\newvertical \in \vHead$ for which we placed a square in $E$ (Step 3(b)i)} \\
	V_{k\geq4} &:= \text{set of all $H_{k\geq4}$-shelves}
\end{align*}
\begin{align*}
	K &:= \text{set containing all $k$ for which } \vBase{k}{M} \neq \emptyset\\
	\beta, \varepsilon_i, \mu,\newvertical_k, \newbuffer_k &:= \text{variables used in the algorithm (see above)}\\
	e &:= \text{the index of the end buffer region $\Ebuffershelf_i$ that was closed last}\\
	\usLength{X} &:= \text{total length of \usedSection{X} (see Section~\ref{sec:generalShelfConcept})}\\
	\totalAreaOf{\newshelf} &:= \occupied{\newshelf} + \assigned{\newshelf} - \sum_{Q\in\newshelf}\extraOf{Q} \text{ for shelf \newshelf}
\end{align*}
To make similar simplifying density considerations as in Corollary~\ref{cor:PackingAreaOfShelf}, we define the following charging scheme that assigns area from high-density regions to low-density regions.
\begin{enumerate}
	\item[] \emph{Charging Scheme:}
	\item[I:] From each square $Q$ that causes an overflow in a shelf \newshelf assign \extraOf{Q} to \newshelf.
	\item[II:] From each $H_3$-square $Q$ that was packed into $\newvertical_3$ in Step~3(a), assign \extraOf{Q} to the buffer region $B_i$.
	\item[III:] From each $H_3$-square $Q$ that was packed into $\newvertical_3$ in Step~3(b)i, assign \extraOf{Q} to the buffer region $\Ebuffershelf_i$.
\end{enumerate}
In the following we use this charging scheme for the definition of \assigned{\newshelf} and
assume, w.l.o.g.\ that $\usLength{\newshelf} \geq w_{\newshelf\setminus\shelfEnd{\newshelf}}$ for any closed shelf \newshelf that is packed by \nextFitShelf{\newshelf}.
Because we only charge squares with their extra area and we do not charge any squares twice, we know that $\totalAreaOf{\newshelf}$ is a lower bound on the actual density of \newshelf.
In the remainder of this section, we prove some general packing properties, which we use in subsequent density considerations.
\begin{lemma}\label{lem:densityClosedVertical}
	Let $\newvertical$ be an $H_k$-shelf in \vClosed, then
$\totalAreaOf{\newvertical} \geq \norm{\newvertical}/2 - (w_\newvertical/2)^2$.
\end{lemma}
\begin{proof}
The claim follows directly with Corollary~\ref{cor:PackingAreaOfShelf} and the fact that each vertical $H_k$-shelf \newvertical is packed (vertically) by \nextFitShelf{\newshelf} with $H_k$-squares only. \qed
\end{proof}
\begin{lemma}\label{lem:densityOpenVertical}
	Let $\newvertical$ be an $H_k$-shelf in \vOpen, then
	$\totalAreaOf{\newvertical} \geq (w_\newvertical/2)^2$.
\end{lemma}
\begin{proof}
The claim follows directly from the fact that, by construction, each open vertical shelf contains at least one square of size at least $w_{\newvertical_o}/2$. \qed
\end{proof}
\begin{lemma}\label{lem:densityAllHkVertical}
	The total area $\sum\limits_{\newvertical\in\vBase{k}{M}} \totalAreaOf{\newvertical}$
	of vertical $H_k$-shelves in the main packing area $M$ is greater or equal to
	$\sum\limits_{\newvertical\in\vBase{k}{M}}\norm{\newvertical}/2
	- \sum\limits_{\newvertical\in\vBase{k}{M}\cap\vClosed}(w_\newvertical/2)^2
	- \frac{1}{4} \cdot 2^{-k}/2+ \totalAreaOf{\newvertical_k}$.
\end{lemma}
\begin{proof}
By construction, we always close the vertical $H_k$-shelf $\newvertical_k$ before opening a new one. Thus, $\newvertical_k$ is the only vertical $H_k$-shelf in $M$ that is open and the claim follows with Lemma~\ref{lem:densityClosedVertical} and the fact that $\norm{\newvertical_k} = 1/4 \cdot 2^{-k}$. \qed
\end{proof}

\begin{lemma}\label{lem:shelfLengthDensity}
For any $H_k$-shelf \newshelf packed by \nextFitShelf{\newshelf}, it holds $\totalAreaOf{\newshelf} \geq h_\newshelf \usLength{\newshelf}/2$. 
\end{lemma}
\begin{proof}
The claim follows directly from the fact that \usedSection{\newshelf} is packed with $H_k$-squares only, which all have size at least $h_\newshelf/2$.
\end{proof}

\begin{lemma}\label{lem:densityHead}
	For all $k\geq4$, we have 
\begin{equation*}
	\totalAreaOf{\newvertical_k} + \totalAreaOf{\newbuffer_k} \geq
		\begin{cases}
			\frac{1}{4}\cdot \frac{2^{-k}}{2} - (\frac{2^{-k}}{2})^2 &\mbox{if } \newvertical_k \in \vHead \text{ and } \newbuffer_k \in \vClosed\\
			\frac{1}{8}\cdot \frac{2^{-k}}{2} 							& \mbox{otherwise}
		\end{cases}
\end{equation*}
\end{lemma}
\begin{proof} 
If $\newbuffer_k$ is closed, then $\totalAreaOf{\newbuffer_k} \geq (1/8 - 2^{-k}/2)\cdot 2^{-k}/2)$ by Corollary~\ref{cor:PackingAreaOfShelf}. Otherwise, $\totalAreaOf{\newbuffer_k} \geq 2^{-k}/2$ because $\newbuffer_k$ contains at least one $H_k$-square.
We only execute Step~3 of the algorithm if the next square $Q$ would intersect with \headOf{$\newvertical_k$} when placed in $\newvertical_k$. Thus, $\totalAreaOf{\newbuffer_k} \geq (1/4 - 2^{-k} - 2^{-k} )\cdot 2^{-k}/2$ if $\newvertical \in \vHead$. If $\newvertical \notin \vHead$, then $\totalAreaOf{\newbuffer_k} \geq (\frac{2^{-k}}{2})^2$ and $\newbuffer_k$ must be closed.
\end{proof}
\begin{lemma}\label{lem:invariant}
After each step in the algorithm it holds $\beta + \sum_i\varepsilon_i \geq \mu + 3/16$.
\end{lemma}
\begin{proof}
To simplify the notation define $\varepsilon := \sum_i\varepsilon_i$. 
Initially, we have $\mu = \beta = \varepsilon = 0$. Now consider the execution of any step in the algorithm that could change any of these values.
\begin{description}
	\item[Step 1.:] The variable values only change if $k=3$.
						If $\beta^{\text{old}}+ x_Q \geq \mu^{\text{old}} + 1/4$, we have
								$\beta^{\text{new}} + \varepsilon^{\text{new}} =
								\beta^{\text{old}} + x_Q - 1/16 + \varepsilon^{\text{old}} 
								\geq \mu^{\text{old}} + 1/4 - x_Q + x_Q - 1/16 = \mu^{\text{new}} + 3/16$.
						Otherwise,  $\Delta\beta = x_Q$ and $\Delta\mu = \Delta\varepsilon = 0$.
	\item[Step 2.:] Nothing changes.
	\item[Step 3.:]
					Independent of $k$, we have $\mu^{\text{new}} = \mu^{old} + 2^{-k}/2$.
			\begin{description}
					\item[Subcase (a) \& $k = 3$:] This step is only executed if
							$\beta^{\text{old}} + \varepsilon^{\text{old}} + x_Q \geq \mu^{\text{new}} + 1/4$, which implies
								$\beta^{\text{new}}+ \varepsilon^{\text{new}}
								= \beta^{old} + x_Q - 1/16 + \varepsilon^{old}
								\geq \mu^{\text{new}} + 1/4  + x_Q - 1/16 - x_Q
								= \mu^{\text{new}} + 3/16$
					\item[Subcase (a) \& $k \geq 4$:]  We have $\Delta\beta = \Delta\varepsilon = 0$ and $\beta^{\text{old}} + \varepsilon^{\text{old}} \geq \mu^{\text{new}} + 1/4$.
					\item[Subcase (b) \& $k = 3$:] We either have $\Delta\beta = x_Q$ and $\Delta\varepsilon = 2^{-k}/2$ (subcase i), or $\Delta\beta = x_Q \geq 2^{-k}/2$ and $\Delta\varepsilon = 0$ (subcase ii).  
					\item[Subcase (b) \& $k \geq 4$:] We either have
								 $\Delta\beta = 0$ and $\Delta\varepsilon = 2^{-k}/2$ (subcase i), or 
								$\Delta\beta = 2^{-k}$ and $\Delta\varepsilon = 0$ (subcase ii).
			\end{description}
	\item[Step 4.:] Nothing changes.
	\item[Step 5.:] We only increase the left hand side of the equation.
\end{description}
The inequality holds in either of the cases and the claim follows by induction.\qed
\end{proof}
\begin{lemma}\label{lem:betaBRel}
In each step of the algorithm, we have $\assigned{B} \geq 1/16 \cdot( \beta - \usLength{B} )$.
\end{lemma}
\begin{proof}
We give a proof by induction over the changes of $\assigned{B}$, $\beta$ and $\usLength{B}$.
Initially, we have $\beta = \assigned{B} = \usLength{B}= 0$. If $k=3$, then $\Delta\beta = (x_Q-\frac{1}{16})$, $\Delta\usLength{B}= 0$ and $\Delta\assigned{B} = x_Q(x_Q-\frac{1}{16}) \geq \frac{1}{16}(\Delta\beta-\Delta\usLength{B})$  in Step 3(a) and $\Delta\assigned{B} = 0$ and $\Delta\beta = \Delta\usLength{B} = x_Q$ in Steps 1 and 3(b). If $k\geq4$, then $\Delta\beta = \Delta\usLength{B} = 2^{-k}$ in Step 3(b). In all other cases, $\Delta\assigned{B} \geq 0 = \Delta\beta = \Delta\usLength{B}$.\qed
\end{proof}
\begin{lemma}\label{lem:mu_and_head}
In each step of the algorithm: $\mu = \sum_{\newvertical\in\vHead}\frac{w_\newvertical}{2}$ and
$\vClosed \subseteq \vHead$.
\end{lemma}
\begin{proof}
By construction, we increase $\mu$ by $\frac{2^{-k}}{2} =\frac{w_\newvertical}{2}$ each time we execute Step~3 for the packing of $H_k$-squares.
In addition, we only close the currently open vertical main packing shelf $\newvertical_k$ (Step 4) after executing Step~3, which proves the claim.\qed
\end{proof}

\subsubsection{\textit{packSmall} Analysis: Overall Density of Separate $H_2$-Square Packing}\label{sec:packSmallAnalysis_sep2}
\begin{lemma}
The algorithm successfully packs any sequence of $H_2$-squares with total area at most $11/32$.
\end{lemma}
\begin{proof}
Using the Next Fit Shelf Algorithm \nextFitShelf{M}, the packing explicitly allocates a position for each incoming $H_2$-square until on overflow occurs in $M_4$. In that case, we have $\norm{\packedseq} + \norm{Q} > 1/4 \cdot w_{M\setminus E} \cdot 1/2 = 1/8\cdot (7/8 + 3/8 + 5/8 + 7/8) = 22/64$ by Corollary~\ref{cor:PackingAreaOfShelf}, which contradicts $\norm{\packedseq} \leq 11/32$. Thus, the algorithm successfully packs all incoming $H_2$-squares.\qed
\end{proof}

\subsubsection{\textit{packSmall} Analysis: Overall Density of Separate $H_3$-Square Packing}\label{sec:packSmallAnalysis_sep3}
In this subsection we analyze the overall packing density for the special case of packing a sequence of squares that all belong to height class $H_3$.
\begin{lemma}\label{lem:densityH3}
	If the input sequence contains only $H_3$-squares, then
$\norm{\packedseq}
	\geq \frac{1}{8}\usLength{M\setminus E}$ after each step of the algorithm.
\end{lemma}
\begin{proof}
By construction, Section $M$ only contains vertical $H_3$-shelves. Thus, we have $\usLength{M} = \sum_{\newvertical\in\vBase{3}{M}}w_\newvertical$ and with Lemma~\ref{lem:densityAllHkVertical} we get
\begin{equation}\label{eq:densityM3}
\totalAreaOf{M} \geq \sum\limits_{\newvertical\in\vBase{3}{M}} \totalAreaOf{\newvertical}
\geq \frac{1}{4} \usLength{M}/2 - \sum\limits_{\newvertical\in\vBase{3}{M}\cap\vClosed}(w_\newvertical/2)^2
- \frac{1}{16}\cdot\frac{1}{4} + \totalAreaOf{\newvertical_3}
\end{equation}
With Lemmas~\ref{lem:shelfLengthDensity}~and~\ref{lem:betaBRel}, we get
\begin{equation}\label{eq:densityB3}
\totalAreaOf{B} \geq \assigned{B} + \occupied{B} - \extraOf{B} \geq \frac{1}{16} \cdot (\beta - \usLength{B}) + \frac{1}{16} \usLength{B} = \frac{1}{16}\beta
\end{equation}
By construction, we have
$\varepsilon_i = 2\usLength{E_i} \; \forall i \leq e$ and
$w_\newvertical/2 = 1/16$
Thus, by combining Equations~\ref{eq:densityM3}~and~\ref{eq:densityB3} and applying Lemma~\ref{lem:invariant}, we get
\begin{flalign*}
\norm{\packedseq} \geq \totalAreaOf{M} + \totalAreaOf{B}
&\geq \frac{1}{8} \usLength{M} + \frac{1}{16} \left(\beta - \frac{3}{16}\right) - \sum\limits_{\newvertical\in\vBase{3}{M}\cap\vClosed}(\frac{w_\newvertical}{2})^2 - \frac{1}{16^2} + \totalAreaOf{\newvertical_3}\\
&\geq \frac{1}{8} \usLength{M\setminus E} + \frac{1}{16} \left(\mu - \sum_{\newvertical\in\vBase{3}{M}\cap\vClosed}\frac{w_\newvertical}{2} \right ) - \frac{1}{16^2} + \totalAreaOf{\newvertical_3}
\end{flalign*}
The claim follows with Lemmas~\ref{lem:mu_and_head}~and~\ref{lem:densityH3}.
\qed
\end{proof}
\begin{lemma}\label{lem:packing_sep3}
The algorithm successfully packs any sequence of $H_3$-squares with total area at most $11/32$.
\end{lemma}
\begin{proof}
The algorithm explicitly assigns an unoccupied space to the next incoming square or vertical shelf until an overflow occurs in $M_4$ or $B_4$.
If an overflow occurred in $M_4$, we would have $\usLength{M_4} = w_{M_4}$ and $M_1$, $M_2$ and $M_3$ are closed.
Thus, $\norm{\packedseq} + \norm{Q} > 1/8 \cdot w_{M\setminus E} = 22/64$ by Lemma~\ref{lem:densityH3}, which contradicts $\norm{\packedseq}\leq11/32$.
Assume we could not fit a square $Q$ into $B_4$ in step 3(b)ii of the algorithm, then $\usLength{B_4} + x_Q > 1/4$ and $B_1$, $B_2$ and $B_3$ are closed. Thus, $\beta + x_Q \geq \sum_i\usLength{B_i} > (7+3+7+4)/16 = 21/16$. However, we only execute Step 3(b)ii if $\beta + x_Q < \mu - \varepsilon + 1/4 = \usLength{M\setminus E}/2 + \usLength{E}/2 - \sum_{i=1}^e 2\usLength{E_i} + 1/4 < 20/16$, which is a contradiction.
Hence, the algorithm successfully packs any sequence of $H_2$-squares.\qed
\end{proof}

\subsubsection{\textit{packSmall} Analysis: Overall Density of Separate $H_{k\geq4}$-Square Packing}\label{sec:packSmallAnalysis_sep4}

In this subsection we analyze the overall packing density for the special case of packing a sequence of squares that all belong to height class $H_k$ for a fixed $k\geq 4$.

\begin{lemma}\label{lem:densityH4}
	If the input sequence contains only $H_k$-squares with $k\geq4$, then
$\norm{\packedseq}
	\geq \frac{1}{8} \usLength{M\setminus E}$
	after each step of the algorithm.
\end{lemma}
\begin{proof}
By the same reasoning as for Equation~\ref{eq:densityM4} in Lemma~\ref{lem:densityH3} we have
\begin{equation}\label{eq:densityM4}
\totalAreaOf{M} \geq \sum\limits_{\newvertical\in\vBase{k}{M}} \totalAreaOf{\newvertical}
\geq \frac{1}{4} \usLength{M}/2 - \sum\limits_{\newvertical\in\vBase{k}{M}\cap\vClosed}(w_\newvertical/2)^2
- \frac{1}{4}\cdot\frac{2^{-k}}{2} + \totalAreaOf{\newvertical_k}
\end{equation}
By construction, we have
\begin{equation}\label{eq:vertBkEq}
|\vHeadB\cap \vBase{k}{M}| = |\vBase{k}{B}|\text{, and }
w_\newvertical = w_\newbuffer \; \forall \newvertical \in \vHeadB\cap\vBase{k}{M},\, \newbuffer \in \vBase{k}{B}\text{.}
\end{equation}
Because we maintain at most one open vertical $H_k$-buffer-shelf ($\newbuffer_k$) at all times, the following equation follows 
with Lemma~\ref{lem:densityClosedVertical} and Equation~\ref{eq:vertBkEq}.
\begin{align}\label{eq:densityVertB4}
\sum\limits_{\newbuffer\in\vBase{k}{B}} \totalAreaOf{\newbuffer} 
&\geq \totalAreaOf{\newbuffer_k} + \sum\limits_{\newbuffer\in\vBase{k}{B}\setminus\{\newbuffer_k\}}\norm{\newbuffer}/2 - (w_\newbuffer/2)^2  \\
&\geq \totalAreaOf{\newbuffer_k}  - \frac{1}{16} \cdot \frac{2^{-k}}{2} - (\frac{2^{-k}}{2})^2 + \sum\limits_{\newvertical\in\vHeadB}(\frac{1}{16} \cdot (w_\newvertical/2) + (w_\newvertical/2)^2)\nonumber
\end{align}
Because Section $B$ only contains vertical $H_k$-shelves and $\beta = \usLength{B} = \sum_{\newvertical\in\vHeadB}w_\newvertical$:
\begin{equation}\label{eq:densityB4}
\totalAreaOf{B}
\geq \totalAreaOf{\newbuffer_k} - \frac{1}{16} \cdot \frac{2^{-k}}{2} - (\frac{2^{-k}}{2})^2
+ \frac{1}{16} \cdot \beta/2 + \sum\limits_{\newvertical\in\vHeadB}(w_\newvertical/2)^2
\end{equation}
Section $A$ contains exactly one horizontal $H_k$-shelf $\newinitial_k$, which is closed before packing $H_k$-squares into $M$. Thus, if $M$ contains at least one vertical $H_k$-shelf, we have
\begin{equation}\label{eq:densityA4}
	\totalAreaOf{A} = \totalAreaOf{\newinitial_k} \geq (1/4 - 2^{-k}/2) \cdot 2^{-k}/2.
\end{equation}
By combining Equations~\ref{eq:densityM4},~\ref{eq:densityB4}~and~\ref{eq:densityA4} and applying
Lemma~\ref{lem:invariant}, we get
\begin{flalign*}
\norm{\packedseq} \geq \totalAreaOf{M} &+ \totalAreaOf{B}+ \totalAreaOf{A}\\
	\geq	\;\frac{1}{8} \usLength{M}
			- &\sum\limits_{\newvertical\in\vBase{k}{M}\cap\vClosed}(w_\newvertical/2)^2
				+ \sum\limits_{\newvertical\in\vHeadB}(w_\newvertical/2)^2
			+ \totalAreaOf{\newbuffer_k} + \totalAreaOf{\newvertical_k} \\
			+ &\frac{1}{16}\left(\beta - \sum\limits_{\newvertical\in\vHeadB\cap\vBase{k}{M}}(w_\newvertical/2)\right)
			- \frac{1}{16} \cdot \frac{2^{-k}}{2} - 2(\frac{2^{-k}}{2})^2 \\
	\geq	\;\frac{1}{8} \usLength{M\setminus E}
			+ &\frac{1}{16} \left(\mu - \sum\limits_{\newvertical\in\vBase{k}{M}\cap\vClosed}\frac{w_\newvertical}{2}\right)
			  + \totalAreaOf{\newbuffer_k} + \totalAreaOf{\newvertical_k} - \frac{1}{16} \cdot \frac{2^{-k}}{2} - 2(\frac{2^{-k}}{2})^2 
\end{flalign*}
The claim follows with Lemmas~\ref{lem:densityHead}~and~\ref{lem:mu_and_head}.
\qed
\end{proof}

\begin{lemma}\label{lem:packing_sep4}
The algorithm successfully packs any sequence of $H_k$-squares with $k\geq 4$ and total area at most $11/32$.
\end{lemma}
\begin{proof}
By the same reasoning as in the proof of Lemma~\ref{lem:packing_sep3}, no overflow occurs in $M$ as long as $\norm{\seq} \leq 11/32$.
Assume we could not fit a vertical $H_k$-shelf into $B_4$, then $\usLength{B_4} + 2^{-k} > 1/4$ and $B_1$, $B_2$ and $B_3$ are closed. Thus, $\beta \geq \sum_i\usLength{B_i} > (7+3+7+4)/16 - 2^{-k} = 21.5/16$. However, for $k\geq 4$, we only execute Step 3(b)ii if $\beta < \mu - \varepsilon + 3/16 = \usLength{M\setminus E}/2 + \usLength{E}/2 - \sum_{i=1}^e 2\usLength{E_i} + 1/4 < 20/16$, which is a contradiction.
Hence, the algorithm successfully packs any sequence of $H_k$-squares.\qed
\end{proof}

\subsubsection{The \textit{packSmall} Algorithm: Mixed Packing of Small Squares}\label{sec:packSmallAlgo_mixed}
In this section we describe the packing created by the \textit{packSmall} Algorithm for the case that the input sequence contains a mixed set of \emph{small} squares.

When receiving squares of different small height classes, not much changes.
We allocate shelves and fill them by placing squares (or vertical subshelves) at the end of their \textit{used sections} according to the \textit{packSmall(k)} algorithm given in Sections~\ref{sec:packSmallAlgo_H2}~to~\ref{sec:packSmallAlgo_H4} for each class separately.
Once we receive a first square of height class $H_k$, we simply start running \textit{packSmall(k)} in parallel to the other  \textit{packSmall} subroutines. For all subsequent $H_k$-squares $Q$ in the input, we simply perform the next step in 
 \textit{packSmall(k)} to pack $Q$.
The variables $\mu$, $\beta$ and $\epsilon_i$ become shared variables. They are initialized once to $0$ in a global Step~0 and are then modified by each of the different subroutines as described above (Steps 1 to 5).
The resulting packing differs from the separate packings in the following two ways.
\begin{itemize}
	\item The main packing regions $M$ and buffer regions $B \cup \Ebuffershelf$ may now contain vertical shelves from a variety of height classes; see Fig.~\ref{fig:mixedShelves}.
	\item Because the vertical shelves for $H_{k\geq3}$ do not fit as nicely into $M$ as is the case in the separate packings, gaps may remain at the end of the packing in each $M_i$. The algorithm uses these gaps for the placement of buffer squares and horizontal buffer shelves; see Fig.~\ref{fig:endBufferPacking} and Step 3(b)i of the algorithm.
\end{itemize}
\begin{figure}[t]
  \centering
  \includegraphics[width = 0.6\textwidth]{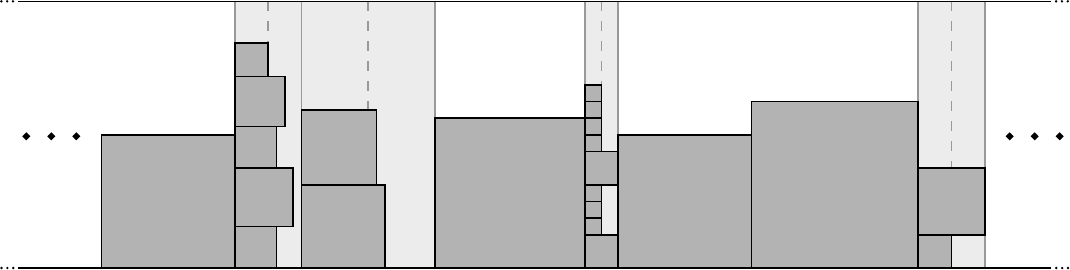}
  \caption{Sample packing of squares and vertical subshelves for mixed small height classes.}
  \label{fig:mixedShelves}
\end{figure}

 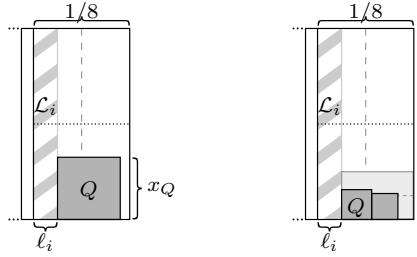
\begin{figure}[tp]
	  \centering
	      \def\svgwidth{0.2\textwidth}
	      \import{./}{end1.tex}
	      \hspace*{1cm}
	      \def\svgwidth{0.2\textwidth}
	      \import{./}{end3.tex}
	\caption{
		The buffer packing performed in the end buffer regions:
		(left) packing of a fitting $H_3$-square and
		(right) subshelf packing of $H_{\geq4}$ squares.
	}
	\label{fig:endBufferPacking}
\end{figure}

\subsubsection{\textit{packSmall} Analysis: Density of Mixed Small Square Packing}\label{sec:packSmallAnalysis_mixed}
In this section we analyze the overall density achieved by the \emph{packSmall} algorithm for any input sequence of small squares. 

\begin{lemma}\label{lem:density_small_mixed}
After each step of the packSmall Algorithm we have
$\norm{\packedseq_s} \geq \frac{1}{8}\usLength{M\setminus E}
+ \frac{1}{16} \cdot \sum_{\newvertical\in\vOpen\cap\vHead}\frac{w_\newvertical}{2}
+ \sum_{k\in K} ( \totalAreaOf{\newvertical_k} - (\frac{2^{-k}}{2})^2 )
+ \sum_{k\in K\setminus{3}} (\totalAreaOf{\newbuffer_k} - (\frac{1}{8} - \frac{2^{-k}}{2})\frac{2^{-k}}{2} )$
\end{lemma}
\begin{proof}
By construction, $M$ only contains $H_2$ squares and vertical $H_{k\geq3}$-shelves.
Thus, $\totalAreaOf{M} \geq \frac{1}{8}\usLength{M} - \sum_{\newvertical\in\vBase{}{M}} \norm{\newvertical} + \sum_{\newvertical\in\vBase{}{M}} \totalAreaOf{\newvertical}$ and with Lemma~\ref{lem:densityAllHkVertical} we get
\begin{equation*}\label{eq:densityMk}
\totalAreaOf{M}
\geq \frac{1}{8} \usLength{M} - \sum\limits_{\newvertical\in\vBase{}{M}\cap\vClosed}(w_\newvertical/2)^2
+\sum_{k\in K}\left(\totalAreaOf{\newvertical_k} - \frac{1}{4}\cdot\frac{2^{-k}}{2}\right)
\end{equation*}
Analogously, we have $\totalAreaOf{B} \geq \frac{1}{16}\usLength{B} - \sum_{\newbuffer\in\vBase{}{B}} \norm{\newbuffer} + \sum_{\newbuffer\in\vBase{}{B}} \totalAreaOf{\newbuffer}$, which together with Lemma~\ref{lem:betaBRel},~Equations~\ref{eq:vertBkEq}~and~\ref{eq:densityVertB4} and $w_\newvertical/2 = 1/16$ for $\newvertical\in\vBase{3}{M}$ implies
\begin{equation*}\label{eq:densityBk}
\totalAreaOf{B}
\geq \frac{1}{16} \left(\beta - \sum_{\newvertical\in\vHeadB}\frac{w_\newvertical}{2}\right) + \sum_{\newvertical\in\vHeadB}(\frac{w_\newvertical}{2})^2
+ \sum_{k\in K\setminus{3}} (\totalAreaOf{\newbuffer_k} - (\frac{1}{8} - \frac{2^{-k}}{2})\frac{2^{-k}}{2} )
\end{equation*}
Let $\bar\varepsilon_i$ be the amount assigned to $\varepsilon_i$ in the first subcase of 3(b)i in \textit{smallPack(3)}.
By construction, we have $\assigned{\Ebuffershelf_i} = (1/16 - \usLength{E_i} )\bar\varepsilon_i$ and 
	\begin{equation}\label{eq:epsValue}
		\varepsilon_i = 
		\begin{cases}
			\bar\varepsilon_i  + \usLength{\Ebuffershelf_i} - \sum_{\newvertical\in\vHeadEi\cap V_{k\geq4}}w_\newvertical/2 &\text{ if $\Ebuffershelf_i$ is open}\\
			\max\{2/16,2/16\usLength{E_i}\} &\text{ if $\Ebuffershelf_i$ is closed}
		\end{cases}
	\end{equation}
If $\usLength{E_i} \geq 1/16$, we have $\vHeadEi = \emptyset$ and $\varepsilon_i := 2\usLength{E_i}$.
Otherwise, $\totalAreaOf{\Ebuffershelf_i} \geq \frac{1}{16}\usLength{\Ebuffershelf_i} - \sum_{\newendbuffer\in\vBase{}{\Ebuffershelf}} \norm{\newendbuffer} + \sum_{\newendbuffer\in\vBase{}{\Ebuffershelf_i}} \totalAreaOf{\newendbuffer}$ and with the same reasoning as for Equation~\ref{eq:densityVertB4} we get
\begin{equation*}\label{eq:densityEk}
\occupied{\Ebuffershelf_i} \geq
\frac{1}{16} (\usLength{\Ebuffershelf_i} - \sum\limits_{\newvertical\in\vHeadEi}\frac{w_\newvertical}{2})
+ \sum\limits_{\newvertical \in \vHeadEi}(\frac{w_\newvertical}{2})^2
- \usLength{E_i}\sum\limits_{\newvertical\in\vHeadEi\cap V_{k\geq4}}\frac{w_\newvertical}{2}
\end{equation*}
By combining the above equations with $\bar\varepsilon_i \leq 2/16$ and $\usLength{E_i} < 1/16$, we get
\begin{equation*}\label{eq:densityEk+}
\totalAreaOf{\Ebuffershelf_i} \geq
\frac{1}{16} (\varepsilon_i - 2\usLength{E_i} - \sum\limits_{\newvertical\in\vHeadEi\cap V_{k\geq4}}\frac{w_\newvertical}{2})+ \sum\limits_{\newvertical \in \vHeadEi}(\frac{w_\newvertical}{2})^2
\end{equation*}
By the same reasoning as for Equation~\ref{eq:densityA4}, we get
\begin{equation*}\label{lem:densityAk}
\totalAreaOf{A} \geq \sum\limits_{k\in K} \totalAreaOf{\newinitial_k} \geq \sum\limits_{k\in K} (1/4 - 2^{-k}/2) \cdot 2^{-k}/2
\end{equation*}
The claim follows with $\norm{\packedseq_s} \geq \totalAreaOf{M} + \totalAreaOf{B} + \totalAreaOf{\Ebuffershelf} + \totalAreaOf{A}$ and Lemmas~\ref{lem:invariant}~and~\ref{lem:mu_and_head}. \qed
\end{proof}

\begin{theorem}\label{thm:packSmall}
    The \textit{packSmall} Algorithm packs any sequence of small
    squares with total area at most $11/32$ into the unit square. 
\end{theorem}
\begin{proof}
Let $Q$ be the next incoming $H_k$-square.
We consider all possible cases in which the algorithm does not explicitly assign an unoccupied space to $Q$.
\begin{enumerate}
\item Assume $Q$ causes an overflow in $M_4$. Then either $k=2$ and $\norm{Q} > 1/8w_Q$ or $k\geq 3$, $\usLength{M\setminus E} + 2^{-k} > 22/16$ and $\totalAreaOf{\newvertical_k} + \norm{Q} > \norm{\newvertical_k}/2 = 1/8 \cdot 2^{-k}$ by construction.
	\item Assume the algorithm cannot open a new vertical buffer shelf  \newbuffer for $H_k$ with $k\geq5$ in $B$. Then $\beta + 2^{-k} > 21/16$ and $\usLength{M\setminus E}/2 = \mu - \usLength{E}/2 > \beta + \varepsilon - \usLength{E}/2 - 3/16 > 21/16 - 2^{-k} + 4.5/16 - 3/16 = 22.5/16 - w_{\newvertical_k}$.
\item Assume $Q$ causes an overflow in $\newvertical_3$ or $\newvertical_4$. Then, by construction, $\newvertical_k \in \vOpen \cap \vHead$, $\totalAreaOf{\newvertical_k} + \norm{Q} > 1/8 \cdot w_{\newvertical_k}\geq 1/8 \cdot 0.5/16$, all $B_i$ regions must have been closed and we have $\beta \geq \sum_i\usLength{B_i} \geq 20/16$. With Lemma~\ref{lem:invariant}, we have $\usLength{M\setminus E}/2 = \mu - \usLength{E}/2 > \beta + \varepsilon - \usLength{E}/2 - 1/4 > 21.5/16$.
\end{enumerate}
In either of the three cases we get $\norm{\seq} > 11/32$ with Lemmas~\ref{lem:densityOpenVertical},~\ref{lem:densityHead}~and~\ref{lem:density_small_mixed}, a contradiction.
Thus, the algorithm successfully packs any sequence of small squares.
\qed
\end{proof}

\subsubsection{\textit{packSmall} Analysis: Some Additional Properties.}\label{sec:packSmallAnalysis_additional}
Before we analyze the algorithms performance in the presence of large and medium squares,
we state a couple of important properties of the packing created with small squares.

Recall that we use variable $\beta$ to quantify the growth of the buffer packing.
    By construction,
    we can relate the length of the buffer region and the total area
    of the input as follows.
\begin{lemma}\label{lem:bufferGrowth}
  Let $Q$ be a small square with side length $x_Q$ in the buffer region $B$
 and let $\packedseq_s$ be the set
  of small squares received so far. Then the total area of the small
  input squares $\norm{\packedseq_s}$ is greater than
$(\beta + \frac{1.5}{16}e + x_Q - 1/16) \cdot 1/4$.
\end{lemma}
    \begin{proof} 
      By construction, we only pack an $H_3$-square into $B$ if
      $\beta + 2\usLength{E} + \varepsilon + x_Q < \mu + 1/4$.
      For $Q\in H_k$ with $k\geq4$ we have $x_Q < 1/16$ and we only extend the buffer packing if
	$\beta +  2\usLength{E} + \varepsilon < \mu + 3/16$.
	In either case we get $\beta + 2\usLength{E} + \varepsilon + x_Q - 1/16 < \mu$.
     Recall that $\mu$ is defined as the total width of the vertical shelves in $M$.
     Thus, with Lemmas~\ref{lem:invariant},~\ref{lem:mu_and_head},
~and~\ref{lem:density_small_mixed} and Equation~\ref{eq:epsValue} we get
     \begin{eqnarray*}
	  \norm{\packedseq_s} \geq 1/8 \usLength{M\setminus E} &\geq& ( \mu - \sum_{\newvertical\in \vBase{}{E}} w_\newvertical /2 ) \cdot 1/4\\
				      & > & (\beta + \varepsilon+ x_Q - 1/16  - \usLength{E}/2 ) \cdot 1/4
     \end{eqnarray*}
     \qed
     \end{proof}
    As a direct implication of Lemma~\ref{lem:bufferGrowth} and the fact that
    when $B_4$ is first used for buffer square placement in step 3(b),
     both end buffer regions $\bar E_1$ and $\bar E_2$ have
     successfully been closed by the algorithm before,
    we get the following
    lower bounds for the total area $\norm{\packedseq_s}$ of small squares packed,
    as a function of the total length of the packing in $B$. 
    \begin{property}\label{prop:squareInB2}
      Let $Q$ be a small square in the buffer region $B_2$ with side
      length $x$ and distance $d > 1/4$ to the left boundary of \unitsq,
      then $\norm{\packedseq_s} > (d + x - 1/16) \cdot 1/4$.
    \end{property}
    \begin{property}\label{prop:squareFirstInB3}
      If there is a small square in $B_3$, Then $\norm{\packedseq_s} > 7/64$.
    \end{property}
    \begin{property}\label{prop:squareNotFirstInB3}
      Let $Q$ be a small square in $B_3$ with side length $x$ that was
      packed in a distance $d>0$ to the bottom of $B_3$.
      Then $\norm{\packedseq_s} > (7.5/16 + d + x) \cdot 1/4$.
    \end{property}
    \begin{property}\label{prop:squareFirstInB4} 
      If there is a small square in $B_4$, then $\norm{\packedseq_s} > 17/64$.
    \end{property}
    \begin{property}\label{prop:squareNotFirstInB4}
      Let $Q$ be a small square in $B_4$ with side length $x$ and
      distance $d>0$ from the bottom of $B_4$.
      Then $\norm{\packedseq_s} > (1 + d + x) \cdot 1/4$.
    \end{property}
The following properties follow directly from the algorithm invariant of Property~\ref{prop:mainPackingDensity}.
   \begin{property}\label{prop:density_M_1}
	When the first small square is packed into $M_2$, then $\norm{\packedseq_s} \geq 7/64$.
    \end{property}
      \begin{property}\label{prop:density_lowerhalf}
	When the first small square is packed into $M_3$,
	then $\totalAreaOf{\lowerhalf} \geq 10/64$.
      \end{property}

\subsection{Combined Analysis}\label{sec:algAnalysis}
In the previous sections we proved that the algorithm successfully packs 
small, medium and large squares separately, as long as input has a total
area of at most $11/32$.
A case distinction over all possible collisions that may appear between the
packings of these height classes can be used to prove the main result.

\begin{theorem}\label{thm:11over32}
 The Recursive Shelf Algorithm packs any sequence of squares with total area
 at most $11/32$ into the unit square.
\end{theorem}
      We prove the claim by showing that if the algorithm fails to pack a square,
      the total area of the given squares must exceed $11/32$.
    In the following we analyze the packing density at the time a
    collision of the different packing subroutines would appear.
    First we consider a collision between a medium and a small square in the upper half \upperhalf
	of the unit square container.
      \begin{figure}[t]
	  \centering
	  \subfigure[Packing created with medium squares.]{
 	      \def\svgwidth{0.25\textwidth}
 	      \import{./}{upperhalfCollision1.tex}
	      \label{fig:upperHalfcol1}
	  }
	  \hskip 2cm
	  \subfigure[Packing created with small squares.]{
 	      \def\svgwidth{0.25\textwidth}
 	      \import{./}{upperhalfCollision2.tex}
	      \label{fig:upperHalfcol2}
	  }
	  \caption[Packing performed in the upper half of \unitsq.]
	      {Packing performed in the upper half of \unitsq.
	      The feasible packing area has light gray background,
	      the medium gray part represents the packing created before $Q$ was placed.}
	\vspace*{-3mm}
      \end{figure}
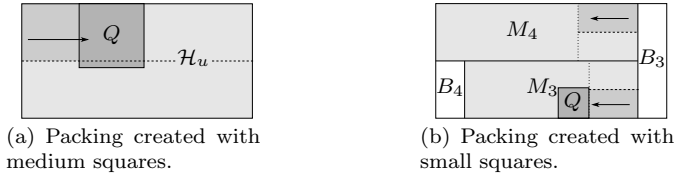
     
    \begin{lemma}\label{lem:coll_H1_small_upperhalf}
      If a medium square $Q_1$ collides with a part of the packing constructed
      with small squares, then $\norm{Q_1} + \totalAreaOf{\upperhalf} \geq 6/32$.
    \end{lemma}
    \begin{proof} 
      Recall that we pack the medium sized squares from left to right aligned with the
      top boundary of \upperhalf; see Fig.~\ref{fig:upperHalfcol1}.
      The packing of small squares (into $M_3$ and $M_4$) is performed from right to left;
      see Fig.~\ref{fig:upperHalfcol2}.
      Also recall that we alternatingly use $M_3$ and $M_4$ as the current main packing
      region
      (choosing which ever half is less full in width) until the packing in $M_4$ reaches
      a total length at least $3/8$.
      Then we only pack $M_3$ until it is completely filled, before finishing the packing
      in $M_4$.
      \begin{figure}[t]
	  \centering
	  \subfigure[]{
	      \def\svgwidth{0.3\textwidth}
	      \import{./}{upper1.tex}
	      \label{fig:h1_h2_M4_collision}
	  }
	  \subfigure[]{
	      \def\svgwidth{0.3\textwidth}
	      \import{./}{upper2.tex}
	      \label{fig:h1_h2_M3_collision1}
	  }
	  \subfigure[]{
		\def\svgwidth{0.3\textwidth}
		\import{./}{upper3.tex}
		\label{fig:h1_h2_M3_collision2}
	  }
	  \caption{Collision of a medium square $Q_1$ with an $H_2$-square $Q_2$ in
		  \upperhalf.
	  }
      \end{figure}
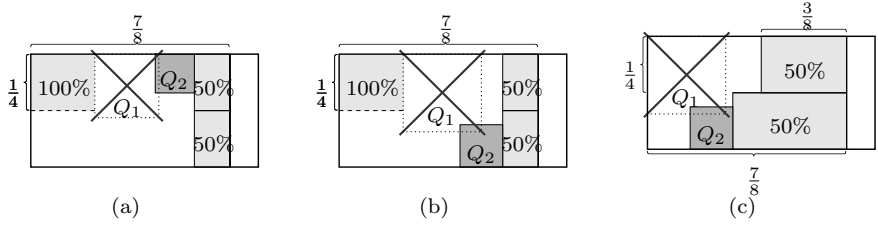

      Let $Q_1$ be a medium square that collides with a small square in
      the upper half \upperhalf of the unit bin.
      Then $Q_1$ either intersects a vertical shelf \newshelf or an $H_2$-square $Q_2$.
      The main idea is to prove that the parts of $\upperhalf\setminus B_3$ both right and
      left to $Q_2$/\newshelf have a density of $1/2$.
      We distinguish six different cases depending on the location of \newshelf or $Q_2$ in
      \upperhalf.
      \begin{enumerate}
 	\item{\em $Q_1$ collides with an $H_2$-square $Q_2$ in $M_4$:}\\
	    We know $\usLength{M_3} > \usLength{M_4} - w_\newshelf$,
	    as otherwise we would have packed $Q_2$ in $M_3$.
	    Therefore, the entire part of $M_3 \cup M_4$ to the right of $Q_2$ must be used by
	    small squares, thus having a density of $1/2$.
	    Additionally, the section used by the $H_1$-squares must be filled to a height of
	    at least $1/4$.
	    Hence, we know that the sections of $\upperhalf\setminus B_3$ both right and left
	    to $Q_2$ are half full; see Fig.~\ref{fig:h1_h2_M4_collision}.
	    Therefore, with $x_2\geq 1/8$,
	    \begin{eqnarray*}
		\totalAreaOf{\upperhalf} &>&	\frac{(7/8-x_2)\cdot 1/2}{2} + x_2^2 \\
					&\geq&	\frac{7}{32} + x_2\left(x_2-\frac{1}{4}\right) \\
					&\geq&	\frac{7}{32} +
			      \frac{1}{8}\cdot\left(\frac{1}{8}-\frac{1}{4}\right)
					> \frac{6}{32}.
	    \end{eqnarray*}
	\item{\em$Q_1$ collides with an $H_2$-square $Q_2$ in $M_3$:}\\
		By construction we have
		  $\usLength{M_3} - x_2 \leq \usLength{M_4}$
		  or $\usLength{M_4} \geq 3/8$,
		as we packed $Q_2$ into $M_3$ instead of $M_4$.
		\begin{enumerate}
		  \item If $\usLength{M_3} - x_2 \leq \usLength{M_4}$,
			the situation is symmetric to the one in the previous case;
			see Fig.~\ref{fig:h1_h2_M3_collision1}.
			Because $Q_1$ is aligned with the top and $Q_2$ with the bottom of
			\upperhalf and $Q_1$ and $Q_2$ collide, we have $x_1 + x_2 > 1/2$.
			Thus, we get
		    \begin{eqnarray*}
			\totalAreaOf{\upperhalf} &>& \frac{(7/8-x_1-x_2)\cdot 1/2}{2}
							    + x_1^2 + x_2^2 \\
						 &\geq& \frac{7}{32} - \frac{x_1 + x_2}{4} +
							    \frac{(x_1 + x_2)^2}{2} \\
						 &>& \frac{7}{32}
						      - \frac{(x_1 + x_2)\cdot 1/2}{2}
						      + \frac{(x_1+ x_2)\cdot 1/2}{2}
						      > \frac{6}{32}.
		    \end{eqnarray*}
	      \item Otherwise, $\usLength{M_3}-x_2 >
				 \usLength{M_4}\geq3/8$;
		    see Fig.~\ref{fig:h1_h2_M3_collision2}. 
		    Again, we know $x_1 + x_2 > 1/2$.
		    Thus,
		    \begin{eqnarray*}
		      \totalAreaOf{\upperhalf} &\geq& x_1^2 + x_2^2
						    + \frac{\norm{\usedSection{M_3}}}{2}
						    + \frac{\norm{\usedSection{M_4}}}{2}\\
					       &>& \frac{(x_1 + x_2)^2}{2}
						    + 2\cdot\frac{\norm{\usedSection{M_4}}}{2}\\
					       &>& \frac{(1/2)^2}{2}
						    + \frac{3}{8}\cdot\frac{1}{4} > \frac{6}{32}.
		    \end{eqnarray*}
	    \end{enumerate}
	  \begin{figure}[t]
	    \centering
	    \subfigure[]{
		  \def\svgwidth{0.3\textwidth}
		  \import{./}{upper2v.tex}
		  \label{fig:h1_shelf_M4_collision}
	    }
	    \subfigure[]{
		\def\svgwidth{0.3\textwidth}
		\import{./}{upper3v.tex}
		\label{fig:h1_shelf_M3_collision1}
	    }
	    \subfigure[]{
		  \def\svgwidth{0.3\textwidth}
		  \import{./}{upper1v.tex}
		  \label{fig:h1_shelf_M3_collision2}
	      }
	    \caption{Collision of a medium square $Q_1$ with vertical shelf \newshelf
		      in \upperhalf.
		    }
	    \vspace*{-3mm}
	  \end{figure}
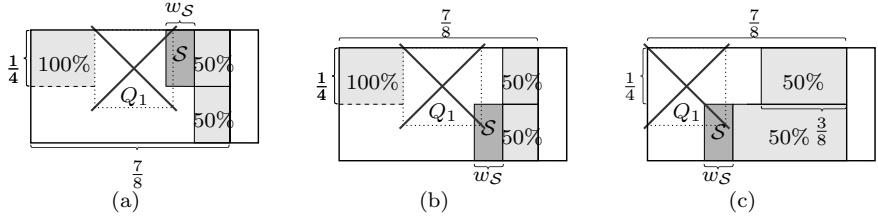
	\item{\em$Q_1$ collides with a vertical shelf \newshelf in $M_4$:}\\
	    Analogously to the first case, we know that the sections of $\upperhalf\setminus B_4$
	    both right and left to \newshelf must be half full;
	    see Fig.~\ref{fig:h1_shelf_M4_collision}.
	    Thus, with $w_\newshelf \leq 1/8$ we get:
	    \begin{eqnarray*}
		\totalAreaOf{\upperhalf} &>& \frac{(7/8-w_\newshelf)\cdot 1/2}{2}
						+ \frac{\norm{\newshelf}}{2}\\
					 &\geq&	\frac{7}{32} - \frac{w_\newshelf}{4}
						    \geq \frac{6}{32}.
	    \end{eqnarray*}
	\item{\em $Q_1$ collides with a vertical shelf \newshelf in $M_3$:}\\
	      Analogously to the second case, we must have 
	      $\usLength{M_3} - w_\newshelf \leq \usLength{M_4}$
	      or $\usLength{M_4} \geq 3/8$
	      as we opened \newshelf in $M_3$.
	    \begin{enumerate}
	      \item If $\usLength{M_3} - w_\newshelf \leq
			      \usLength{M_4}$,
		    then we have the same conditions as described in the first case;
		    see Fig.~\ref{fig:h1_shelf_M3_collision1}.
		    We analogously get
			$$\totalAreaOf{\upperhalf} >
\left(\frac{7}{8}-w_\newshelf\right)\frac{1}{4} 
						  \geq \frac{6}{32}.$$
	      \item Otherwise, $\usLength{M_3} - w_\newshelf
					    > \usLength{M_4} \geq 3/8$.
		    Let $\ell^1$ be the length of the \upperhalf-section used by $H_1$.
		    We know $\ell^1 \geq 1/4$ and $\usLength{M_3} > 7/8 - \ell^1$;
		    see Fig.~\ref{fig:h1_shelf_M3_collision2}.
		    Thus,
		    \begin{eqnarray*}
			\totalAreaOf{\upperhalf}	&\geq& \ell^1\cdot\frac{1}{4} +
							    \frac{\norm{\usedSection{M_3}}}{2} +
							    \frac{\norm{\usedSection{M_4}}}{2} \\
						    &>&	\ell^1\cdot\frac{1}{4} +
							    \frac{(7/8 - \ell^1)\cdot 1/4}{2} +
							    \frac{3/8\cdot 1/4}{2} \\
						    &\geq&	\frac{\ell^1}{8} + \frac{7}{64} +
							    \frac{3}{64} \geq \frac{6}{32}.
		    \end{eqnarray*}
		    \qed
	    \end{enumerate}
      \end{enumerate}
    \end{proof}
   We are now able to prove Theorem~\ref{thm:11over32}.
   \begin{proof}[of Theorem~\ref{thm:11over32}]
      Let $Q$ be the square at which the algorithm stops.
      Denote \seq the set of all input squares and \packedseq the set of all squares
      packed at the time $Q$ arrives. We claim $\norm{Q} + \norm{\packedseq} > 11/32$.
      To prove this statement we distinguish the different types of collisions that might
      cause the algorithm to stop with failure.
	Note that we covered the cases in which \seq consists of either all large, all medium or all small squares in the previous sections. In the following we denote $\packedseq_s$ the set of all small squares in \packedseq and  $\packedseq_m$ the set of all medium squares in \packedseq.
      \begin{enumerate}
	\item \emph{A large square $Q_0$ collides with a medium square $Q_1$:}\\
	      In this case, the first (and only) square of $H_0$ collides with the L-shaped packing
	      produced by the Ceiling Algorithm; see Fig.~\ref{fig:11_32_proof_h1_h0}.
	      We know $\norm{Q_0} > (1/2)^2 = 1/4$ and
	      the shelf packing for the $H_1$-squares must reach from the left boundary
	      to more than a distance of $x_0$ from the right boundary.
	      Thus, as $x_i \geq 1/4$ for any square $Q_i \in H_1$,
	      the total area of the input sequence $\norm{\seq}$ is at least
	      $\norm{\packedseq} + \norm{Q_0} >  x_0^2 + (1-x_0)\cdot\frac{1}{4} \geq 3/8 > 11/32$.
      \item \emph{A large square $Q_0$ collides with a small square $Q_s$:}\\
	      If the side length $x_0$ of $Q_0$ is greater than $\sqrt{11/32}$,
	      then $\norm{Q_0} > 11/32$ and we are done.
	      Therefore, we assume $x_0 \leq \sqrt{11/32} < 5/8$.
	      There are two cases:
	      \begin{enumerate}
	      \item \emph{$Q_s$ is in the main packing area:}\\
		    Because $x_0<5/8<3/4$, $Q_s$ must have been packed into $M_2$, $M_3$ or $M_4$.
		    In any case, a small square must be in $M_2$ and by Property~\ref{prop:density_M_1}
		    we have a total area of more than $7/64$ from small squares. Additionally,
		    we have $x_0 \geq 1/2$, as $Q_0$ is large. That is,
		    $\norm{\seq} \geq \norm{Q_0} + \norm{\packedseq_s}
			  > (1/2)^2 + 7/64 = 23/64 > 11/32$;
		    Fig.~\ref{fig:11_32_proof_h0_small1}.
	      \item \emph{$Q_s$ is in the buffer area:}\\
			We have that $Q_s$ is not in $B_1$ since $x_0 < 5/8$.
		    If $Q_s$ is located in $B_3$ or $B_4$, then $\norm{\packedseq_s} > 7/64$ according to Property~\ref{prop:squareFirstInB3} and $\norm{\seq} > (1/2)^2 + 7/64 > 11/32$.
		   Otherwise, $Q_s$ is located in $B_2$.
		    Let $d$ be the distance of $Q_s$ to the left boundary of the unit square.
		    As $Q_0$ and $Q_s$ collide, we have $d + x_s + x_0 > 1$;
		    see Fig.~\ref{fig:11_32_proof_h0_small2}.
		    We distinguish two cases for the side length of $Q_0$:
		    \begin{enumerate}
		      \item $x_0 \in (1/2,9/16)$:
			    Then $d + x_s > 1 - x_0 > 7/16$, which implies $d>7/16-1/8>1/4$.
			    Thus, by Property~\ref{prop:squareInB2} we get
			    $$\norm{\packedseq_s} > \frac{d + x_s - 1/16}{4}
					      > \frac{7/16 - 1/16}{4} = 6/64$$
		      \item $x_0 \in [9/16,5/8)$:
			    Then $d + x_s > 1 - x_0 > 3/8$, which implies $d>3/8-1/8=1/4$.
			    Thus, by Property~\ref{prop:squareInB2} we get
			    $$\norm{\packedseq_s} > \frac{d + x_s - 1/16}{4}
					      > \frac{3/8 - 1/16}{4} = 5/64$$
		    \end{enumerate}
		    In total we get\\
			$\norm{\seq} \geq \norm{Q_0} + \norm{\packedseq_s}
					> \min\{(1/2)^2 + 6/64, (9/16)^2 + 5/64\} = 11/32$.
	      \end{enumerate}
      \item \emph{A medium square $Q_1$ collides with a small square $Q_s$:}\\
	      There are many different types of collisions that might appear between
	      the small square packing and a square of $H_1$.
	      Note that the ceiling packing never interacts with the buffer area of \lowerhalf,
	      but might interact with the buffers in \upperhalf.
	      \begin{enumerate}
		  \item \emph{$Q_s$ is (a buffer square) in $B_3$:}\\
			Recall that all medium squares are packed from left to right
			aligned with the top boundary of \unitsq.
			Let $d$ be the distance of $Q_s$ to the lower boundary of $B_3$.

			We distinguish two cases for $x_s + d$:
			\begin{enumerate}
			  \item $x_s + d \leq 1/8$:
				Then $Q_1$ intersects the $1/8$-high section at the bottom 
				of $B_3$; see Fig.~\ref{11_32_proof_h1_b3_1}.
				That is, either an overflow of $H_1$-squares occured in
				\upperhalf and we have $\norm{\packedseq_m} > 1\cdot1/4$,
				or $Q_1$ coincides with the top of \unitsq, which implies $x_1 > 3/8$,
				and we have
				  $$\norm{\packedseq_m} > x_1^2 + (\frac{7}{8}-x_1)\cdot\frac{1}{4}
				      \geq \frac{7}{32} + x_1\cdot(x_1-\frac{1}{4})
				      \geq \frac{14}{64} + \frac{3}{8}\cdot\frac{1}{8}
				      = \frac{17}{64} > \frac{1}{4}.$$
				As $Q_s$ is in $B_3$ we get $\norm{\packedseq_s} > 7/64$
				by Property~\ref{prop:squareFirstInB3}.
				In both cases, we get
				    $$\norm{\seq} \geq \norm{\packedseq_m} + \norm{\packedseq_s}
						  > 1/4 + 7/64
						  = 11/32.$$
			  \item $x_s + d > 1/8$:
				This case is depicted in Fig.~\ref{11_32_proof_h1_b3_2}.
				By Property~\ref{prop:squareNotFirstInB3} we get
				    $$\norm{\packedseq_s} > \frac{7.5/16 + d + x_s}{4}
				      > \frac{7.5/16 + 2/16}{4} = \frac{9.5}{64}.$$
				Because $Q_1$ intersects with $B_3$, we know that the
				total length of the medium square packing is greater than
				$7/8$.
				Thus we get
				  $$\norm{\seq} \geq \norm{\packedseq_m} + \norm{\packedseq_s}
					     > 7/8 \cdot 1/4 + 9.5/64 > 23.5/64 > 11/32.$$    
			\end{enumerate}
		  \item \emph{$Q_s$ is (a buffer square) in $B_4$:}\\
			Recall that we start packing medium squares coinciding with
			the top of \unitsq.
			We fill the buffer region $B_4$ from bottom to top.
			Let $d$ be the distance of $Q_s$ to the lower boundary of $B_4$.
			We distinguish two cases for $x_s + d$:
			\begin{enumerate}
			  \item $x_s + d \leq 1/8$:
				Then $Q_1$ perturbs the $1/8$-high section at the bottom
				of $B_4$, i.e. we have $x_1 > 3/8$.
				Because $Q_s$ is in $B_4$, we get $\norm{\packedseq_s} > 17/64$
				by Property~\ref{prop:squareFirstInB4}.
				Thus, in total we have
				    $$\norm{\seq} \geq \norm{Q_1} + \norm{\packedseq_s}
					      > (3/8)^2 + 17/64 = 26/64 > 11/32.$$
			  \item $x_s + d > 1/8$:
				By Property~\ref{prop:squareNotFirstInB4} we have
				    $$\norm{\packedseq_s} > \frac{1 + d + x_s}{4}
					> \frac{1 + 1/8}{4} = \frac{9}{32}.$$
				Because $Q_1$ is a medium square, we have $x \geq 1/4$
				and get
				  $$\norm{\seq} \geq \norm{Q_1} + \norm{\packedseq_s}
					      > (1/4)^2 + 9/32 = 11/32.$$    
			\end{enumerate}
		  \item \emph{$Q_s$ is packed into $M_3$ or $M_4$:}\\
			By Property~\ref{prop:density_lowerhalf} and Lemma~\ref{lem:coll_H1_small_upperhalf}
			we have $\totalAreaOf{\lowerhalf} \geq 5/32$ and
			$\totalAreaOf{\upperhalf} \geq \totalAreaOf{M_3\setminus E_i\cup M_4}
						  \geq 6/32$, respectively.
			Therefore, $\totalAreaOf{\unitsq} \geq 11/32$.
		  \item \emph{$Q_s$ is a buffer square in $E_i$:}\\
			We start treating the end $E_i$ of a main packing area $M_i$
			only if $M_i\setminus E_i$ is fully \emph{used}.
			Therefore, this type of collision can be handled analogously
			to the collision of $Q_0$ with a square in $M_i$.

		  \item \emph{$Q_1$ overlaps with $M_2$ but not with $M_1$:}\\
		      This only happens if $Q_1$ provokes an overflow in the upper
		      half of \unitsq and is therefore packed into the second shelf
		      of the Ceiling Packing; see Fig.~\ref{fig:11_32_proof_h1_small1}.
		      In this case, the total area of squares from $H_1$ is greater
		      than $1/4$.
		      By assumption, $Q_s$ is placed in $M_2$ and we get an additional
		      packing area of at least $7/64$ from small squares; see
		      Property~\ref{prop:density_M_1}.
		      In total, we have
		      $\norm{\seq} \geq \norm{\packedseq_m} + \norm{\packedseq_s}
				  > 1/4 + 7/64 > 11/32$.
		  \item \emph{$Q_1$ overlaps with $M_1$:}\\
		      Because $Q_1$ intersects $M_1$, the lower boundary of $Q_1$
		      must have a distance greater than $3/4$ from the top of \unitsq;
		      see Fig.~\ref{fig:11_32_proof_h1_small2}.
		      Hence, $\norm{\packedseq_m} > 1/4\cdot(3/4 + 2/4) = 10/32$.
		      As no $H_1$-square ever touches the left half of \lowerhalf,
		      we must have an area of at least $1/2 \cdot 1/8 = 2/32$ occupied
		      by small squares. In total we get\\
		      $\norm{\seq} \geq \norm{\packedseq_m} + \norm{\packedseq_s}
				  > 10/32 + 2/32 > 11/32$. \qed
	      \end{enumerate}
      \end{enumerate}
     \end{proof}

     \begin{figure}
	\centering
	\subfigure[Packing region denotations.]{
	  \def\svgwidth{.25\textwidth}
	  \import{./}{smallPartition1.tex}
	}
	\hskip5mm
	\subfigure[A large square $Q_0$ only collides with a medium square
		    $Q_1$ if $\norm{\seq}>3/8$.]{
	  \def\svgwidth{.29\textwidth}
	  \import{./}{11_32_proof_a.tex}
	  \label{fig:11_32_proof_h1_h0}
	}
	\hskip2mm
	\subfigure[If a large square collides with a small square in $M_2$,
	      then the total area of small squares is at least $3/32$.]{
	  \def\svgwidth{.3\textwidth}
	  \import{./}{11_32_proof_b.tex}
	\label{fig:11_32_proof_h0_small1}
	}
	\subfigure[We can relate the total area of small squares packed to the
		  total length of the buffer packing.]{
	  \def\svgwidth{.29\textwidth}
	  \import{./}{11_32_proof_c.tex}
	\label{fig:11_32_proof_h0_small2}
	}
	\hskip1mm
	\subfigure[If a medium square collides with the packing in $B_3$,
		    then we packed at least 7/64 with small squares.]{
	  \def\svgwidth{.29\textwidth}
	  \import{./}{11_32_proof_d.tex}
	  \label{11_32_proof_h1_b3_1}
	}
	\hskip1mm
	\subfigure[The greater the distance of $Q_s$ to the bottom of $B_3$,
		    the more area of \lowerhalf we packed with small squares.]{
	  \def\svgwidth{.30\textwidth}
	  \import{./}{11_32_proof_e.tex}
	  \label{11_32_proof_h1_b3_2}
	}
	\subfigure[The total area of medium squares is greater than $1/4$,
		      the total area of small squares is at least $7/64$.]{
	  \def\svgwidth{.29\textwidth}
	  \import{./}{11_32_proof_f.tex}
	  \label{fig:11_32_proof_h1_small1}
	}
	\hskip5mm
	\subfigure[The total area of medium squares is greater than $5/16$,
		    the total area of small squares is at least $1/16$.]{
	  \def\svgwidth{.28\textwidth}
	  \import{./}{11_32_proof_g.tex}
	  \label{fig:11_32_proof_h1_small2}
	}
	\caption{Different types of collision that may appear if the total
		area of the input exceeds $11/32$.}
      \end{figure}
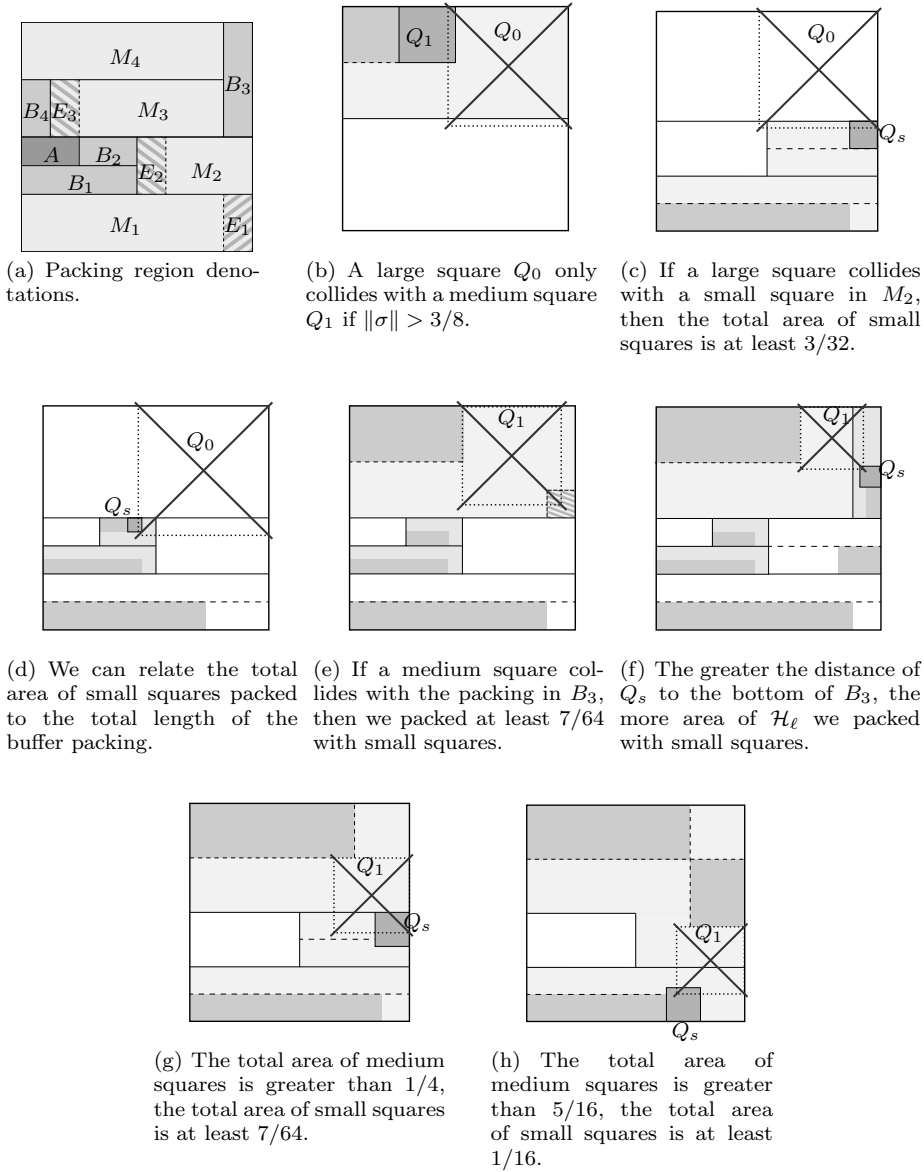

This concludes the proof of the main Theorem~\ref{thm:11over32}.

%% file: ceiling2.tex
\begingroup%
  \makeatletter%
  \providecommand\color[2][]{%
    \errmessage{(Inkscape) Color is used for the text in Inkscape, but the package 'color.sty' is not loaded}%
    \renewcommand\color[2][]{}%
  }%
  \providecommand\transparent[1]{%
    \errmessage{(Inkscape) Transparency is used (non-zero) for the text in Inkscape, but the package 'transparent.sty' is not loaded}%
    \renewcommand\transparent[1]{}%
  }%
  \providecommand\rotatebox[2]{#2}%
  \ifx\svgwidth\undefined%
    \setlength{\unitlength}{547.60625bp}%
    \ifx\svgscale\undefined%
      \relax%
    \else%
      \setlength{\unitlength}{\unitlength * \real{\svgscale}}%
    \fi%
  \else%
    \setlength{\unitlength}{\svgwidth}%
  \fi%
  \global\let\svgwidth\undefined%
  \global\let\svgscale\undefined%
  \makeatother%
  \begin{picture}(1,0.47041099)%
    \put(0,0){\includegraphics[width=\unitlength]{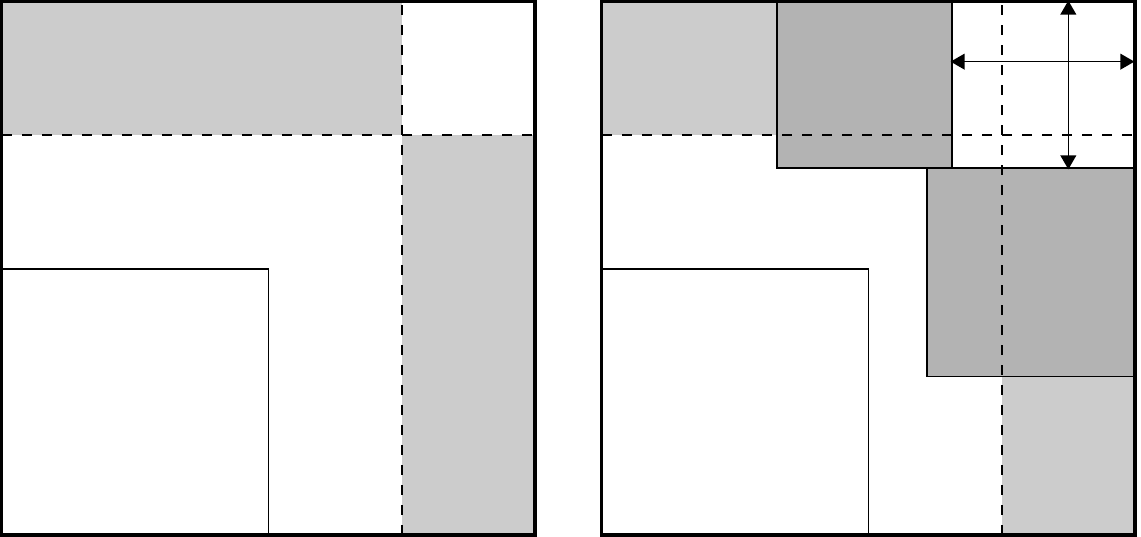}}%
    \put(0.76259173,0.3856973){\color[rgb]{0,0,0}\makebox(0,0)[b]{\smash{$Q_1$}}}%
    \put(0.91890843,0.22061518){\color[rgb]{0,0,0}\makebox(0,0)[b]{\smash{$Q_2$}}}%
    \put(0.91160391,0.42660261){\color[rgb]{0,0,0}\makebox(0,0)[b]{\smash{$d_1$}}}%
    \put(0.96857916,0.36670556){\color[rgb]{0,0,0}\makebox(0,0)[b]{\smash{$d_2$}}}%
  \end{picture}%
\endgroup%

%% file: shelf1.tex
\begingroup%
  \makeatletter%
  \providecommand\color[2][]{%
    \errmessage{(Inkscape) Color is used for the text in Inkscape, but the package 'color.sty' is not loaded}%
    \renewcommand\color[2][]{}%
  }%
  \providecommand\transparent[1]{%
    \errmessage{(Inkscape) Transparency is used (non-zero) for the text in Inkscape, but the package 'transparent.sty' is not loaded}%
    \renewcommand\transparent[1]{}%
  }%
  \providecommand\rotatebox[2]{#2}%
  \ifx\svgwidth\undefined%
    \setlength{\unitlength}{487.74375bp}%
    \ifx\svgscale\undefined%
      \relax%
    \else%
      \setlength{\unitlength}{\unitlength * \real{\svgscale}}%
    \fi%
  \else%
    \setlength{\unitlength}{\svgwidth}%
  \fi%
  \global\let\svgwidth\undefined%
  \global\let\svgscale\undefined%
  \makeatother%
  \begin{picture}(1,0.3436569)%
    \put(0,0){\includegraphics[width=\unitlength]{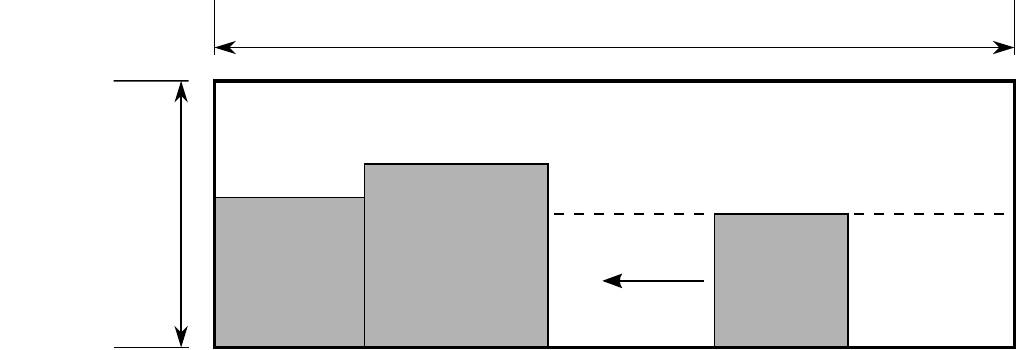}}%
    \put(0.62111252,0.32314151){\color[rgb]{0,0,0}\makebox(0,0)[b]{\smash{$w_\newshelf$}}}%
    \put(0.16841579,0.111555){\color[rgb]{0,0,0}\makebox(0,0)[rb]{\smash{$h_\newshelf$}}}%
  \end{picture}%
\endgroup%

%% file: shelf2.tex
\begingroup%
  \makeatletter%
  \providecommand\color[2][]{%
    \errmessage{(Inkscape) Color is used for the text in Inkscape, but the package 'color.sty' is not loaded}%
    \renewcommand\color[2][]{}%
  }%
  \providecommand\transparent[1]{%
    \errmessage{(Inkscape) Transparency is used (non-zero) for the text in Inkscape, but the package 'transparent.sty' is not loaded}%
    \renewcommand\transparent[1]{}%
  }%
  \providecommand\rotatebox[2]{#2}%
  \ifx\svgwidth\undefined%
    \setlength{\unitlength}{524.96875bp}%
    \ifx\svgscale\undefined%
      \relax%
    \else%
      \setlength{\unitlength}{\unitlength * \real{\svgscale}}%
    \fi%
  \else%
    \setlength{\unitlength}{\svgwidth}%
  \fi%
  \global\let\svgwidth\undefined%
  \global\let\svgscale\undefined%
  \makeatother%
  \begin{picture}(1,0.34062314)%
    \put(0,0){\includegraphics[width=\unitlength]{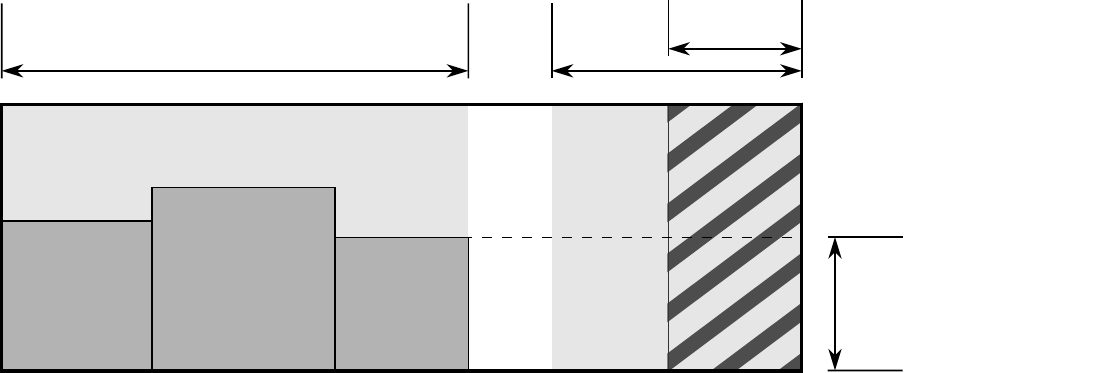}}%
    \put(0.76652182,0.03964094){\color[rgb]{0,0,0}\makebox(0,0)[lb]{\smash{$\frac{h_\newshelf}2$
}}}%
    \put(0.55927139,0.29718008){\color[rgb]{0,0,0}\makebox(0,0)[b]{\smash{$h_\newshelf$
}}}%
    \put(0.21944163,0.29718008){\color[rgb]{0,0,0}\makebox(0,0)[b]{\smash{$\ell_\newshelf$}}}%
    \put(0.6842312,0.32156248){\color[rgb]{0,0,0}\makebox(0,0)[b]{\smash{$\frac{h_\newshelf}2$
}}}%
  \end{picture}%
\endgroup%

%% file: shelfend.tex
\begingroup%
  \makeatletter%
  \providecommand\color[2][]{%
    \errmessage{(Inkscape) Color is used for the text in Inkscape, but the package 'color.sty' is not loaded}%
    \renewcommand\color[2][]{}%
  }%
  \providecommand\transparent[1]{%
    \errmessage{(Inkscape) Transparency is used (non-zero) for the text in Inkscape, but the package 'transparent.sty' is not loaded}%
    \renewcommand\transparent[1]{}%
  }%
  \providecommand\rotatebox[2]{#2}%
  \ifx\svgwidth\undefined%
    \setlength{\unitlength}{388.00499802bp}%
    \ifx\svgscale\undefined%
      \relax%
    \else%
      \setlength{\unitlength}{\unitlength * \real{\svgscale}}%
    \fi%
  \else%
    \setlength{\unitlength}{\svgwidth}%
  \fi%
  \global\let\svgwidth\undefined%
  \global\let\svgscale\undefined%
  \makeatother%
  \begin{picture}(1,0.4192931)%
    \put(0,0){\includegraphics[width=\unitlength]{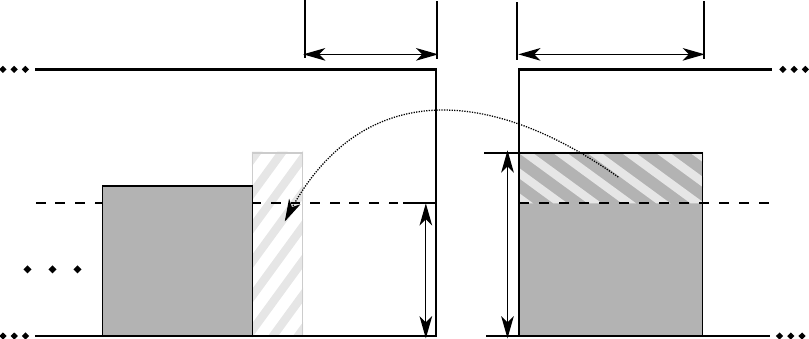}}%
    \put(0.47091128,0.39278918){\color[rgb]{0,0,0}\makebox(0,0)[b]{\smash{$\frac{h_\newshelf}2$
}}}%
    \put(0.77606198,0.37217089){\color[rgb]{0,0,0}\makebox(0,0)[b]{\smash{$x$
}}}%
    \put(0.76369101,0.0876385){\color[rgb]{0,0,0}\makebox(0,0)[b]{\smash{$Q$
}}}%
    \put(0.41546501,0.06545811){\color[rgb]{0,0,0}\makebox(0,0)[lb]{\smash{$\frac{h_\newshelf}2$
}}}%
    \put(0.13093259,0.24483724){\color[rgb]{0,0,0}\makebox(0,0)[b]{\smash{$\newshelf$
}}}%
    \put(0.57628768,0.09844737){\color[rgb]{0,0,0}\makebox(0,0)[lb]{\smash{$x$
}}}%
  \end{picture}%
\endgroup%

%% file: smallPartition.tex
\begingroup%
  \makeatletter%
  \providecommand\color[2][]{%
    \errmessage{(Inkscape) Color is used for the text in Inkscape, but the package 'color.sty' is not loaded}%
    \renewcommand\color[2][]{}%
  }%
  \providecommand\transparent[1]{%
    \errmessage{(Inkscape) Transparency is used (non-zero) for the text in Inkscape, but the package 'transparent.sty' is not loaded}%
    \renewcommand\transparent[1]{}%
  }%
  \providecommand\rotatebox[2]{#2}%
  \ifx\svgwidth\undefined%
    \setlength{\unitlength}{614.67852783bp}%
    \ifx\svgscale\undefined%
      \relax%
    \else%
      \setlength{\unitlength}{\unitlength * \real{\svgscale}}%
    \fi%
  \else%
    \setlength{\unitlength}{\svgwidth}%
  \fi%
  \global\let\svgwidth\undefined%
  \global\let\svgscale\undefined%
  \makeatother%
  \begin{picture}(1,0.9289953)%
    \put(0,0){\includegraphics[width=\unitlength]{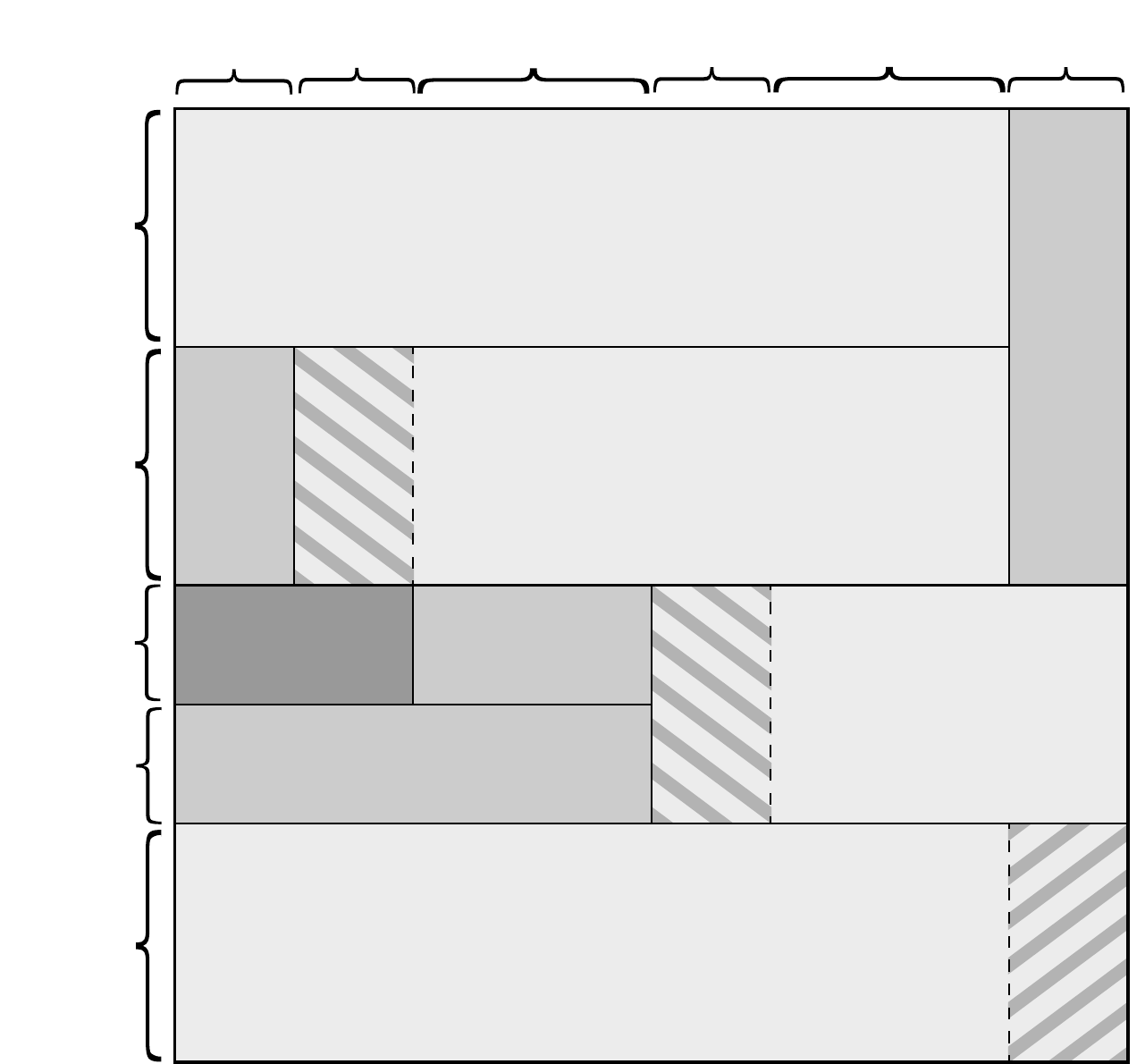}}%
    \put(0.52827774,0.68338685){\color[rgb]{0,0,0}\makebox(0,0)[b]{\smash{$M_4$}}}%
    \put(0.93348697,0.5855169){\color[rgb]{0,0,0}\makebox(0,0)[b]{\smash{$B_3$}}}%
    \put(0.62819547,0.48881024){\color[rgb]{0,0,0}\makebox(0,0)[b]{\smash{$M_3$}}}%
    \put(0.20649079,0.48618084){\color[rgb]{0,0,0}\makebox(0,0)[b]{\smash{$B_4$}}}%
    \put(0.30703687,0.48749556){\color[rgb]{0,0,0}\makebox(0,0)[b]{\smash{$E_3$}}}%
    \put(0.26003582,0.32710132){\color[rgb]{0,0,0}\makebox(0,0)[b]{\smash{$A$}}}%
    \put(0.46680274,0.32841604){\color[rgb]{0,0,0}\makebox(0,0)[b]{\smash{$B_2$}}}%
    \put(0.366885,0.22586891){\color[rgb]{0,0,0}\makebox(0,0)[b]{\smash{$B_1$}}}%
    \put(0.6199371,0.2653101){\color[rgb]{0,0,0}\makebox(0,0)[b]{\smash{$E_2$}}}%
    \put(0.82145734,0.2679395){\color[rgb]{0,0,0}\makebox(0,0)[b]{\smash{$M_2$}}}%
    \put(0.52827774,0.07204819){\color[rgb]{0,0,0}\makebox(0,0)[b]{\smash{$M_1$}}}%
    \put(0.93283726,0.07336295){\color[rgb]{0,0,0}\makebox(0,0)[b]{\smash{$E_1$}}}%
    \put(0.07335898,0.35489472){\color[rgb]{0,0,0}\makebox(0,0)[b]{\smash{$\frac{1}{8}$}}}%
    \put(0.07075599,0.23776033){\color[rgb]{0,0,0}\makebox(0,0)[b]{\smash{$\frac{1}{8}$}}}%
    \put(0.07335898,0.50066198){\color[rgb]{0,0,0}\makebox(0,0)[b]{\smash{$\frac{1}{4}$}}}%
    \put(0.07335898,0.70890091){\color[rgb]{0,0,0}\makebox(0,0)[b]{\smash{$\frac{1}{4}$}}}%
    \put(0.07335898,0.08158117){\color[rgb]{0,0,0}\makebox(0,0)[b]{\smash{$\frac{1}{4}$}}}%
    \put(0.4658456,0.9082749){\color[rgb]{0,0,0}\makebox(0,0)[b]{\smash{$\frac{1}{4}$}}}%
    \put(0.77560099,0.9082749){\color[rgb]{0,0,0}\makebox(0,0)[b]{\smash{$\frac{1}{4}$}}}%
    \put(0.2033158,0.90864676){\color[rgb]{0,0,0}\makebox(0,0)[b]{\smash{$\frac{1}{8}$}}}%
    \put(0.31115382,0.9082749){\color[rgb]{0,0,0}\makebox(0,0)[b]{\smash{$\frac{1}{8}$}}}%
    \put(0.62314035,0.90715934){\color[rgb]{0,0,0}\makebox(0,0)[b]{\smash{$\frac{1}{8}$}}}%
    \put(0.92954905,0.90864676){\color[rgb]{0,0,0}\makebox(0,0)[b]{\smash{$\frac{1}{8}$}}}%
  \end{picture}%
\endgroup%

%% file: h3_packing_new1.tex
\begingroup%
  \makeatletter%
  \providecommand\color[2][]{%
    \errmessage{(Inkscape) Color is used for the text in Inkscape, but the package 'color.sty' is not loaded}%
    \renewcommand\color[2][]{}%
  }%
  \providecommand\transparent[1]{%
    \errmessage{(Inkscape) Transparency is used (non-zero) for the text in Inkscape, but the package 'transparent.sty' is not loaded}%
    \renewcommand\transparent[1]{}%
  }%
  \providecommand\rotatebox[2]{#2}%
  \ifx\svgwidth\undefined%
    \setlength{\unitlength}{513.6bp}%
    \ifx\svgscale\undefined%
      \relax%
    \else%
      \setlength{\unitlength}{\unitlength * \real{\svgscale}}%
    \fi%
  \else%
    \setlength{\unitlength}{\svgwidth}%
  \fi%
  \global\let\svgwidth\undefined%
  \global\let\svgscale\undefined%
  \makeatother%
  \begin{picture}(1,0.50155763)%
    \put(0,0){\includegraphics[width=\unitlength]{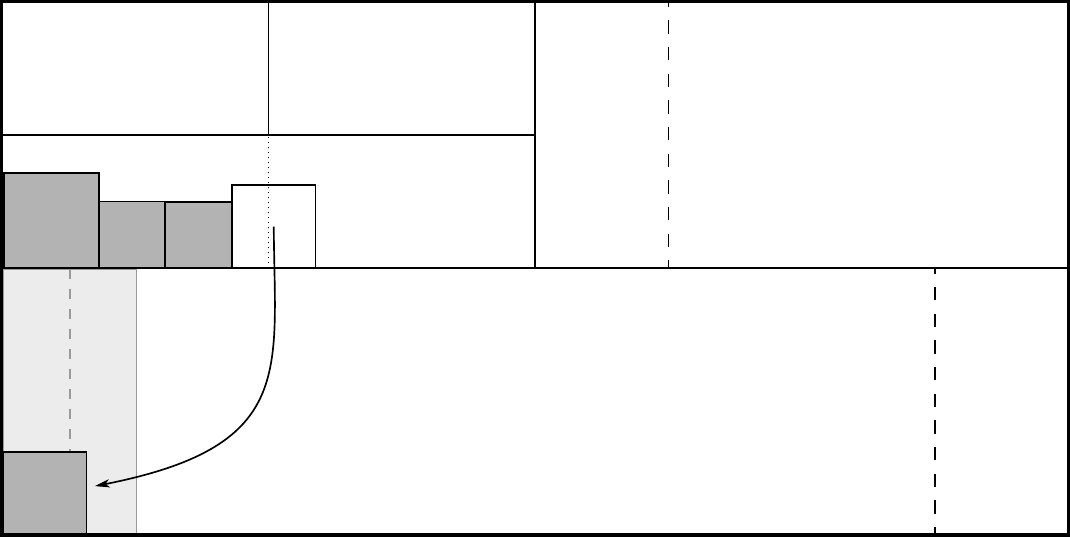}}%
    \put(0.02180685,0.30709606){\color[rgb]{0,0,0}\makebox(0,0)[lt]{\begin{minipage}{0.05562973\unitlength}\raggedright $Q_1$\end{minipage}}}%
    \put(0.09987927,0.29087966){\color[rgb]{0,0,0}\makebox(0,0)[lt]{\begin{minipage}{0.05562973\unitlength}\raggedright $Q_2$\end{minipage}}}%
    \put(0.16218456,0.29087966){\color[rgb]{0,0,0}\makebox(0,0)[lt]{\begin{minipage}{0.05562973\unitlength}\raggedright $Q_3$\end{minipage}}}%
    \put(0.01843734,0.05078174){\color[rgb]{0,0,0}\makebox(0,0)[lt]{\begin{minipage}{0.05562973\unitlength}\raggedright $Q_4$\end{minipage}}}%
    \put(0.40419376,0.33048391){\color[rgb]{0,0,0}\makebox(0,0)[lt]{\begin{minipage}{0.10165208\unitlength}\raggedright $B_1$\end{minipage}}}%
    \put(0.57429056,0.14376663){\color[rgb]{0,0,0}\makebox(0,0)[lt]{\begin{minipage}{0.10165208\unitlength}\raggedright $M_1$\end{minipage}}}%
  \end{picture}%
\endgroup%

%% file: h3_packing_new2.tex
\begingroup%
  \makeatletter%
  \providecommand\color[2][]{%
    \errmessage{(Inkscape) Color is used for the text in Inkscape, but the package 'color.sty' is not loaded}%
    \renewcommand\color[2][]{}%
  }%
  \providecommand\transparent[1]{%
    \errmessage{(Inkscape) Transparency is used (non-zero) for the text in Inkscape, but the package 'transparent.sty' is not loaded}%
    \renewcommand\transparent[1]{}%
  }%
  \providecommand\rotatebox[2]{#2}%
  \ifx\svgwidth\undefined%
    \setlength{\unitlength}{513.6bp}%
    \ifx\svgscale\undefined%
      \relax%
    \else%
      \setlength{\unitlength}{\unitlength * \real{\svgscale}}%
    \fi%
  \else%
    \setlength{\unitlength}{\svgwidth}%
  \fi%
  \global\let\svgwidth\undefined%
  \global\let\svgscale\undefined%
  \makeatother%
  \begin{picture}(1,0.50155763)%
    \put(0,0){\includegraphics[width=\unitlength]{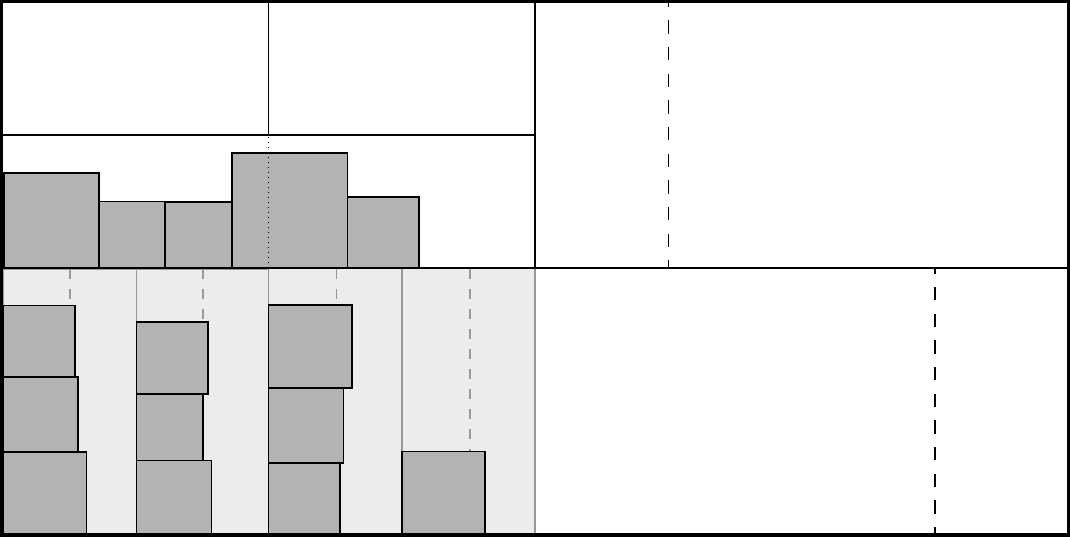}}%
    \put(0.02180685,0.30709606){\color[rgb]{0,0,0}\makebox(0,0)[lt]{\begin{minipage}{0.05562973\unitlength}\raggedright $Q_1$\end{minipage}}}%
    \put(0.09987927,0.29087966){\color[rgb]{0,0,0}\makebox(0,0)[lt]{\begin{minipage}{0.05562973\unitlength}\raggedright $Q_2$\end{minipage}}}%
    \put(0.16218456,0.29087966){\color[rgb]{0,0,0}\makebox(0,0)[lt]{\begin{minipage}{0.05562973\unitlength}\raggedright $Q_3$\end{minipage}}}%
    \put(0.01843734,0.05078174){\color[rgb]{0,0,0}\makebox(0,0)[lt]{\begin{minipage}{0.05562973\unitlength}\raggedright $Q_4$\end{minipage}}}%
    \put(0.40419376,0.33048391){\color[rgb]{0,0,0}\makebox(0,0)[lt]{\begin{minipage}{0.10165208\unitlength}\raggedright $B_1$\end{minipage}}}%
    \put(0.57429056,0.14376663){\color[rgb]{0,0,0}\makebox(0,0)[lt]{\begin{minipage}{0.10165208\unitlength}\raggedright $M_1$\end{minipage}}}%
    \put(0.24540664,0.3160243){\color[rgb]{0,0,0}\makebox(0,0)[lt]{\begin{minipage}{0.05562973\unitlength}\raggedright $Q_5$\end{minipage}}}%
    \put(0.01287437,0.12310041){\color[rgb]{0,0,0}\makebox(0,0)[lt]{\begin{minipage}{0.05561064\unitlength}\raggedright $Q_6$\end{minipage}}}%
    \put(0.00953659,0.18985608){\color[rgb]{0,0,0}\makebox(0,0)[lt]{\begin{minipage}{0.05562973\unitlength}\raggedright $Q_7$\\ \end{minipage}}}%
    \put(0.13748496,0.04521879){\color[rgb]{0,0,0}\makebox(0,0)[lt]{\begin{minipage}{0.05562973\unitlength}\raggedright $Q_8$\end{minipage}}}%
    \put(0.13414718,0.11086187){\color[rgb]{0,0,0}\makebox(0,0)[lt]{\begin{minipage}{0.05562973\unitlength}\raggedright $Q_9$\end{minipage}}}%
    \put(0.13659489,0.17850761){\color[rgb]{0,0,0}\makebox(0,0)[lt]{\begin{minipage}{0.07788163\unitlength}\raggedright $Q_{10}$\end{minipage}}}%
    \put(0.25468637,0.18651829){\color[rgb]{0,0,0}\makebox(0,0)[lt]{\begin{minipage}{0.07343125\unitlength}\raggedright $Q_{14}$\end{minipage}}}%
    \put(0.32866044,0.29549374){\color[rgb]{0,0,0}\makebox(0,0)[lt]{\begin{minipage}{0.07343125\unitlength}\raggedright $Q_{12}$\end{minipage}}}%
    \put(0.2491234,0.11308706){\color[rgb]{0,0,0}\makebox(0,0)[lt]{\begin{minipage}{0.07343125\unitlength}\raggedright $Q_{13}$\end{minipage}}}%
    \put(0.24467302,0.04633139){\color[rgb]{0,0,0}\makebox(0,0)[lt]{\begin{minipage}{0.07343125\unitlength}\raggedright $Q_{11}$\end{minipage}}}%
    \put(0.37763156,0.04767213){\color[rgb]{0,0,0}\makebox(0,0)[lt]{\begin{minipage}{0.07343125\unitlength}\raggedright $Q_{15}$\end{minipage}}}%
  \end{picture}%
\endgroup%

%% file: end1.tex
\begingroup%
  \makeatletter%
  \providecommand\color[2][]{%
    \errmessage{(Inkscape) Color is used for the text in Inkscape, but the package 'color.sty' is not loaded}%
    \renewcommand\color[2][]{}%
  }%
  \providecommand\transparent[1]{%
    \errmessage{(Inkscape) Transparency is used (non-zero) for the text in Inkscape, but the package 'transparent.sty' is not loaded}%
    \renewcommand\transparent[1]{}%
  }%
  \providecommand\rotatebox[2]{#2}%
  \ifx\svgwidth\undefined%
    \setlength{\unitlength}{123.25713501bp}%
    \ifx\svgscale\undefined%
      \relax%
    \else%
      \setlength{\unitlength}{\unitlength * \real{\svgscale}}%
    \fi%
  \else%
    \setlength{\unitlength}{\svgwidth}%
  \fi%
  \global\let\svgwidth\undefined%
  \global\let\svgscale\undefined%
  \makeatother%
  \begin{picture}(1,1.39638022)%
    \put(0,0){\includegraphics[width=\unitlength]{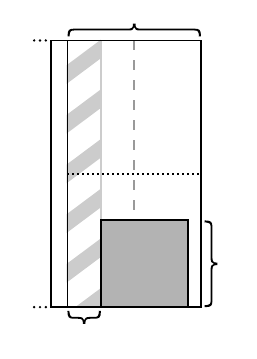}}%
    \put(0.32777007,0.04450253){\color[rgb]{0,0,0}\makebox(0,0)[b]{\smash{$\ell_i$}}}%
    \put(0.32777007,0.78441912){\color[rgb]{0,0,0}\makebox(0,0)[b]{\smash{$\overlapshelf_i$}}}%
    \put(0.33933127,1.3813855){\color[rgb]{0,0,0}\makebox(0,0)[lt]{\begin{minipage}{0.36624943\unitlength}\centering $1/8$\end{minipage}}}%
    \put(0.87780291,0.35685262){\color[rgb]{0,0,0}\makebox(0,0)[lb]{\smash{$x_Q$}}}%
    \put(0.45526199,0.40657948){\color[rgb]{0,0,0}\makebox(0,0)[lt]{\begin{minipage}{0.21210014\unitlength}\centering $Q$\end{minipage}}}%
  \end{picture}%
\endgroup%

%% file: end3.tex
\begingroup%
  \makeatletter%
  \providecommand\color[2][]{%
    \errmessage{(Inkscape) Color is used for the text in Inkscape, but the package 'color.sty' is not loaded}%
    \renewcommand\color[2][]{}%
  }%
  \providecommand\transparent[1]{%
    \errmessage{(Inkscape) Transparency is used (non-zero) for the text in Inkscape, but the package 'transparent.sty' is not loaded}%
    \renewcommand\transparent[1]{}%
  }%
  \providecommand\rotatebox[2]{#2}%
  \ifx\svgwidth\undefined%
    \setlength{\unitlength}{123.25713501bp}%
    \ifx\svgscale\undefined%
      \relax%
    \else%
      \setlength{\unitlength}{\unitlength * \real{\svgscale}}%
    \fi%
  \else%
    \setlength{\unitlength}{\svgwidth}%
  \fi%
  \global\let\svgwidth\undefined%
  \global\let\svgscale\undefined%
  \makeatother%
  \begin{picture}(1,1.39638022)%
    \put(0,0){\includegraphics[width=\unitlength]{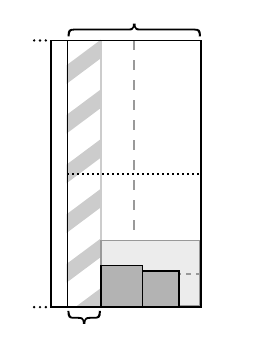}}%
    \put(0.32777007,0.04450253){\color[rgb]{0,0,0}\makebox(0,0)[b]{\smash{$\ell_i$}}}%
    \put(0.32777007,0.78441912){\color[rgb]{0,0,0}\makebox(0,0)[b]{\smash{$\overlapshelf_i$}}}%
    \put(0.33933127,1.3813855){\color[rgb]{0,0,0}\makebox(0,0)[lt]{\begin{minipage}{0.36624943\unitlength}\centering $1/8$\end{minipage}}}%
    \put(0.36439504,0.31571253){\color[rgb]{0,0,0}\makebox(0,0)[lt]{\begin{minipage}{0.21210014\unitlength}\centering $Q$\end{minipage}}}%
  \end{picture}%
\endgroup%

%% file: upperhalfCollision1.tex
\begingroup%
  \makeatletter%
  \providecommand\color[2][]{%
    \errmessage{(Inkscape) Color is used for the text in Inkscape, but the package 'color.sty' is not loaded}%
    \renewcommand\color[2][]{}%
  }%
  \providecommand\transparent[1]{%
    \errmessage{(Inkscape) Transparency is used (non-zero) for the text in Inkscape, but the package 'transparent.sty' is not loaded}%
    \renewcommand\transparent[1]{}%
  }%
  \providecommand\rotatebox[2]{#2}%
  \ifx\svgwidth\undefined%
    \setlength{\unitlength}{514.4bp}%
    \ifx\svgscale\undefined%
      \relax%
    \else%
      \setlength{\unitlength}{\unitlength * \real{\svgscale}}%
    \fi%
  \else%
    \setlength{\unitlength}{\svgwidth}%
  \fi%
  \global\let\svgwidth\undefined%
  \global\let\svgscale\undefined%
  \makeatother%
  \begin{picture}(1,0.50223561)%
    \put(0,0){\includegraphics[width=\unitlength]{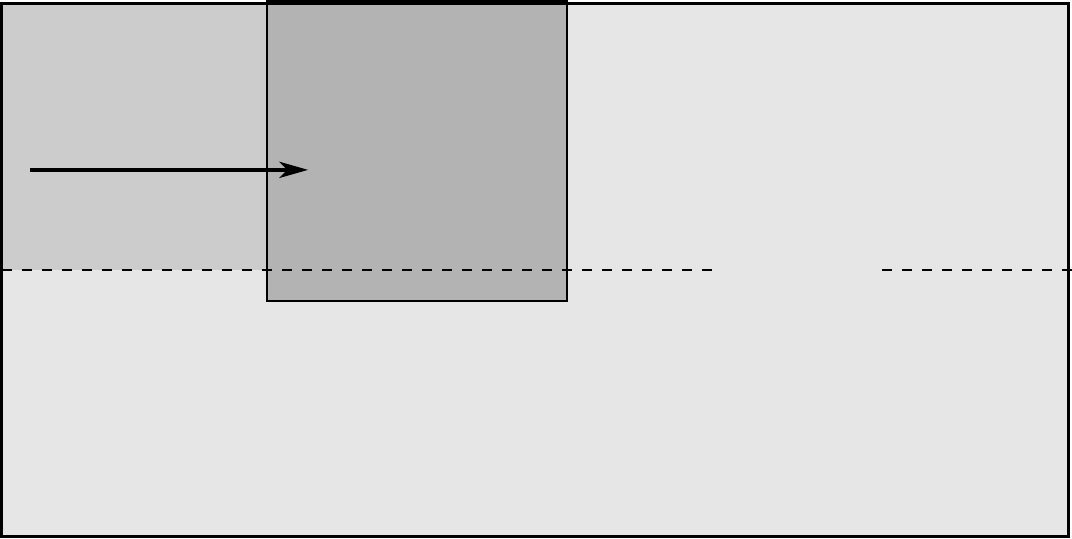}}%
    \put(0.74361252,0.2393001){\color[rgb]{0,0,0}\makebox(0,0)[b]{\smash{$\upperhalf$}}}%
    \put(0.38924682,0.34261046){\color[rgb]{0,0,0}\makebox(0,0)[b]{\smash{$Q$}}}%
  \end{picture}%
\endgroup%

%% file: upperhalfCollision2.tex
\begingroup%
  \makeatletter%
  \providecommand\color[2][]{%
    \errmessage{(Inkscape) Color is used for the text in Inkscape, but the package 'color.sty' is not loaded}%
    \renewcommand\color[2][]{}%
  }%
  \providecommand\transparent[1]{%
    \errmessage{(Inkscape) Transparency is used (non-zero) for the text in Inkscape, but the package 'transparent.sty' is not loaded}%
    \renewcommand\transparent[1]{}%
  }%
  \providecommand\rotatebox[2]{#2}%
  \ifx\svgwidth\undefined%
    \setlength{\unitlength}{513.6bp}%
    \ifx\svgscale\undefined%
      \relax%
    \else%
      \setlength{\unitlength}{\unitlength * \real{\svgscale}}%
    \fi%
  \else%
    \setlength{\unitlength}{\svgwidth}%
  \fi%
  \global\let\svgwidth\undefined%
  \global\let\svgscale\undefined%
  \makeatother%
  \begin{picture}(1,0.50155763)%
    \put(0,0){\includegraphics[width=\unitlength]{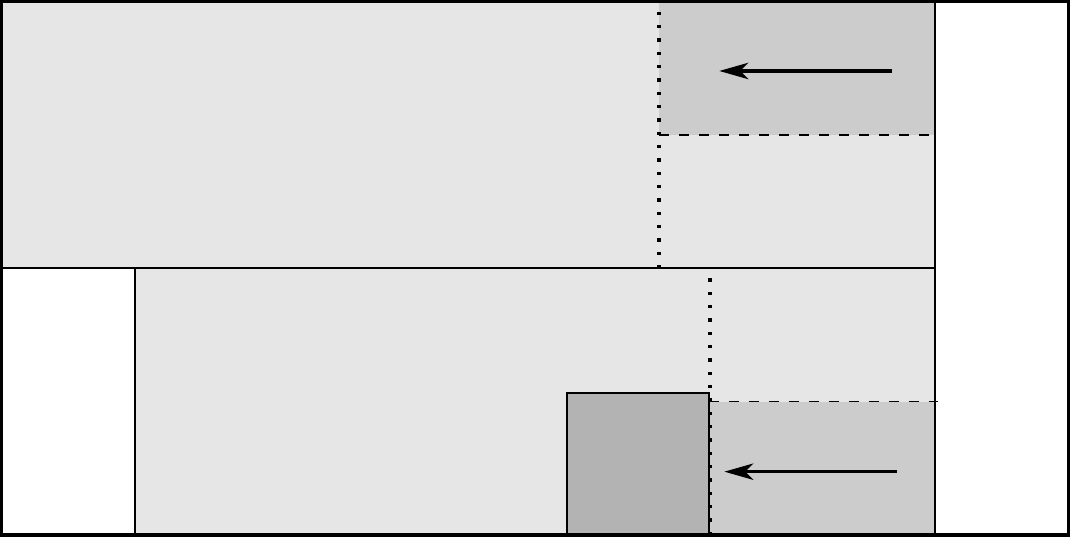}}%
    \put(0.93613707,0.23522246){\color[rgb]{0,0,0}\makebox(0,0)[b]{\smash{$B_3$}}}%
    \put(0.06386293,0.11061187){\color[rgb]{0,0,0}\makebox(0,0)[b]{\smash{$B_4$}}}%
    \put(0.46261682,0.11061187){\color[rgb]{0,0,0}\makebox(0,0)[b]{\smash{$M_3$}}}%
    \put(0.37538941,0.35983305){\color[rgb]{0,0,0}\makebox(0,0)[b]{\smash{$M_4$}}}%
    \put(0.59234534,0.05164435){\color[rgb]{0,0,0}\makebox(0,0)[b]{\smash{$Q$}}}%
  \end{picture}%
\endgroup%

%% file: upper1.tex
\begingroup%
  \makeatletter%
  \providecommand\color[2][]{%
    \errmessage{(Inkscape) Color is used for the text in Inkscape, but the package 'color.sty' is not loaded}%
    \renewcommand\color[2][]{}%
  }%
  \providecommand\transparent[1]{%
    \errmessage{(Inkscape) Transparency is used (non-zero) for the text in Inkscape, but the package 'transparent.sty' is not loaded}%
    \renewcommand\transparent[1]{}%
  }%
  \providecommand\rotatebox[2]{#2}%
  \ifx\svgwidth\undefined%
    \setlength{\unitlength}{312.8bp}%
    \ifx\svgscale\undefined%
      \relax%
    \else%
      \setlength{\unitlength}{\unitlength * \real{\svgscale}}%
    \fi%
  \else%
    \setlength{\unitlength}{\svgwidth}%
  \fi%
  \global\let\svgwidth\undefined%
  \global\let\svgscale\undefined%
  \makeatother%
  \begin{picture}(1,0.76982097)%
    \put(0,0){\includegraphics[width=\unitlength]{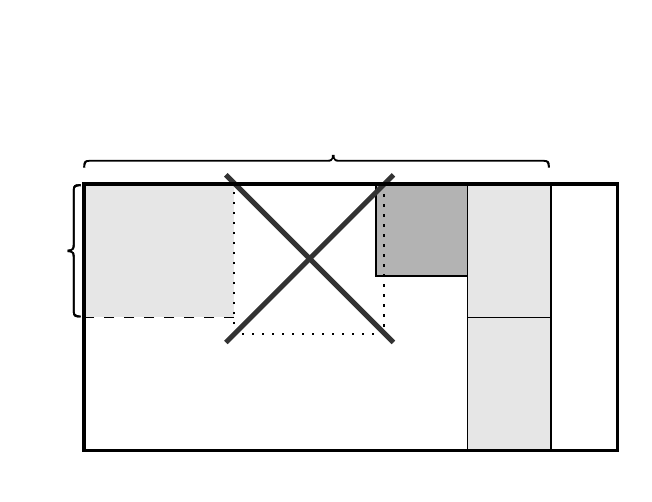}}%
    \put(-0.00289633,0.39550119){\color[rgb]{0,0,0}\makebox(0,0)[lt]{\begin{minipage}{0.09864815\unitlength}\raggedleft $\frac14$\end{minipage}}}%
    \put(0.15386265,0.39526589){\color[rgb]{0,0,0}\makebox(0,0)[lt]{\begin{minipage}{0.1543661\unitlength}\centering 100\%\end{minipage}}}%
    \put(0.71037964,0.19486399){\color[rgb]{0,0,0}\makebox(0,0)[lt]{\begin{minipage}{0.1543661\unitlength}\centering 50\%\end{minipage}}}%
    \put(0.71012865,0.39526588){\color[rgb]{0,0,0}\makebox(0,0)[lt]{\begin{minipage}{0.1543661\unitlength}\centering 50\%\end{minipage}}}%
    \put(0.4158975,0.32345369){\color[rgb]{0,0,0}\makebox(0,0)[lt]{\begin{minipage}{0.12331018\unitlength}\centering $Q_1$\end{minipage}}}%
    \put(0.58320804,0.42428104){\color[rgb]{0,0,0}\makebox(0,0)[lt]{\begin{minipage}{0.12331018\unitlength}\centering $Q_2$\end{minipage}}}%
    \put(0.43204239,0.62644857){\color[rgb]{0,0,0}\makebox(0,0)[lt]{\begin{minipage}{0.16258677\unitlength}\centering $\frac78$\end{minipage}}}%
    \put(-0.00289633,0.39550119){\color[rgb]{0,0,0}\makebox(0,0)[lt]{\begin{minipage}{0.09864815\unitlength}\raggedleft $\frac14$\end{minipage}}}%
  \end{picture}%
\endgroup%

%% file: upper2.tex
\begingroup%
  \makeatletter%
  \providecommand\color[2][]{%
    \errmessage{(Inkscape) Color is used for the text in Inkscape, but the package 'color.sty' is not loaded}%
    \renewcommand\color[2][]{}%
  }%
  \providecommand\transparent[1]{%
    \errmessage{(Inkscape) Transparency is used (non-zero) for the text in Inkscape, but the package 'transparent.sty' is not loaded}%
    \renewcommand\transparent[1]{}%
  }%
  \providecommand\rotatebox[2]{#2}%
  \ifx\svgwidth\undefined%
    \setlength{\unitlength}{312.8bp}%
    \ifx\svgscale\undefined%
      \relax%
    \else%
      \setlength{\unitlength}{\unitlength * \real{\svgscale}}%
    \fi%
  \else%
    \setlength{\unitlength}{\svgwidth}%
  \fi%
  \global\let\svgwidth\undefined%
  \global\let\svgscale\undefined%
  \makeatother%
  \begin{picture}(1,0.76982097)%
    \put(0,0){\includegraphics[width=\unitlength]{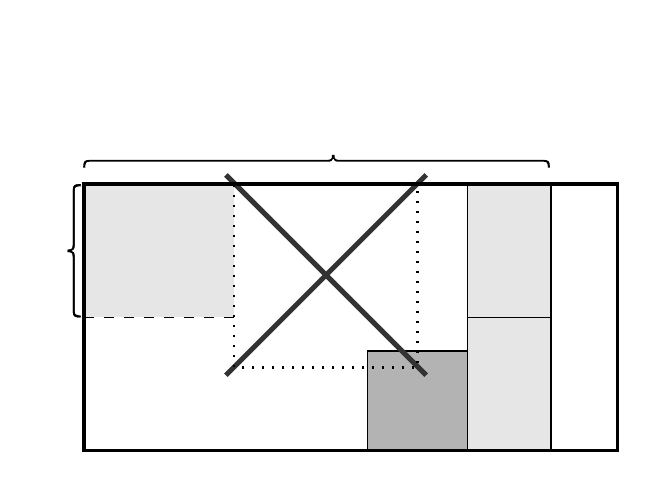}}%
    \put(-0.00289633,0.39550119){\color[rgb]{0,0,0}\makebox(0,0)[lt]{\begin{minipage}{0.09864815\unitlength}\raggedleft $\frac14$\end{minipage}}}%
    \put(0.15386265,0.39526589){\color[rgb]{0,0,0}\makebox(0,0)[lt]{\begin{minipage}{0.1543661\unitlength}\centering 100\%\end{minipage}}}%
    \put(0.71037964,0.19486399){\color[rgb]{0,0,0}\makebox(0,0)[lt]{\begin{minipage}{0.1543661\unitlength}\centering 50\%\end{minipage}}}%
    \put(0.71012865,0.39526588){\color[rgb]{0,0,0}\makebox(0,0)[lt]{\begin{minipage}{0.1543661\unitlength}\centering 50\%\end{minipage}}}%
    \put(0.43579683,0.29392692){\color[rgb]{0,0,0}\makebox(0,0)[lt]{\begin{minipage}{0.12331018\unitlength}\centering $Q_1$\end{minipage}}}%
    \put(0.57901934,0.1558195){\color[rgb]{0,0,0}\makebox(0,0)[lt]{\begin{minipage}{0.12331018\unitlength}\centering $Q_2$\end{minipage}}}%
    \put(0.43204239,0.62644857){\color[rgb]{0,0,0}\makebox(0,0)[lt]{\begin{minipage}{0.16258677\unitlength}\centering $\frac78$\end{minipage}}}%
    \put(-0.00289633,0.39550119){\color[rgb]{0,0,0}\makebox(0,0)[lt]{\begin{minipage}{0.09864815\unitlength}\raggedleft $\frac14$\end{minipage}}}%
  \end{picture}%
\endgroup%

%% file: upper3.tex
\begingroup%
  \makeatletter%
  \providecommand\color[2][]{%
    \errmessage{(Inkscape) Color is used for the text in Inkscape, but the package 'color.sty' is not loaded}%
    \renewcommand\color[2][]{}%
  }%
  \providecommand\transparent[1]{%
    \errmessage{(Inkscape) Transparency is used (non-zero) for the text in Inkscape, but the package 'transparent.sty' is not loaded}%
    \renewcommand\transparent[1]{}%
  }%
  \providecommand\rotatebox[2]{#2}%
  \ifx\svgwidth\undefined%
    \setlength{\unitlength}{312.8bp}%
    \ifx\svgscale\undefined%
      \relax%
    \else%
      \setlength{\unitlength}{\unitlength * \real{\svgscale}}%
    \fi%
  \else%
    \setlength{\unitlength}{\svgwidth}%
  \fi%
  \global\let\svgwidth\undefined%
  \global\let\svgscale\undefined%
  \makeatother%
  \begin{picture}(1,0.83439898)%
    \put(0,0){\includegraphics[width=\unitlength]{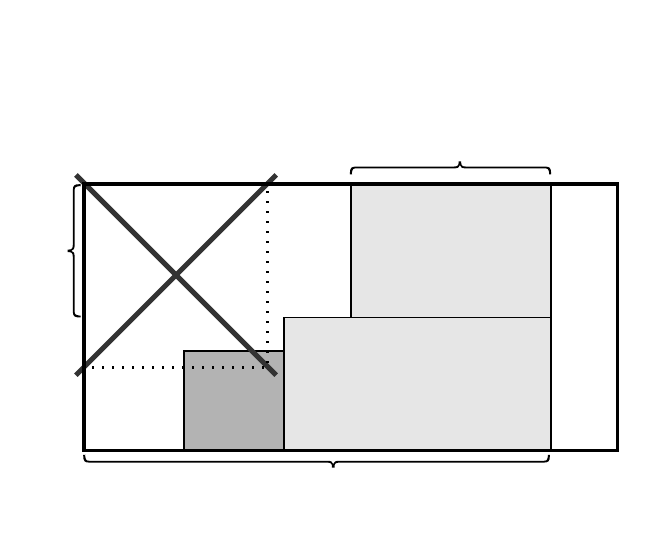}}%
    \put(0.2056178,0.35850493){\color[rgb]{0,0,0}\makebox(0,0)[lt]{\begin{minipage}{0.12331018\unitlength}\centering $Q_1$\end{minipage}}}%
    \put(0.56715714,0.259442){\color[rgb]{0,0,0}\makebox(0,0)[lt]{\begin{minipage}{0.1543661\unitlength}\centering 50\%\end{minipage}}}%
    \put(0.61805703,0.45984389){\color[rgb]{0,0,0}\makebox(0,0)[lt]{\begin{minipage}{0.1543661\unitlength}\centering 50\%\end{minipage}}}%
    \put(0.29768941,0.2255126){\color[rgb]{0,0,0}\makebox(0,0)[lt]{\begin{minipage}{0.12331018\unitlength}\centering $Q_2$\end{minipage}}}%
    \put(0.62329878,0.67198773){\color[rgb]{0,0,0}\makebox(0,0)[lt]{\begin{minipage}{0.16258677\unitlength}\centering $\frac38$\end{minipage}}}%
    \put(0.43204239,0.07721583){\color[rgb]{0,0,0}\makebox(0,0)[lt]{\begin{minipage}{0.16258677\unitlength}\centering $\frac78$\end{minipage}}}%
    \put(-0.00289633,0.46007919){\color[rgb]{0,0,0}\makebox(0,0)[lt]{\begin{minipage}{0.09864815\unitlength}\raggedleft $\frac14$\end{minipage}}}%
  \end{picture}%
\endgroup%

%% file: upper2v.tex
\begingroup%
  \makeatletter%
  \providecommand\color[2][]{%
    \errmessage{(Inkscape) Color is used for the text in Inkscape, but the package 'color.sty' is not loaded}%
    \renewcommand\color[2][]{}%
  }%
  \providecommand\transparent[1]{%
    \errmessage{(Inkscape) Transparency is used (non-zero) for the text in Inkscape, but the package 'transparent.sty' is not loaded}%
    \renewcommand\transparent[1]{}%
  }%
  \providecommand\rotatebox[2]{#2}%
  \ifx\svgwidth\undefined%
    \setlength{\unitlength}{312.8bp}%
    \ifx\svgscale\undefined%
      \relax%
    \else%
      \setlength{\unitlength}{\unitlength * \real{\svgscale}}%
    \fi%
  \else%
    \setlength{\unitlength}{\svgwidth}%
  \fi%
  \global\let\svgwidth\undefined%
  \global\let\svgscale\undefined%
  \makeatother%
  \begin{picture}(1,0.83439898)%
    \put(0,0){\includegraphics[width=\unitlength]{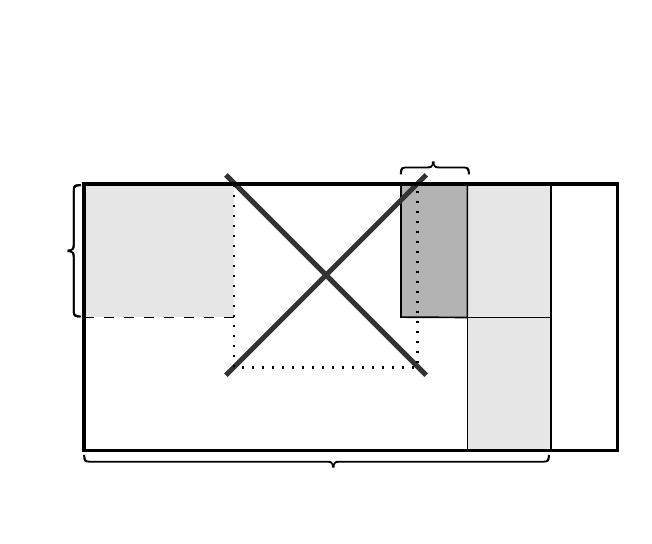}}%
    \put(-0.00289633,0.46007919){\color[rgb]{0,0,0}\makebox(0,0)[lt]{\begin{minipage}{0.09864815\unitlength}\raggedleft $\frac14$\end{minipage}}}%
    \put(0.15386265,0.45984389){\color[rgb]{0,0,0}\makebox(0,0)[lt]{\begin{minipage}{0.1543661\unitlength}\centering 100\%\end{minipage}}}%
    \put(0.71037964,0.259442){\color[rgb]{0,0,0}\makebox(0,0)[lt]{\begin{minipage}{0.1543661\unitlength}\centering 50\%\end{minipage}}}%
    \put(0.71012865,0.45984389){\color[rgb]{0,0,0}\makebox(0,0)[lt]{\begin{minipage}{0.1543661\unitlength}\centering 50\%\end{minipage}}}%
    \put(0.66715387,0.44038717){\color[rgb]{0,0,0}\makebox(0,0)[b]{\smash{$\mathcal{S}$}}}%
    \put(0.44091192,0.34315966){\color[rgb]{0,0,0}\makebox(0,0)[lt]{\begin{minipage}{0.12331018\unitlength}\centering $Q_1$\end{minipage}}}%
    \put(0.58549507,0.65010586){\color[rgb]{0,0,0}\makebox(0,0)[lt]{\begin{minipage}{0.16258677\unitlength}\centering $w_{\mathcal{S}}$\end{minipage}}}%
    \put(0.43204239,0.07721583){\color[rgb]{0,0,0}\makebox(0,0)[lt]{\begin{minipage}{0.16258677\unitlength}\centering $\frac78$\end{minipage}}}%
    \put(-0.00289633,0.46007919){\color[rgb]{0,0,0}\makebox(0,0)[lt]{\begin{minipage}{0.09864815\unitlength}\raggedleft $\frac14$\end{minipage}}}%
  \end{picture}%
\endgroup%

%% file: upper3v.tex
\begingroup%
  \makeatletter%
  \providecommand\color[2][]{%
    \errmessage{(Inkscape) Color is used for the text in Inkscape, but the package 'color.sty' is not loaded}%
    \renewcommand\color[2][]{}%
  }%
  \providecommand\transparent[1]{%
    \errmessage{(Inkscape) Transparency is used (non-zero) for the text in Inkscape, but the package 'transparent.sty' is not loaded}%
    \renewcommand\transparent[1]{}%
  }%
  \providecommand\rotatebox[2]{#2}%
  \ifx\svgwidth\undefined%
    \setlength{\unitlength}{312.8bp}%
    \ifx\svgscale\undefined%
      \relax%
    \else%
      \setlength{\unitlength}{\unitlength * \real{\svgscale}}%
    \fi%
  \else%
    \setlength{\unitlength}{\svgwidth}%
  \fi%
  \global\let\svgwidth\undefined%
  \global\let\svgscale\undefined%
  \makeatother%
  \begin{picture}(1,0.76982097)%
    \put(0,0){\includegraphics[width=\unitlength]{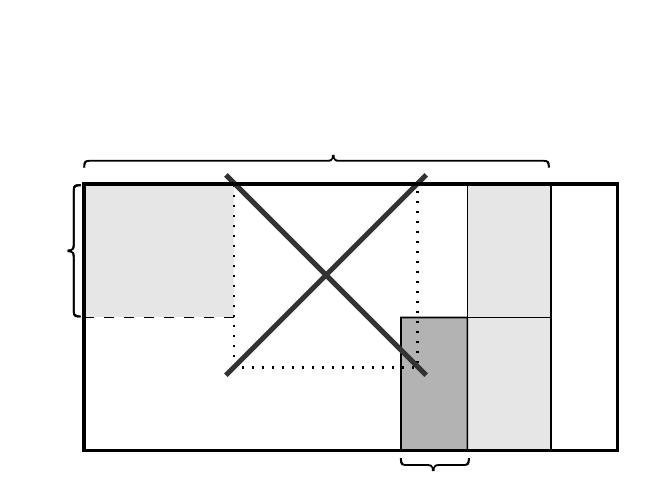}}%
    \put(-0.00289633,0.39550119){\color[rgb]{0,0,0}\makebox(0,0)[lt]{\begin{minipage}{0.09864815\unitlength}\raggedleft $\frac14$\end{minipage}}}%
    \put(0.15386265,0.39526589){\color[rgb]{0,0,0}\makebox(0,0)[lt]{\begin{minipage}{0.1543661\unitlength}\centering 100\%\end{minipage}}}%
    \put(0.71037964,0.19486399){\color[rgb]{0,0,0}\makebox(0,0)[lt]{\begin{minipage}{0.1543661\unitlength}\centering 50\%\end{minipage}}}%
    \put(0.71012865,0.39526588){\color[rgb]{0,0,0}\makebox(0,0)[lt]{\begin{minipage}{0.1543661\unitlength}\centering 50\%\end{minipage}}}%
    \put(0.44091192,0.28369674){\color[rgb]{0,0,0}\makebox(0,0)[lt]{\begin{minipage}{0.12331018\unitlength}\centering $Q_1$\end{minipage}}}%
    \put(0.66715387,0.17120556){\color[rgb]{0,0,0}\makebox(0,0)[b]{\smash{$\mathcal{S}$}}}%
    \put(0.58549507,0.03821328){\color[rgb]{0,0,0}\makebox(0,0)[lt]{\begin{minipage}{0.16258677\unitlength}\centering $w_{\mathcal{S}}$\end{minipage}}}%
    \put(0.43204239,0.62644857){\color[rgb]{0,0,0}\makebox(0,0)[lt]{\begin{minipage}{0.16258677\unitlength}\centering $\frac78$\end{minipage}}}%
    \put(-0.00289633,0.39550119){\color[rgb]{0,0,0}\makebox(0,0)[lt]{\begin{minipage}{0.09864815\unitlength}\raggedleft $\frac14$\end{minipage}}}%
  \end{picture}%
\endgroup%

%% file: upper1v.tex
\begingroup%
  \makeatletter%
  \providecommand\color[2][]{%
    \errmessage{(Inkscape) Color is used for the text in Inkscape, but the package 'color.sty' is not loaded}%
    \renewcommand\color[2][]{}%
  }%
  \providecommand\transparent[1]{%
    \errmessage{(Inkscape) Transparency is used (non-zero) for the text in Inkscape, but the package 'transparent.sty' is not loaded}%
    \renewcommand\transparent[1]{}%
  }%
  \providecommand\rotatebox[2]{#2}%
  \ifx\svgwidth\undefined%
    \setlength{\unitlength}{312.8bp}%
    \ifx\svgscale\undefined%
      \relax%
    \else%
      \setlength{\unitlength}{\unitlength * \real{\svgscale}}%
    \fi%
  \else%
    \setlength{\unitlength}{\svgwidth}%
  \fi%
  \global\let\svgwidth\undefined%
  \global\let\svgscale\undefined%
  \makeatother%
  \begin{picture}(1,0.76982097)%
    \put(0,0){\includegraphics[width=\unitlength]{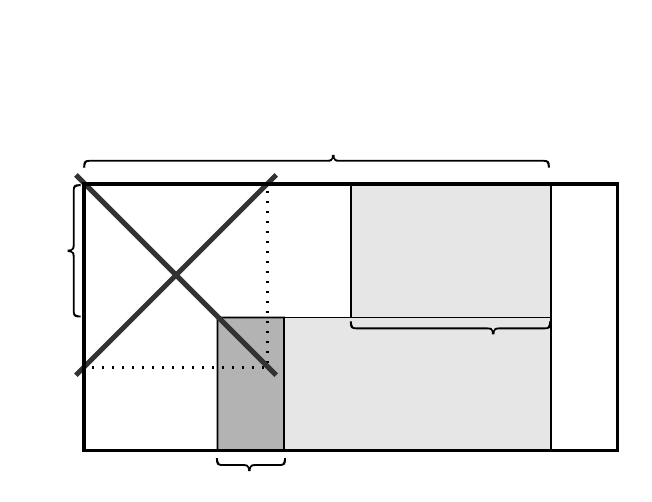}}%
    \put(0.61294194,0.39526588){\color[rgb]{0,0,0}\makebox(0,0)[lt]{\begin{minipage}{0.1543661\unitlength}\centering 50\%\end{minipage}}}%
    \put(0.56204205,0.19486399){\color[rgb]{0,0,0}\makebox(0,0)[lt]{\begin{minipage}{0.1543661\unitlength}\centering 50\%\end{minipage}}}%
    \put(0.21073289,0.28369674){\color[rgb]{0,0,0}\makebox(0,0)[lt]{\begin{minipage}{0.12331018\unitlength}\centering $Q_1$\end{minipage}}}%
    \put(0.38582394,0.1558603){\color[rgb]{0,0,0}\makebox(0,0)[b]{\smash{$\mathcal{S}$}}}%
    \put(0.30075298,0.03821328){\color[rgb]{0,0,0}\makebox(0,0)[lt]{\begin{minipage}{0.16258677\unitlength}\centering $w_{\mathcal{S}}$\end{minipage}}}%
    \put(0.67444968,0.23400819){\color[rgb]{0,0,0}\makebox(0,0)[lt]{\begin{minipage}{0.16258677\unitlength}\centering $\frac38$\end{minipage}}}%
    \put(0.43204239,0.62644857){\color[rgb]{0,0,0}\makebox(0,0)[lt]{\begin{minipage}{0.16258677\unitlength}\centering $\frac78$\end{minipage}}}%
    \put(-0.00289633,0.39550119){\color[rgb]{0,0,0}\makebox(0,0)[lt]{\begin{minipage}{0.09864815\unitlength}\raggedleft $\frac14$\end{minipage}}}%
  \end{picture}%
\endgroup%

%% file: smallPartition1.tex
\begingroup%
  \makeatletter%
  \providecommand\color[2][]{%
    \errmessage{(Inkscape) Color is used for the text in Inkscape, but the package 'color.sty' is not loaded}%
    \renewcommand\color[2][]{}%
  }%
  \providecommand\transparent[1]{%
    \errmessage{(Inkscape) Transparency is used (non-zero) for the text in Inkscape, but the package 'transparent.sty' is not loaded}%
    \renewcommand\transparent[1]{}%
  }%
  \providecommand\rotatebox[2]{#2}%
  \ifx\svgwidth\undefined%
    \setlength{\unitlength}{513.6bp}%
    \ifx\svgscale\undefined%
      \relax%
    \else%
      \setlength{\unitlength}{\unitlength * \real{\svgscale}}%
    \fi%
  \else%
    \setlength{\unitlength}{\svgwidth}%
  \fi%
  \global\let\svgwidth\undefined%
  \global\let\svgscale\undefined%
  \makeatother%
  \begin{picture}(1,1)%
    \put(0,0){\includegraphics[width=\unitlength]{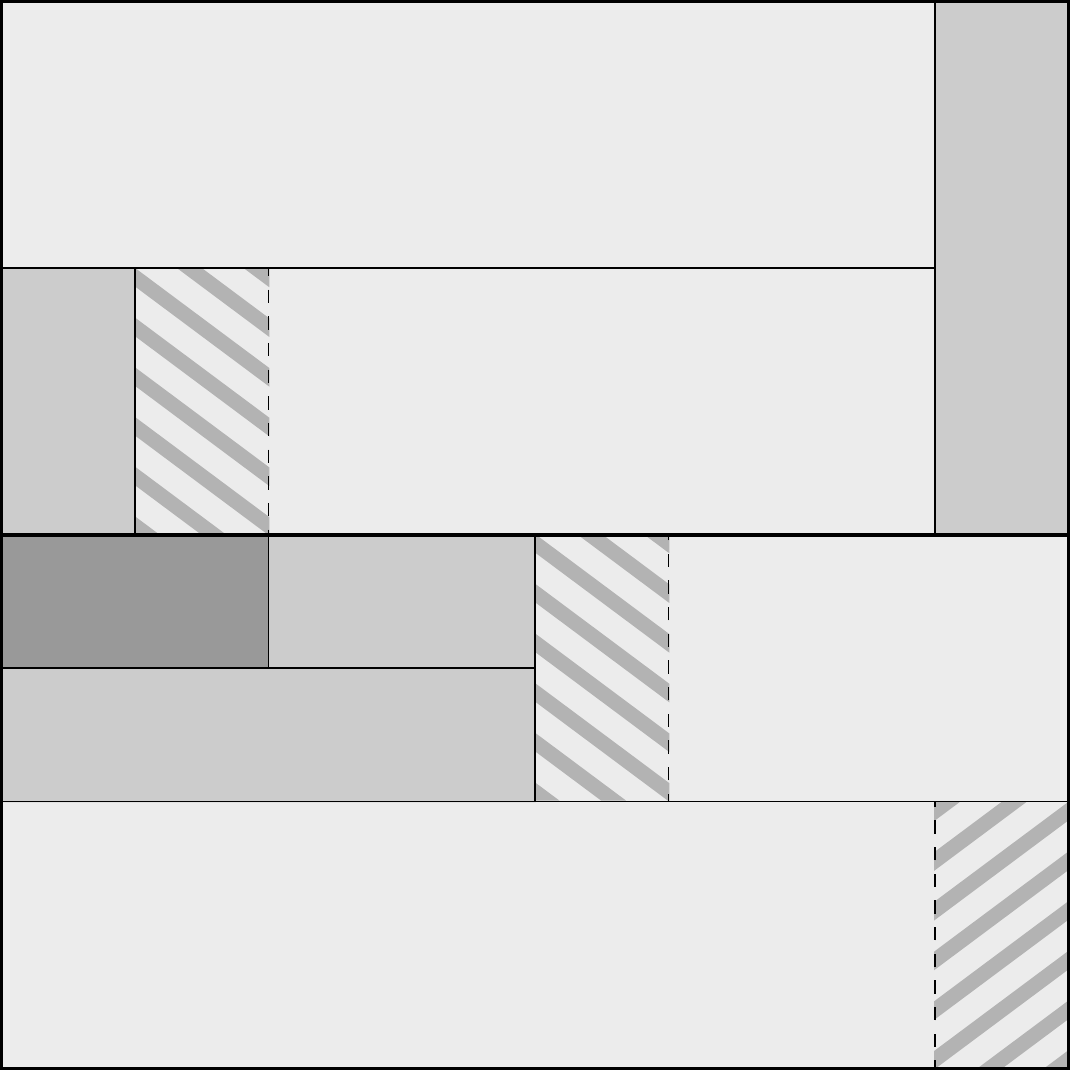}}%
    \put(0.45118105,0.8178801){\color[rgb]{0,0,0}\makebox(0,0)[b]{\smash{$M_4$}}}%
    \put(0.93613707,0.70074895){\color[rgb]{0,0,0}\makebox(0,0)[b]{\smash{$B_3$}}}%
    \put(0.57076299,0.58501004){\color[rgb]{0,0,0}\makebox(0,0)[b]{\smash{$M_3$}}}%
    \put(0.06606515,0.58186317){\color[rgb]{0,0,0}\makebox(0,0)[b]{\smash{$B_4$}}}%
    \put(0.1863991,0.58343663){\color[rgb]{0,0,0}\makebox(0,0)[b]{\smash{$E_3$}}}%
    \put(0.13014805,0.39147616){\color[rgb]{0,0,0}\makebox(0,0)[b]{\smash{$A$}}}%
    \put(0.37760752,0.39304963){\color[rgb]{0,0,0}\makebox(0,0)[b]{\smash{$B_2$}}}%
    \put(0.25802558,0.27032081){\color[rgb]{0,0,0}\makebox(0,0)[b]{\smash{$B_1$}}}%
    \put(0.56087934,0.31752418){\color[rgb]{0,0,0}\makebox(0,0)[b]{\smash{$E_2$}}}%
    \put(0.80205957,0.32067106){\color[rgb]{0,0,0}\makebox(0,0)[b]{\smash{$M_2$}}}%
    \put(0.45118105,0.08622756){\color[rgb]{0,0,0}\makebox(0,0)[b]{\smash{$M_1$}}}%
    \put(0.93535949,0.08780107){\color[rgb]{0,0,0}\makebox(0,0)[b]{\smash{$E_1$}}}%
  \end{picture}%
\endgroup%

%% file: 11_32_proof_a.tex
\begingroup%
  \makeatletter%
  \providecommand\color[2][]{%
    \errmessage{(Inkscape) Color is used for the text in Inkscape, but the package 'color.sty' is not loaded}%
    \renewcommand\color[2][]{}%
  }%
  \providecommand\transparent[1]{%
    \errmessage{(Inkscape) Transparency is used (non-zero) for the text in Inkscape, but the package 'transparent.sty' is not loaded}%
    \renewcommand\transparent[1]{}%
  }%
  \providecommand\rotatebox[2]{#2}%
  \ifx\svgwidth\undefined%
    \setlength{\unitlength}{304.8bp}%
    \ifx\svgscale\undefined%
      \relax%
    \else%
      \setlength{\unitlength}{\unitlength * \real{\svgscale}}%
    \fi%
  \else%
    \setlength{\unitlength}{\svgwidth}%
  \fi%
  \global\let\svgwidth\undefined%
  \global\let\svgscale\undefined%
  \makeatother%
  \begin{picture}(1,1)%
    \put(0,0){\includegraphics[width=\unitlength]{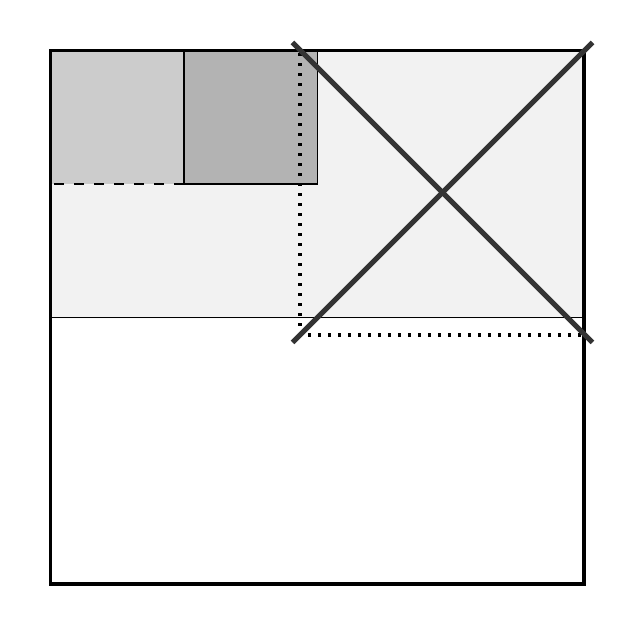}}%
    \put(0.36904895,0.78220898){\color[rgb]{0,0,0}\makebox(0,0)[b]{\smash{$Q_1$}}}%
    \put(0.69685039,0.80562044){\color[rgb]{0,0,0}\makebox(0,0)[b]{\smash{$Q_0$}}}%
  \end{picture}%
\endgroup%

%% file: 11_32_proof_b.tex
\begingroup%
  \makeatletter%
  \providecommand\color[2][]{%
    \errmessage{(Inkscape) Color is used for the text in Inkscape, but the package 'color.sty' is not loaded}%
    \renewcommand\color[2][]{}%
  }%
  \providecommand\transparent[1]{%
    \errmessage{(Inkscape) Transparency is used (non-zero) for the text in Inkscape, but the package 'transparent.sty' is not loaded}%
    \renewcommand\transparent[1]{}%
  }%
  \providecommand\rotatebox[2]{#2}%
  \ifx\svgwidth\undefined%
    \setlength{\unitlength}{321.72416992bp}%
    \ifx\svgscale\undefined%
      \relax%
    \else%
      \setlength{\unitlength}{\unitlength * \real{\svgscale}}%
    \fi%
  \else%
    \setlength{\unitlength}{\svgwidth}%
  \fi%
  \global\let\svgwidth\undefined%
  \global\let\svgscale\undefined%
  \makeatother%
  \begin{picture}(1,0.9473954)%
    \put(0,0){\includegraphics[width=\unitlength]{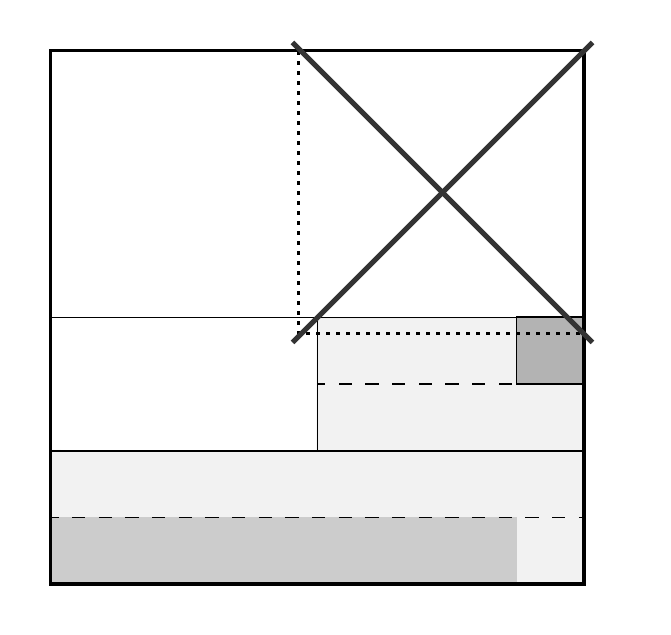}}%
    \put(0.67138242,0.77087853){\color[rgb]{0,0,0}\makebox(0,0)[b]{\smash{$Q_0$}}}%
    \put(0.92751458,0.40972628){\color[rgb]{0,0,0}\makebox(0,0)[b]{\smash{$Q_s$}}}%
  \end{picture}%
\endgroup%

%% file: 11_32_proof_c.tex
\begingroup%
  \makeatletter%
  \providecommand\color[2][]{%
    \errmessage{(Inkscape) Color is used for the text in Inkscape, but the package 'color.sty' is not loaded}%
    \renewcommand\color[2][]{}%
  }%
  \providecommand\transparent[1]{%
    \errmessage{(Inkscape) Transparency is used (non-zero) for the text in Inkscape, but the package 'transparent.sty' is not loaded}%
    \renewcommand\transparent[1]{}%
  }%
  \providecommand\rotatebox[2]{#2}%
  \ifx\svgwidth\undefined%
    \setlength{\unitlength}{304.8bp}%
    \ifx\svgscale\undefined%
      \relax%
    \else%
      \setlength{\unitlength}{\unitlength * \real{\svgscale}}%
    \fi%
  \else%
    \setlength{\unitlength}{\svgwidth}%
  \fi%
  \global\let\svgwidth\undefined%
  \global\let\svgscale\undefined%
  \makeatother%
  \begin{picture}(1,1)%
    \put(0,0){\includegraphics[width=\unitlength]{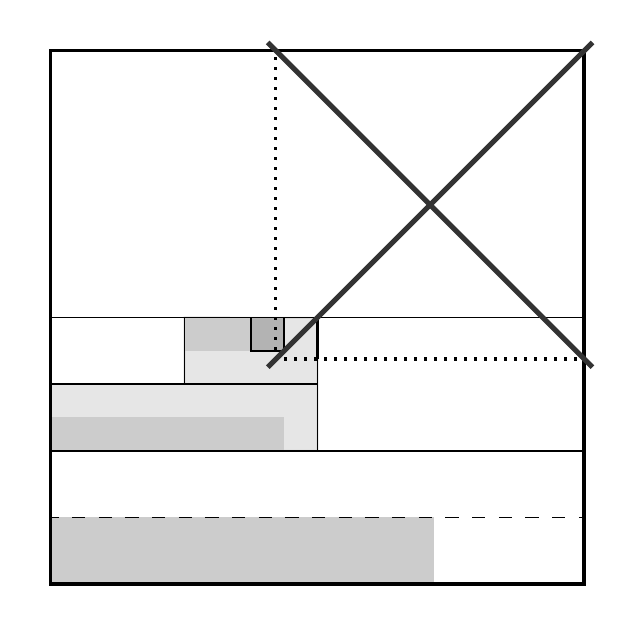}}%
    \put(0.667979,0.76512556){\color[rgb]{0,0,0}\makebox(0,0)[b]{\smash{$Q_0$}}}%
    \put(0.36049768,0.52051843){\color[rgb]{0,0,0}\makebox(0,0)[b]{\smash{$Q_s$}}}%
  \end{picture}%
\endgroup%

%% file: 11_32_proof_d.tex
\begingroup%
  \makeatletter%
  \providecommand\color[2][]{%
    \errmessage{(Inkscape) Color is used for the text in Inkscape, but the package 'color.sty' is not loaded}%
    \renewcommand\color[2][]{}%
  }%
  \providecommand\transparent[1]{%
    \errmessage{(Inkscape) Transparency is used (non-zero) for the text in Inkscape, but the package 'transparent.sty' is not loaded}%
    \renewcommand\transparent[1]{}%
  }%
  \providecommand\rotatebox[2]{#2}%
  \ifx\svgwidth\undefined%
    \setlength{\unitlength}{304.8bp}%
    \ifx\svgscale\undefined%
      \relax%
    \else%
      \setlength{\unitlength}{\unitlength * \real{\svgscale}}%
    \fi%
  \else%
    \setlength{\unitlength}{\svgwidth}%
  \fi%
  \global\let\svgwidth\undefined%
  \global\let\svgscale\undefined%
  \makeatother%
  \begin{picture}(1,1)%
    \put(0,0){\includegraphics[width=\unitlength]{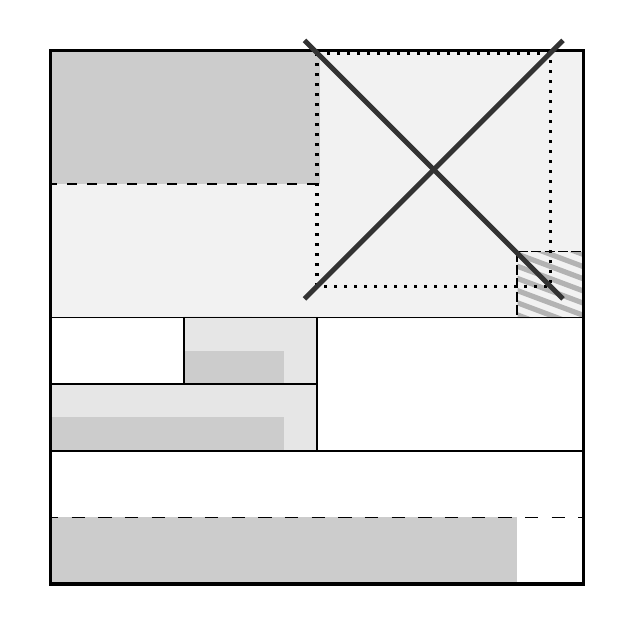}}%
    \put(0.68454724,0.85362103){\color[rgb]{0,0,0}\makebox(0,0)[b]{\smash{$Q_1$}}}%
  \end{picture}%
\endgroup%

%% file: 11_32_proof_e.tex
\begingroup%
  \makeatletter%
  \providecommand\color[2][]{%
    \errmessage{(Inkscape) Color is used for the text in Inkscape, but the package 'color.sty' is not loaded}%
    \renewcommand\color[2][]{}%
  }%
  \providecommand\transparent[1]{%
    \errmessage{(Inkscape) Transparency is used (non-zero) for the text in Inkscape, but the package 'transparent.sty' is not loaded}%
    \renewcommand\transparent[1]{}%
  }%
  \providecommand\rotatebox[2]{#2}%
  \ifx\svgwidth\undefined%
    \setlength{\unitlength}{316.75166016bp}%
    \ifx\svgscale\undefined%
      \relax%
    \else%
      \setlength{\unitlength}{\unitlength * \real{\svgscale}}%
    \fi%
  \else%
    \setlength{\unitlength}{\svgwidth}%
  \fi%
  \global\let\svgwidth\undefined%
  \global\let\svgscale\undefined%
  \makeatother%
  \begin{picture}(1,0.96226804)%
    \put(0,0){\includegraphics[width=\unitlength]{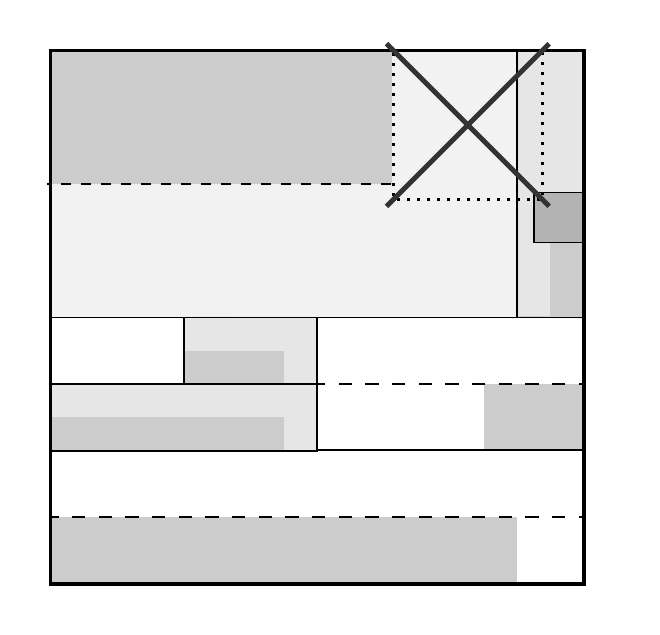}}%
    \put(0.72831495,0.83475562){\color[rgb]{0,0,0}\makebox(0,0)[b]{\smash{$Q_1$}}}%
    \put(0.92637667,0.64039195){\color[rgb]{0,0,0}\makebox(0,0)[b]{\smash{$Q_s$}}}%
  \end{picture}%
\endgroup%

%% file: 11_32_proof_f.tex
\begingroup%
  \makeatletter%
  \providecommand\color[2][]{%
    \errmessage{(Inkscape) Color is used for the text in Inkscape, but the package 'color.sty' is not loaded}%
    \renewcommand\color[2][]{}%
  }%
  \providecommand\transparent[1]{%
    \errmessage{(Inkscape) Transparency is used (non-zero) for the text in Inkscape, but the package 'transparent.sty' is not loaded}%
    \renewcommand\transparent[1]{}%
  }%
  \providecommand\rotatebox[2]{#2}%
  \ifx\svgwidth\undefined%
    \setlength{\unitlength}{313.32421875bp}%
    \ifx\svgscale\undefined%
      \relax%
    \else%
      \setlength{\unitlength}{\unitlength * \real{\svgscale}}%
    \fi%
  \else%
    \setlength{\unitlength}{\svgwidth}%
  \fi%
  \global\let\svgwidth\undefined%
  \global\let\svgscale\undefined%
  \makeatother%
  \begin{picture}(1,0.97279426)%
    \put(0,0){\includegraphics[width=\unitlength]{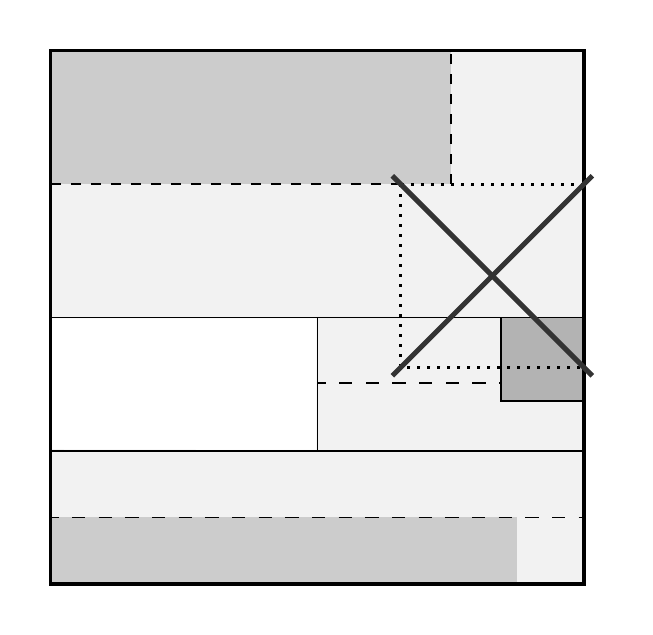}}%
    \put(0.75194923,0.63901491){\color[rgb]{0,0,0}\makebox(0,0)[b]{\smash{$Q_1$}}}%
    \put(0.92557131,0.42709385){\color[rgb]{0,0,0}\makebox(0,0)[b]{\smash{$Q_s$}}}%
  \end{picture}%
\endgroup%

%% file: 11_32_proof_g.tex
\begingroup%
  \makeatletter%
  \providecommand\color[2][]{%
    \errmessage{(Inkscape) Color is used for the text in Inkscape, but the package 'color.sty' is not loaded}%
    \renewcommand\color[2][]{}%
  }%
  \providecommand\transparent[1]{%
    \errmessage{(Inkscape) Transparency is used (non-zero) for the text in Inkscape, but the package 'transparent.sty' is not loaded}%
    \renewcommand\transparent[1]{}%
  }%
  \providecommand\rotatebox[2]{#2}%
  \ifx\svgwidth\undefined%
    \setlength{\unitlength}{304.8bp}%
    \ifx\svgscale\undefined%
      \relax%
    \else%
      \setlength{\unitlength}{\unitlength * \real{\svgscale}}%
    \fi%
  \else%
    \setlength{\unitlength}{\svgwidth}%
  \fi%
  \global\let\svgwidth\undefined%
  \global\let\svgscale\undefined%
  \makeatother%
  \begin{picture}(1,1)%
    \put(0,0){\includegraphics[width=\unitlength]{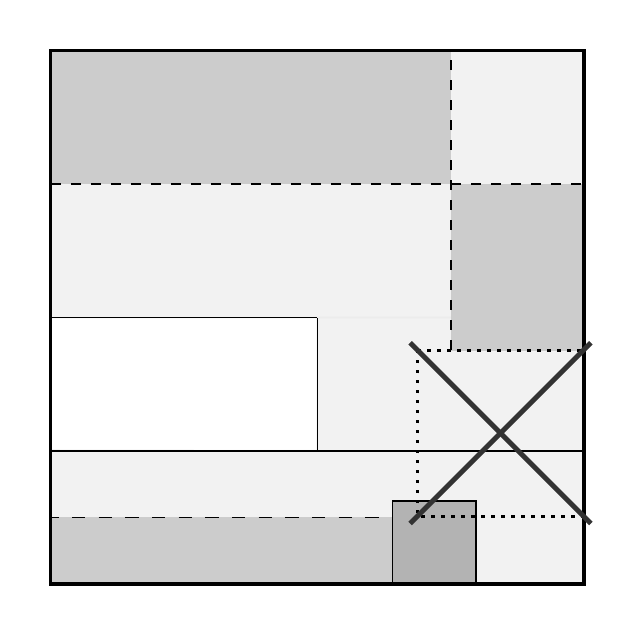}}%
    \put(0.78289154,0.40098047){\color[rgb]{0,0,0}\makebox(0,0)[b]{\smash{$Q_1$}}}%
    \put(0.6936527,0.02040304){\color[rgb]{0,0,0}\makebox(0,0)[b]{\smash{$Q_s$}}}%
  \end{picture}%
\endgroup%

%% file: dynamicContainer.tex
\section{Packing into a Dynamic Container}\label{sec:dynamicContainer}
Now we discuss the problem of online packing a sequence of squares into a dynamic square container. At each stage, the container must be large enough
to accommodate all objects; this requires keeping the container tight early
on, but may require increasing its edge length
appropriately during the process.

In the following, we give a non-trivial family of instances,
which prove that no online algorithm can maintain a packing density greater than
$3/7$ for an arbitrary input sequence of squares and introduce an online square packing
algorithm that maintains a packing density of $1/8$ for an arbitrarily input
sequence of squares.

\subsection{An Upper Bound on $\delta$}\label{sec:dynContUpperBound}
If the total area of the given sequence is unknown in advance, the problem
of finding a dense online packing becomes harder. As it turns
out, a density of $1/2$ cannot be achieved.

  \begin{figure}[ht]
    \centering
    \subfigure[]{
      \def\svgwidth{.42\textwidth}
      \import{./}{cornerP1_2.tex}
      \label{fig:cornerP1_2}
    }
    \hskip 0.1cm
    \subfigure[]{
      \def\svgwidth{.42\textwidth}
      \import{./}{cornerP2_2.tex}
      \label{fig:cornerP2_2}
    }
    \hskip 0.5cm
    \subfigure[]{
      \def\svgwidth{.45\textwidth}
      \import{./}{delta1.tex}
      \label{fig:delta1}
    }
    \hskip 0.05cm
    \subfigure[]{
      \def\svgwidth{.45\textwidth}
      \import{./}{delta2.tex}
      \label{fig:delta2}
    }
    \hskip 0.5cm
    \subfigure[]{
	\def\svgwidth{.50\textwidth}
	\import{./}{qe2.tex}
    }
    \caption{\label{fig:lowdyn}Different choices in the lower-bound sequence:
(a) Packing after choosing corner positions. (b) Packing after choosing a center position.
(c) Recursion parameters for corner positions. (d) Recursion parameters for center position.
(e) Packing a last square.}
 \end{figure}
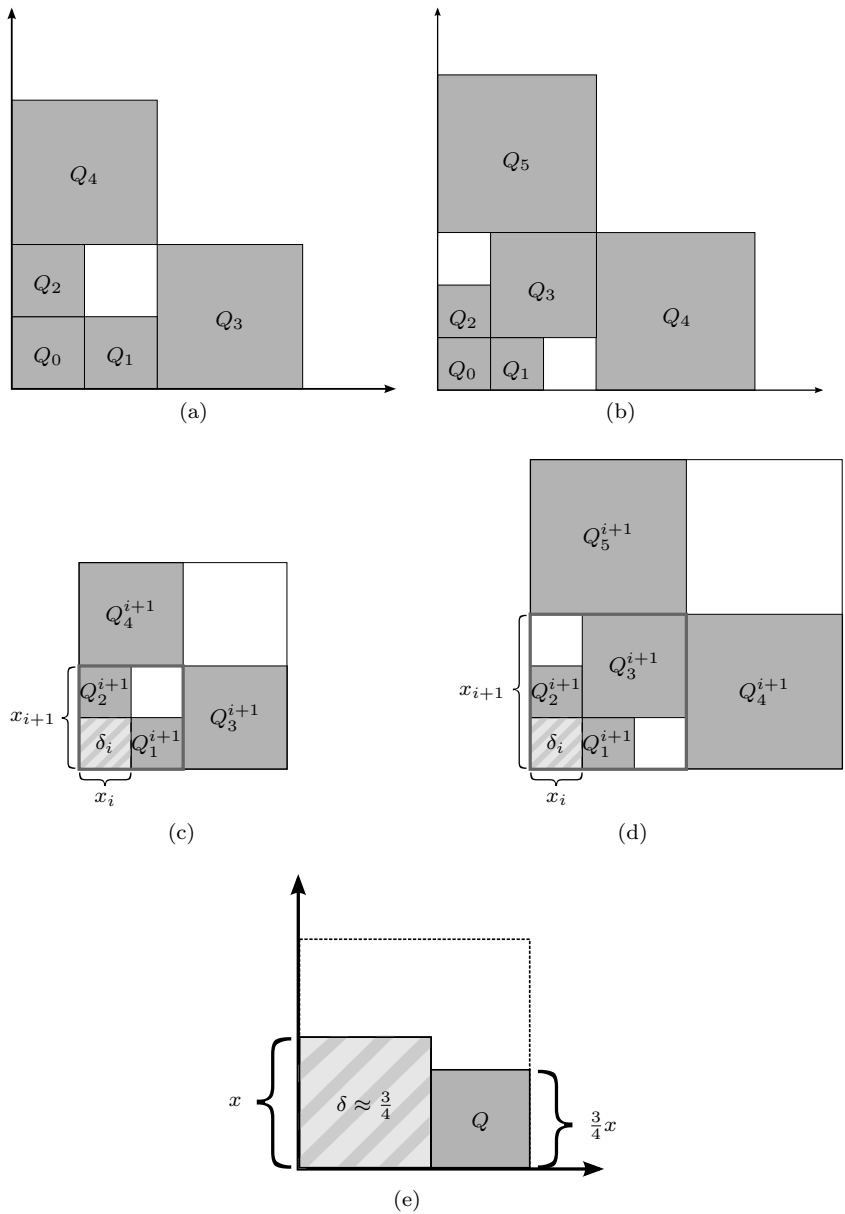

\begin{theorem}
\label{th:3-7}
    There are sequences for which no deterministic online packing algorithm
    can maintain a density strictly greater than $3/7 \approx 0.4286$. 
\end{theorem}

\begin{proof}
We construct an appropriate sequence of squares, depending
on what choices a deterministic player makes; 
see Fig.~\ref{fig:lowdyn}. At each stage,
the player must place a square $Q_3$ into a corner position
(Fig.~\ref{fig:lowdyn}(a)) or into a center position (Fig.~\ref{fig:lowdyn}(b));
the adversary responds by either requesting another square of the same size (a),
or two of the size of the current spanning box.
This is repeated. 

If the player keeps choosing corner positions, 
the density $\delta_i$ for the enclosing square of size $x_i$ satisfies
the recursion $\delta_{i+1}=1/4\cdot\delta_i+1/2$, as shown in Fig.~\ref{fig:lowdyn}(d).
The sequence is decreasing and bounded from below, so solving the equation
$\delta_{\infty}=1/4\cdot\delta_{\infty}+1/2$ yields $\lim_{i\to \infty} \delta_i = \delta_{\infty} = 2/3$.
If the player keeps choosing center positions, 
the density $\delta_i$ for the enclosing square of size $x_i$ satisfies
the recursion $\delta_{i+1}=1/9\cdot\delta_i+2/3$, as shown in Fig.~\ref{fig:lowdyn}(e).
This sequence is also decreasing and bounded from below, so solving the equation
$\delta_{\infty}=1/9\cdot\delta_{\infty}+2/3$ yields $\lim_{i\to \infty} \delta_i = \delta_{\infty} = 3/4$.
For mixed choices, the density lies in between. Therefore the opponent can force the density 
below $3/4+\varepsilon$, for any $\varepsilon>0$. Once that is the case, 
with the center position occupied, the adversary can request
a final square of size $3/4\cdot x$, where $x$ is the size of the current spanning box.
The resulting density is arbitrarily close to
$\frac{3/4\cdot x^2+(3/4\cdot x)^2}{(x+ 3/4\cdot x)^2}=3/7$.
If the center position does not get occupied, the density is even worse. \qed
\end{proof}

It is an easy consequence of continuity that
this upper bound can be lowered by a very small amount by slightly decreasing the value for the center case,
while increasing the value for the corner case, until they are balanced. 
More specifically, we can decrease the density for the center case by increasing the square sizes by more than a factor of 2
at each step of the recursion. When only focusing on the center case, the best such factor is $1+\sqrt{3}$, for an asymptotic density of $\sqrt{3}-1=0.73204\ldots$,
yielding a resulting final density of $0.42265\ldots$ as a lower bound for the achievable value. However, this is much beyond what can actually be achieved when
also accounting for the corner case: the bounding box for each iteration
becomes a rectangle, so the worst-case density for the corner case increases quite rapidly. This 
keeps the total upper bound for the final density much closer to $3/7=0.42857\ldots$. 
As a consequence, we omit the tedious computations for the resulting tiny improvement.

\subsection{A Lower Bound on $\delta$}\label{sec:dynContLowerBounds} 
When placing squares into a dynamic container, we cannot use our 
Recursive Shelf Algorithm, as it requires allocating shelves from all four
container boundaries, which are not known in advance.
However, we can adapt the 
Brick Algorithm by~\cite{jl97}, which we describe in the following.
 
The method is based on a partition of the unit square into \emph{bricks}.
Bricks are rectangles with aspect ratio $\sqrt{2}$ (or $1/\sqrt{2}$),
which are well-known from the international ISO 216 paper formats, in particular the common A series.
The most important property of these rectangles is that by bisecting a brick
with dimensions $(b, b/\sqrt{2})$, we create two new smaller bricks of size
$(b/2, b/\sqrt{2})$.  This way, we can construct bricks with side lengths
$b2^{-k/2}$ and $b2^{(-k-1)/2}$ for any $k = 0, 1, \ldots$ via a recursive
bisection.  All of the bricks created this way are called \emph{subbricks} of
$B$.  For any square $Q$ let $S_b(Q)$ denote the smallest brick with side
lengths $b/(\sqrt2)^{k}$ and $b/(\sqrt2)^{k+1}$that may contain $Q$.
Obviously, there is some space left if $Q$ is packed into $S_b(Q)$.
Independent of the base $b$ side length we can bound this free space as follows.
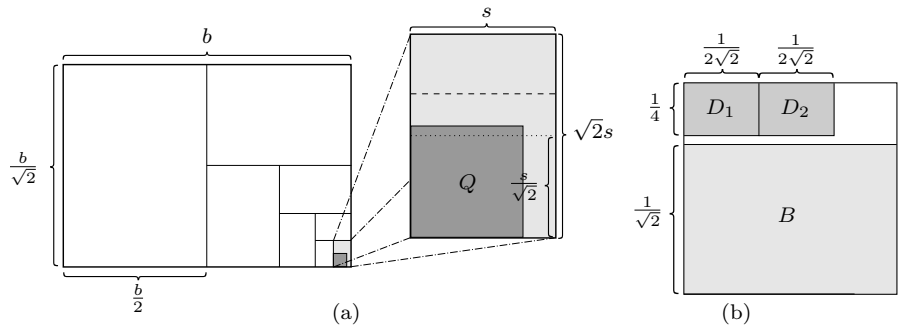
\begin{figure}
  \centering
  \subfigure[]{
	\hspace{-2cm}
	\def\svgwidth{.75\textwidth}
	\import{./}{brickConcept.tex}
	 \label{fig:brickConcept}
	}
	\hspace{-1cm}
  \subfigure[]{
  	\def\svgwidth{.33\textwidth}
	\import{./}{brickPartition.tex}
 	 \label{fig:brickPartition}
	}
  \caption[Subdivision Scheme of the Brick Concept.]{(a) Subdivision Scheme of the Brick Concept:
	    (Left) The recursive bisection of a brick $B$.
	    (Right) A brick (equal to $S_B(Q)$) occupied by a square $Q$; the dashed line marks the upper bound, the dotted line the lower bound on the possible side lengths of $Q$.
	(b) Partition of the unit square \unitsq used by Januszewski and Lassak.
}
\end{figure}

\begin{lem}\label{lem:brickdensity}
 Let $Q$ be a square and $b$ be a real number. Then $\frac{\norm{S_b(Q)}}{2\sqrt2} < \norm{Q} \leq \frac{\norm{S_b(Q)}}{\sqrt2}$.
\end{lem}
\begin{proof}
  Let $s$ and $\sqrt{2}s$ be the side lengths of $S_B(Q)$.
  By definition of $S_b(Q)$, we have
	$$\frac{s}{\sqrt{2}} = \frac{b}{(\sqrt2)^{k+1}} < x \leq \frac{b}{(\sqrt2)^k} =: s$$
  for some $k \in \mathbb{N}$. With $\norm{S_b(Q)} = s \cdot \sqrt{2}s = \sqrt{2}s^2$ we get
    $$\frac{\norm{S_b(Q)}}{2\sqrt2} =  \frac{\sqrt{2}s^2}{2\sqrt2}= \left(\frac{s}{\sqrt{2}}\right)^2 < x^2 = \norm{Q}$$
    $$\text{ and } \qquad\norm{Q} = x^2 \leq s^2 = \frac{\norm{S_b(Q)}}{\sqrt2}$$
\end{proof}
In other words, with a strategy that packs squares into their respective smallest subbrick,
we cannot hope to generate a packing density higher than $1/(2\sqrt2)$.
We denote the bricks that contain a square \emph{occupied}. All other bricks are called \emph{free}.
%
Based on this subdivision, Januszewski and Lassak developed a recursive packing algorithm, which they call \emph{the method of the first free fitting subbrick}.
They first construct three bricks in the unit square as shown in Fig.~\ref{fig:brickPartition}.
Then they pack each square $Q$ into a brick congruent to $S_1(Q)$ after recursively subdividing the smallest free brick that can accomodate $Q$. 

\vskip 1em
\begin{minipage}{\textwidth}
  \textbf{The Brick Algorithm:}
  \begin{enumerate}
    \item Construct the bricks 
	\begin{eqnarray*}
	  B\,	&=&	\{(x_1,x_2): 0 \leq x_1 \leq 1,\,0 \leq x_2 \leq \sqrt{2}\},\\
	  D_1	&=&	\{(x_1,x_2): 0 \leq x_1 \leq 2^{-3/2}, 3/4 \leq x_2 \leq 1\}\quad\text{and}\\
	  D_2	&=&	\{(x_1,x_2): 2^{-3/2} \leq x_1 \leq 2\cdot2^{-3/2}, 3/4 \leq x_2 \leq 1\}
	\end{eqnarray*}
    \item For each incoming square $Q$:
    \begin{enumerate}
	\item Let \newbuffer be the first base brick in the order $D_1$, $D_2$, $B$ that has a free subbrick $\newbuffer'$ of size greater or equal to $S_1(Q)$.
	\item If $\norm{\newbuffer'} = \norm{S_1(Q)}$, then pack $Q$ into $\newbuffer'$.
	\item Otherwise, recursively bisect (one half of) the smallest subbrick of $\newbuffer'$ until a brick $\newbuffer''$ of size $S_1(Q)$ is created.
	\item Pack $Q$ into $\newbuffer''$.
    \end{enumerate}
  \end{enumerate}
\end{minipage}
\vskip 1em

We call the bricks $B$, $D_1$ and $D_2$ the \emph{base bricks}.
Note that all three base bricks have side length equal to a power of $\sqrt{2}$.
That is, all (sub-)bricks created by the algorithm have only side lengths equal to a power of $\sqrt{2}$, too.

In order to adapt this approach to our setting (with increasing instead of decreasing brick size), we keep
some properties, but adjust others.
We still consider bricks with side lengths equal to a power of $\sqrt{2}$
(and aspect ratio $1/\sqrt{2}$ or $\sqrt{2}$).
We let $B_k$ denote the brick of size $({\sqrt{2}}^k,{\sqrt{2}}^{k+1})$
and let $S(Q)$ denote the smallest brick $B_i$ that may contain a given square $Q$.

There are two crucial modifications: 
(1) The first square $Q$ is packed into a brick of size $S(Q)$ with its lower left
corner in the origin and (2) instead of always subdividing the
existing bricks (starting with three fixed ones), we may repeatedly
double the current maximum existing brick \maxBrick to make room for large
incoming squares. Apart from that, we keep the same packing scheme:
Place each square $Q$ into (a subbrick of) the smallest free brick
that can contain $Q$; see Fig.~\ref{fig:dynBrickCases} for an
illustration.

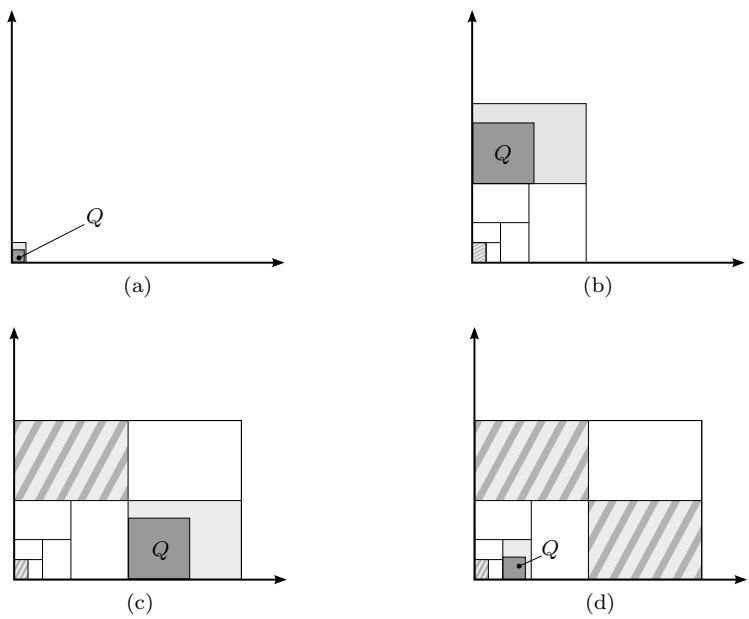
\begin{figure}[t]
  \centering
  \hskip 5mm
  \subfigure[]{
    \def\svgwidth{.30\textwidth}
    \import{./}{dynamicBrick1.tex}
    \label{fig:dynBrick1}
  }
  \hskip 2cm
  \subfigure[]{
    \def\svgwidth{.30\textwidth}
    \import{./}{dynamicBrick2.tex}
    \label{fig:dynBrick2}
  }
  \newline
  \subfigure[]{
    \def\svgwidth{.30\textwidth}
    \import{./}{dynamicBrick3.tex}
    \label{fig:dynBrick3}
  }
  \hskip 2cm
  \subfigure[]{
    \def\svgwidth{.30\textwidth}
    \import{./}{dynamicBrick4.tex}
    \label{fig:dynBrick4}
  }
  \caption{The modified Brick-Packing algorithm
for an input square $Q$. Occupied bricks are hatched, 
free bricks are blank. (a) A first square gets placed
          into the lower left corner, $\maxBrick = S(Q)$.
(b) If $S(Q)>$\maxBrick, 
            we double \maxBrick until $Q$ fits. (c) If $Q$ does not fit into \maxBrick, but $\norm{S(Q)} < \norm{\maxBrick}$, we double
$\maxBrick$ and subdivide the resulting brick.
(d) If $Q$ fits into \maxBrick,
            we pack it into the smallest free fitting subbrick.}
  \label{fig:dynBrickCases}
\end{figure}

\begin{theorem}
\label{th:1-8}
  For any input sequence of squares, the Dynamic Brick Algorithm maintains
  a packing density of at least $1/8$.
\end{theorem}
\begin{proof}
  By construction, every occupied brick has a density of at least $1/(2\sqrt2)$.
  It is easy to see that in every step of the algorithm
  at most half the area of \maxBrick consists of free bricks; compare \cite{jl97}.
  Because \maxBrick always contains all occupied bricks (and thus all packed squares),
  the ratio of \norm{\maxBrick} to the area of the smallest enclosing square is at least
  $1/\sqrt2$. Therefore, the algorithm maintains an overall density of at least
  $(1/(2\sqrt2))\cdot(1/2)\cdot(1/\sqrt2)=1/8$. \qed
\end{proof}

\subsection{Minimizing Container Size}
The above results consider the worst-case ratio for the packing density.
A closely related question is the online optimization problem of maintaining
a square container with minimum edge length. The following is an easy consequence
of Theorem~\ref{th:1-8}, as a square of edge length $2\sqrt{2}$ can accommodate
a unit area when packed with density 1/8.
By considering optimal offline packings for the class of examples constructed
in Theorem~\ref{th:3-7}, it is straightforward to get a lower bound of 4/3 for any
deterministic online algorithm.

\begin{corollary}\label{cor:ub-edgelength}
  Dynamic Brick Packing provides a competitive factor of $2\sqrt{2}=2.82\ldots$
  for packing an online sequence of squares into a square container with small
  edge length. The same problem has a lower bound of 4/3 for the competitive factor.
\end{corollary}

%% file: cornerP1_2.tex
\begingroup%
  \makeatletter%
  \providecommand\color[2][]{%
    \errmessage{(Inkscape) Color is used for the text in Inkscape, but the package 'color.sty' is not loaded}%
    \renewcommand\color[2][]{}%
  }%
  \providecommand\transparent[1]{%
    \errmessage{(Inkscape) Transparency is used (non-zero) for the text in Inkscape, but the package 'transparent.sty' is not loaded}%
    \renewcommand\transparent[1]{}%
  }%
  \providecommand\rotatebox[2]{#2}%
  \ifx\svgwidth\undefined%
    \setlength{\unitlength}{342.57272169bp}%
    \ifx\svgscale\undefined%
      \relax%
    \else%
      \setlength{\unitlength}{\unitlength * \real{\svgscale}}%
    \fi%
  \else%
    \setlength{\unitlength}{\svgwidth}%
  \fi%
  \global\let\svgwidth\undefined%
  \global\let\svgscale\undefined%
  \makeatother%
  \begin{picture}(1,0.99998008)%
    \put(0,0){\includegraphics[width=\unitlength]{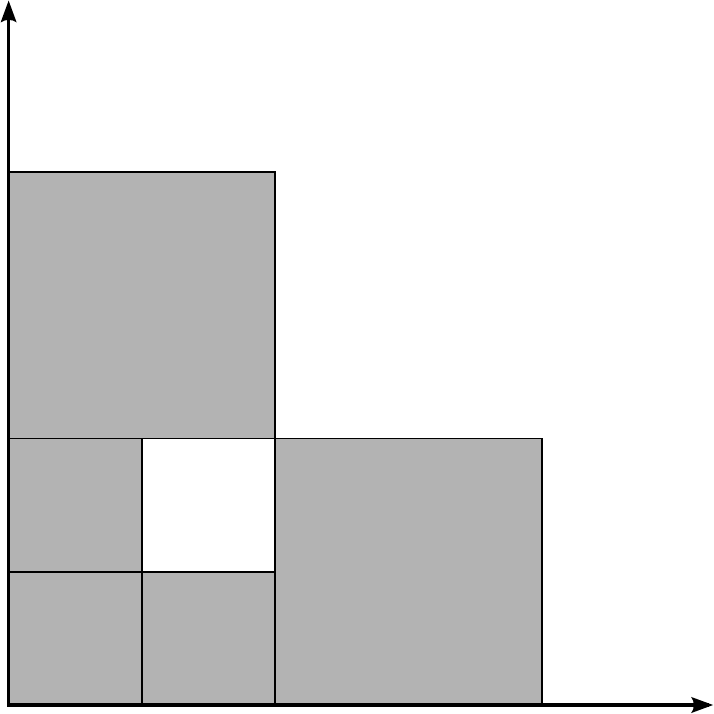}}%
    \put(0.10554529,0.08217266){\color[rgb]{0,0,0}\makebox(0,0)[b]{\smash{$Q_0$}}}%
    \put(0.10554529,0.2689943){\color[rgb]{0,0,0}\makebox(0,0)[b]{\smash{$Q_2$}}}%
    \put(0.29236694,0.08217266){\color[rgb]{0,0,0}\makebox(0,0)[b]{\smash{$Q_1$}}}%
    \put(0.57259941,0.17558348){\color[rgb]{0,0,0}\makebox(0,0)[b]{\smash{$Q_3$}}}%
    \put(0.19895611,0.54922677){\color[rgb]{0,0,0}\makebox(0,0)[b]{\smash{$Q_4$}}}%
  \end{picture}%
\endgroup%

%% file: cornerP2_2.tex
\begingroup%
  \makeatletter%
  \providecommand\color[2][]{%
    \errmessage{(Inkscape) Color is used for the text in Inkscape, but the package 'color.sty' is not loaded}%
    \renewcommand\color[2][]{}%
  }%
  \providecommand\transparent[1]{%
    \errmessage{(Inkscape) Transparency is used (non-zero) for the text in Inkscape, but the package 'transparent.sty' is not loaded}%
    \renewcommand\transparent[1]{}%
  }%
  \providecommand\rotatebox[2]{#2}%
  \ifx\svgwidth\undefined%
    \setlength{\unitlength}{470.57272169bp}%
    \ifx\svgscale\undefined%
      \relax%
    \else%
      \setlength{\unitlength}{\unitlength * \real{\svgscale}}%
    \fi%
  \else%
    \setlength{\unitlength}{\svgwidth}%
  \fi%
  \global\let\svgwidth\undefined%
  \global\let\svgscale\undefined%
  \makeatother%
  \begin{picture}(1,0.9999855)%
    \put(0,0){\includegraphics[width=\unitlength]{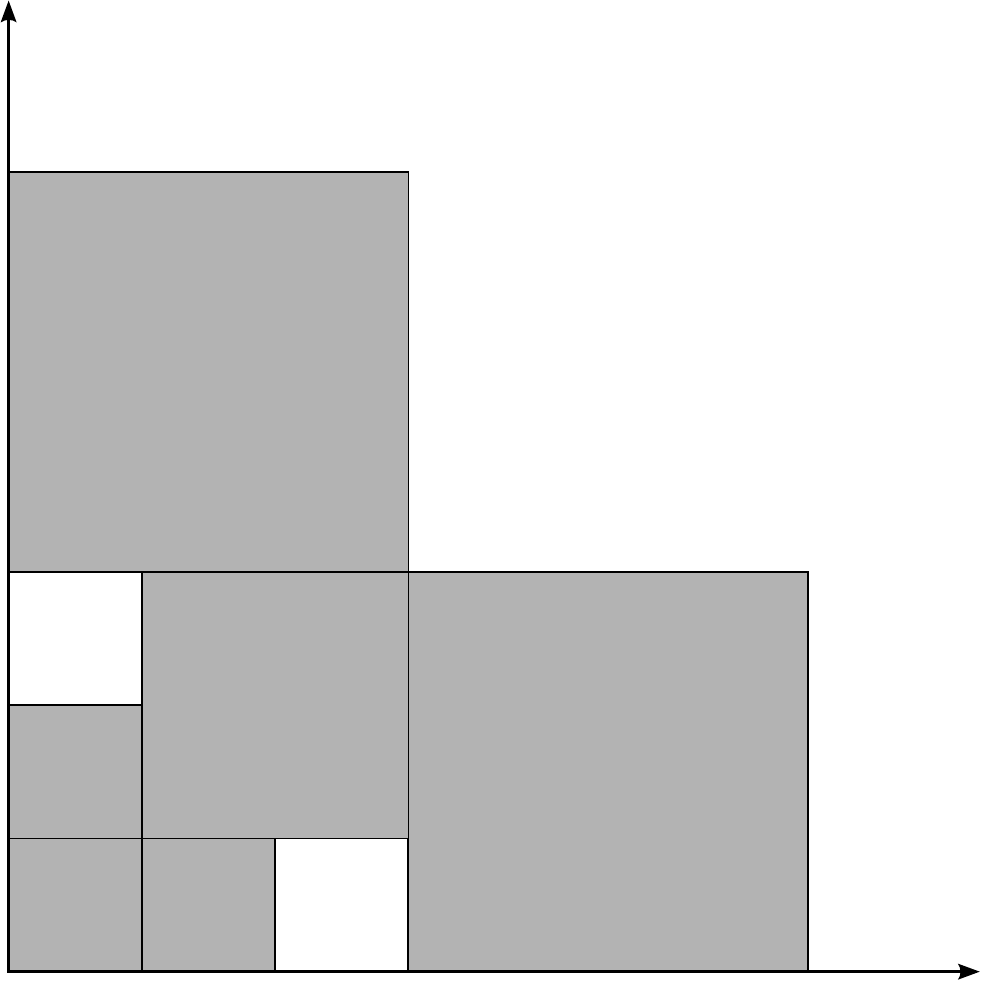}}%
    \put(0.27742547,0.24925565){\color[rgb]{0,0,0}\makebox(0,0)[b]{\smash{$Q_3$}}}%
    \put(0.07546911,0.1771284){\color[rgb]{0,0,0}\makebox(0,0)[b]{\smash{$Q_2$}}}%
    \put(0.07306487,0.04489504){\color[rgb]{0,0,0}\makebox(0,0)[b]{\smash{$Q_0$}}}%
    \put(0.21251093,0.04489504){\color[rgb]{0,0,0}\makebox(0,0)[b]{\smash{$Q_1$}}}%
    \put(0.62363637,0.18193689){\color[rgb]{0,0,0}\makebox(0,0)[b]{\smash{$Q_4$}}}%
    \put(0.21251093,0.58584957){\color[rgb]{0,0,0}\makebox(0,0)[b]{\smash{$Q_5$}}}%
  \end{picture}%
\endgroup%

%% file: delta1.tex
\begingroup%
  \makeatletter%
  \providecommand\color[2][]{%
    \errmessage{(Inkscape) Color is used for the text in Inkscape, but the package 'color.sty' is not loaded}%
    \renewcommand\color[2][]{}%
  }%
  \providecommand\transparent[1]{%
    \errmessage{(Inkscape) Transparency is used (non-zero) for the text in Inkscape, but the package 'transparent.sty' is not loaded}%
    \renewcommand\transparent[1]{}%
  }%
  \providecommand\rotatebox[2]{#2}%
  \ifx\svgwidth\undefined%
    \setlength{\unitlength}{512.8bp}%
    \ifx\svgscale\undefined%
      \relax%
    \else%
      \setlength{\unitlength}{\unitlength * \real{\svgscale}}%
    \fi%
  \else%
    \setlength{\unitlength}{\svgwidth}%
  \fi%
  \global\let\svgwidth\undefined%
  \global\let\svgscale\undefined%
  \makeatother%
  \begin{picture}(1,0.87441498)%
    \put(0,0){\includegraphics[width=\unitlength]{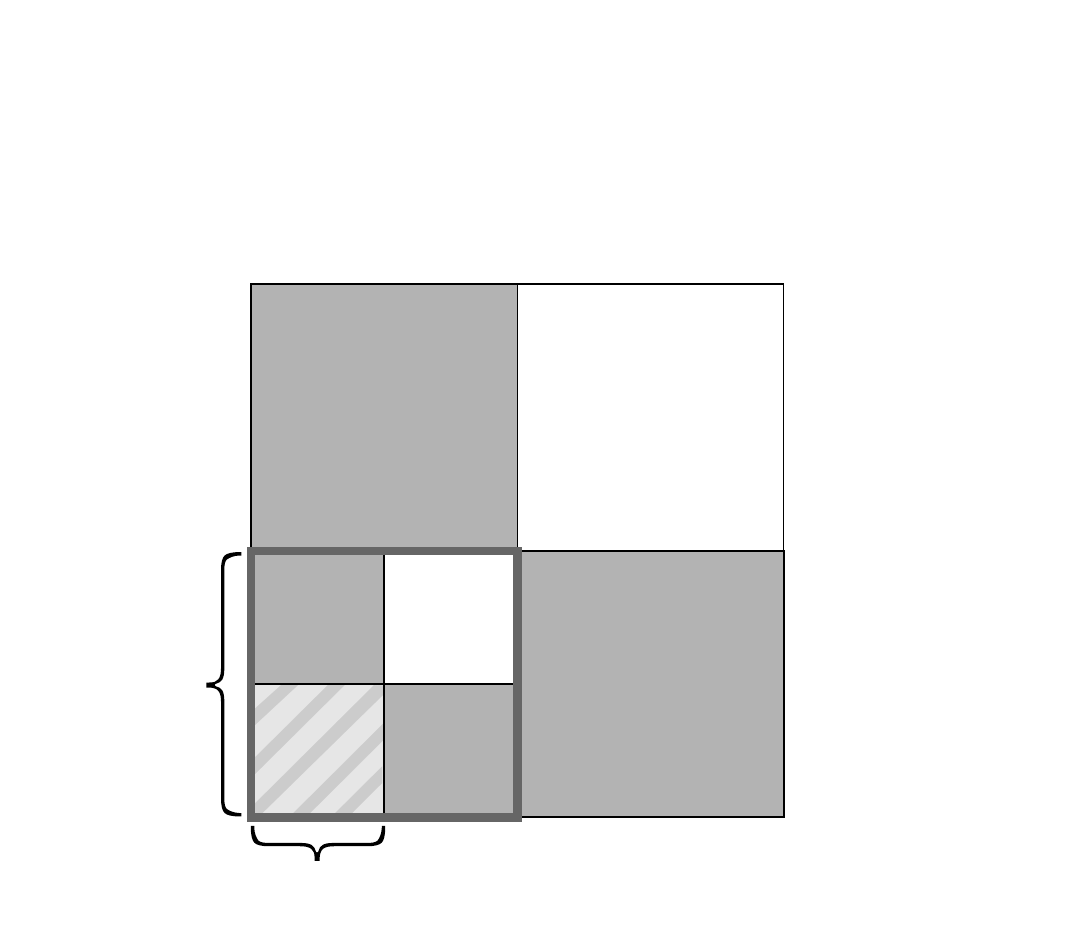}}%
    \put(0.29607756,0.15714056){\color[rgb]{0,0,0}\makebox(0,0)[b]{\smash{$\delta_i$}}}%
    \put(0.42088254,0.15491189){\color[rgb]{0,0,0}\makebox(0,0)[b]{\smash{$Q_1^{i+1}$}}}%
    \put(0.29719189,0.28083123){\color[rgb]{0,0,0}\makebox(0,0)[b]{\smash{$Q_2^{i+1}$}}}%
    \put(0.60920437,0.21842873){\color[rgb]{0,0,0}\makebox(0,0)[b]{\smash{$Q_3^{i+1}$}}}%
    \put(0.35959438,0.46803872){\color[rgb]{0,0,0}\makebox(0,0)[b]{\smash{$Q_4^{i+1}$}}}%
    \put(0.17963004,0.22622904){\color[rgb]{0,0,0}\makebox(0,0)[rb]{\smash{$x_{i+1}$}}}%
    \put(0.29760323,0.02935648){\color[rgb]{0,0,0}\makebox(0,0)[b]{\smash{$x_{i}$}}}%
  \end{picture}%
\endgroup%

%% file: delta2.tex
\begingroup%
  \makeatletter%
  \providecommand\color[2][]{%
    \errmessage{(Inkscape) Color is used for the text in Inkscape, but the package 'color.sty' is not loaded}%
    \renewcommand\color[2][]{}%
  }%
  \providecommand\transparent[1]{%
    \errmessage{(Inkscape) Transparency is used (non-zero) for the text in Inkscape, but the package 'transparent.sty' is not loaded}%
    \renewcommand\transparent[1]{}%
  }%
  \providecommand\rotatebox[2]{#2}%
  \ifx\svgwidth\undefined%
    \setlength{\unitlength}{512.8bp}%
    \ifx\svgscale\undefined%
      \relax%
    \else%
      \setlength{\unitlength}{\unitlength * \real{\svgscale}}%
    \fi%
  \else%
    \setlength{\unitlength}{\svgwidth}%
  \fi%
  \global\let\svgwidth\undefined%
  \global\let\svgscale\undefined%
  \makeatother%
  \begin{picture}(1,0.87441498)%
    \put(0,0){\includegraphics[width=\unitlength]{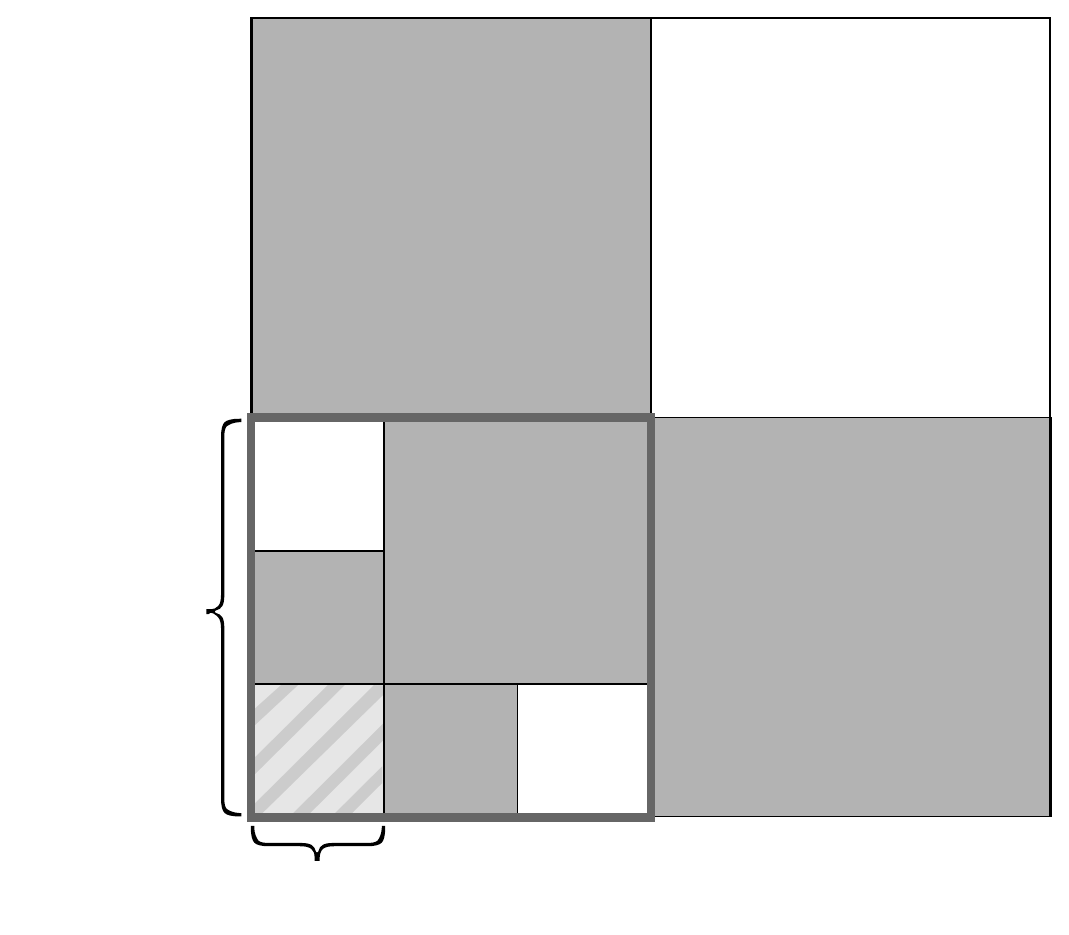}}%
    \put(0.29607756,0.15714056){\color[rgb]{0,0,0}\makebox(0,0)[b]{\smash{$\delta_i$}}}%
    \put(0.42088254,0.15491189){\color[rgb]{0,0,0}\makebox(0,0)[b]{\smash{$Q_1^{i+1}$}}}%
    \put(0.29719189,0.28083123){\color[rgb]{0,0,0}\makebox(0,0)[b]{\smash{$Q_2^{i+1}$}}}%
    \put(0.48439938,0.34323372){\color[rgb]{0,0,0}\makebox(0,0)[b]{\smash{$Q_3^{i+1}$}}}%
    \put(0.79641186,0.28083123){\color[rgb]{0,0,0}\makebox(0,0)[b]{\smash{$Q_4^{i+1}$}}}%
    \put(0.17499208,0.28818971){\color[rgb]{0,0,0}\makebox(0,0)[rb]{\smash{$x_{i+1}$}}}%
    \put(0.42199688,0.6552462){\color[rgb]{0,0,0}\makebox(0,0)[b]{\smash{$Q_5^{i+1}$}}}%
    \put(0.29875195,0.02966118){\color[rgb]{0,0,0}\makebox(0,0)[b]{\smash{$x_{i}$}}}%
  \end{picture}%
\endgroup%

%% file: qe2.tex
\begingroup%
  \makeatletter%
  \providecommand\color[2][]{%
    \errmessage{(Inkscape) Color is used for the text in Inkscape, but the package 'color.sty' is not loaded}%
    \renewcommand\color[2][]{}%
  }%
  \providecommand\transparent[1]{%
    \errmessage{(Inkscape) Transparency is used (non-zero) for the text in Inkscape, but the package 'transparent.sty' is not loaded}%
    \renewcommand\transparent[1]{}%
  }%
  \providecommand\rotatebox[2]{#2}%
  \ifx\svgwidth\undefined%
    \setlength{\unitlength}{225.45bp}%
    \ifx\svgscale\undefined%
      \relax%
    \else%
      \setlength{\unitlength}{\unitlength * \real{\svgscale}}%
    \fi%
  \else%
    \setlength{\unitlength}{\svgwidth}%
  \fi%
  \global\let\svgwidth\undefined%
  \global\let\svgscale\undefined%
  \makeatother%
  \begin{picture}(1,0.66939777)%
    \put(0,0){\includegraphics[width=\unitlength]{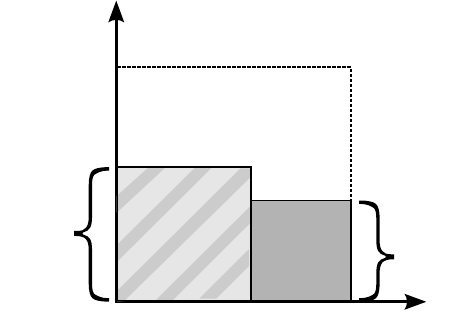}}%
    \put(0.12546213,0.16362306){\color[rgb]{0,0,0}\makebox(0,0)[rb]{\smash{$x$}}}%
    \put(0.87292429,0.10908088){\color[rgb]{0,0,0}\makebox(0,0)[lb]{\smash{$\frac34x$}}}%
    \put(0.64049678,0.11825351){\color[rgb]{0,0,0}\makebox(0,0)[b]{\smash{$Q$}}}%
    \put(0.39210468,0.1537381){\color[rgb]{0,0,0}\makebox(0,0)[b]{\smash{$\delta \approx \frac34$}}}%
  \end{picture}%
\endgroup%

%% file: brickConcept.tex
\begingroup%
  \makeatletter%
  \providecommand\color[2][]{%
    \errmessage{(Inkscape) Color is used for the text in Inkscape, but the package 'color.sty' is not loaded}%
    \renewcommand\color[2][]{}%
  }%
  \providecommand\transparent[1]{%
    \errmessage{(Inkscape) Transparency is used (non-zero) for the text in Inkscape, but the package 'transparent.sty' is not loaded}%
    \renewcommand\transparent[1]{}%
  }%
  \providecommand\rotatebox[2]{#2}%
  \ifx\svgwidth\undefined%
    \setlength{\unitlength}{1155.4359375bp}%
    \ifx\svgscale\undefined%
      \relax%
    \else%
      \setlength{\unitlength}{\unitlength * \real{\svgscale}}%
    \fi%
  \else%
    \setlength{\unitlength}{\svgwidth}%
  \fi%
  \global\let\svgwidth\undefined%
  \global\let\svgscale\undefined%
  \makeatother%
  \begin{picture}(1,0.41287981)%
    \put(0,0){\includegraphics[width=\unitlength]{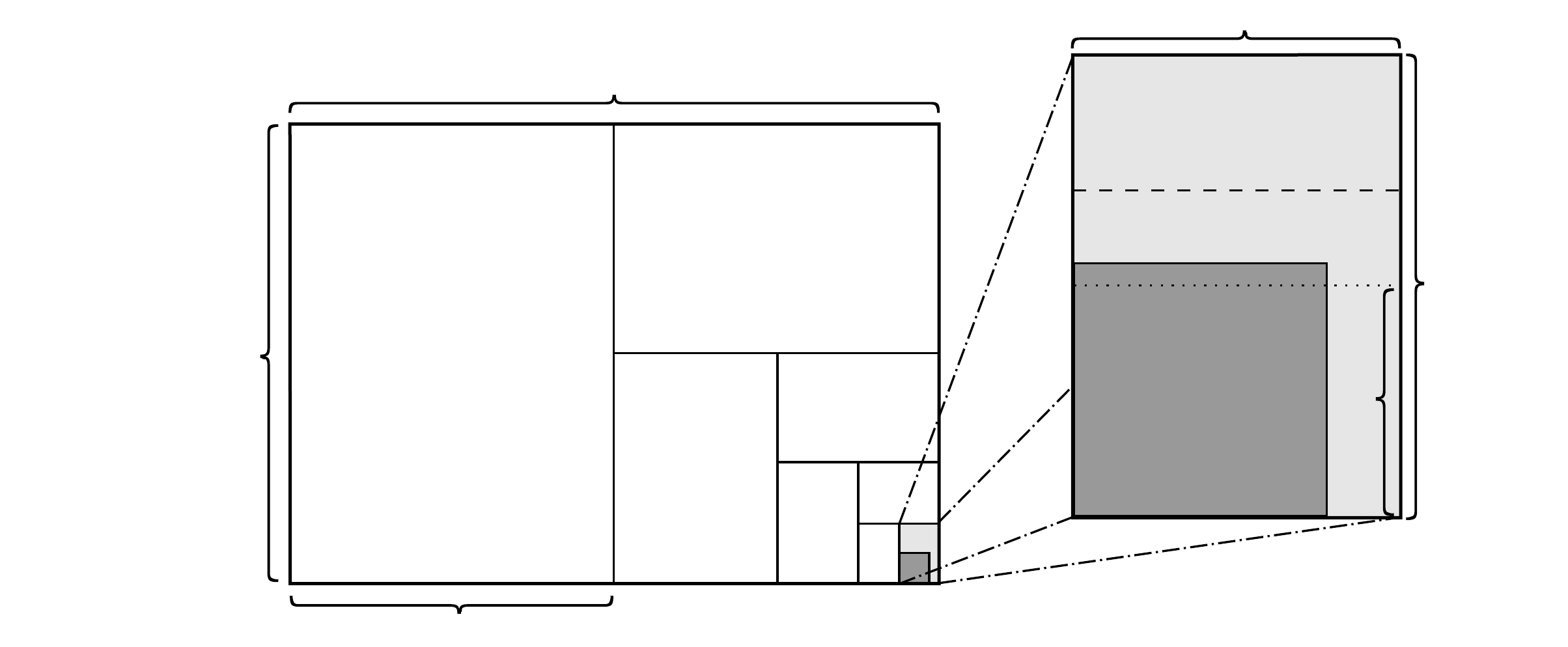}}%
    \put(0.76673446,0.15411528){\color[rgb]{0,0,0}\makebox(0,0)[b]{\smash{$Q$}}}%
    \put(0.39171147,0.36463729){\color[rgb]{0,0,0}\makebox(0,0)[b]{\smash{$b$}}}%
    \put(0.15020034,0.17802798){\color[rgb]{0,0,0}\makebox(0,0)[rb]{\smash{$\frac{b}{\sqrt{2}}$}}}%
    \put(0.24630438,0.01605196){\color[rgb]{0,0,0}\makebox(0,0)[lt]{\begin{minipage}{0.09302136\unitlength}\centering $\frac{b}2$\end{minipage}}}%
    \put(0.9180166,0.22755317){\color[rgb]{0,0,0}\makebox(0,0)[lb]{\smash{$\sqrt{2}s$}}}%
    \put(0.79445934,0.40235889){\color[rgb]{0,0,0}\makebox(0,0)[b]{\smash{$s$}}}%
    \put(0.87096726,0.1537947){\color[rgb]{0,0,0}\makebox(0,0)[rb]{\smash{$\frac{s}{\sqrt{2}}$}}}%
  \end{picture}%
\endgroup%

%% file: brickPartition.tex
\begingroup%
  \makeatletter%
  \providecommand\color[2][]{%
    \errmessage{(Inkscape) Color is used for the text in Inkscape, but the package 'color.sty' is not loaded}%
    \renewcommand\color[2][]{}%
  }%
  \providecommand\transparent[1]{%
    \errmessage{(Inkscape) Transparency is used (non-zero) for the text in Inkscape, but the package 'transparent.sty' is not loaded}%
    \renewcommand\transparent[1]{}%
  }%
  \providecommand\rotatebox[2]{#2}%
  \ifx\svgwidth\undefined%
    \setlength{\unitlength}{459.35585022bp}%
    \ifx\svgscale\undefined%
      \relax%
    \else%
      \setlength{\unitlength}{\unitlength * \real{\svgscale}}%
    \fi%
  \else%
    \setlength{\unitlength}{\svgwidth}%
  \fi%
  \global\let\svgwidth\undefined%
  \global\let\svgscale\undefined%
  \makeatother%
  \begin{picture}(1,0.8240672)%
    \put(0,0){\includegraphics[width=\unitlength]{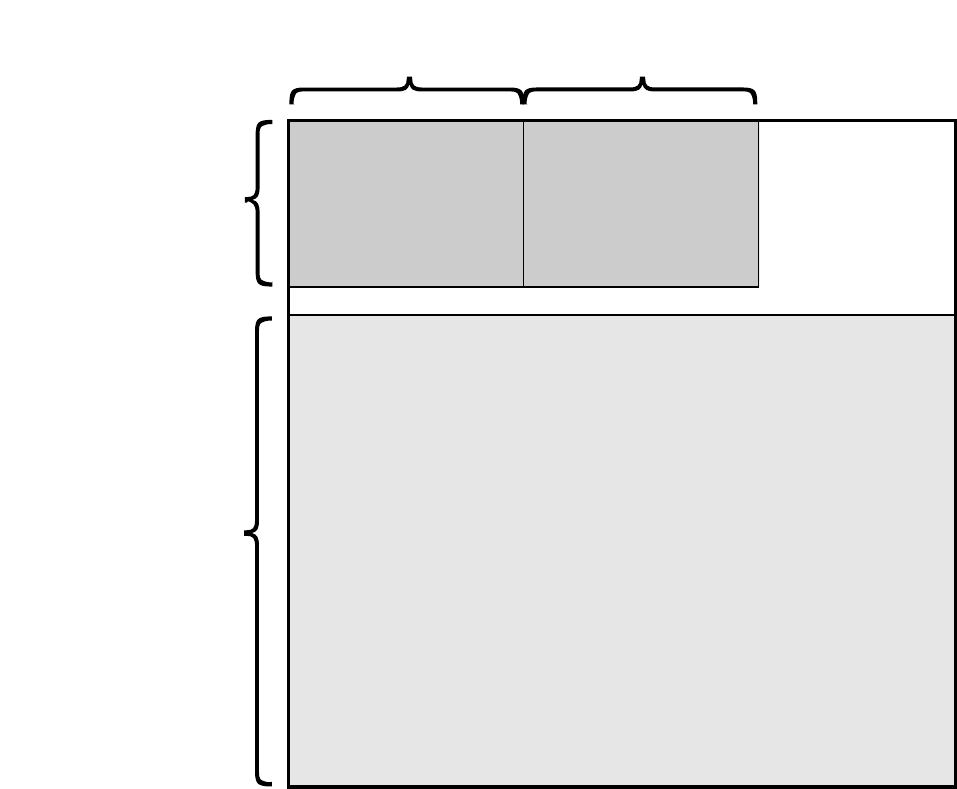}}%
    \put(0.42231534,0.59533503){\color[rgb]{0,0,0}\makebox(0,0)[b]{\smash{$D_1$}}}%
    \put(0.66649755,0.59473743){\color[rgb]{0,0,0}\makebox(0,0)[b]{\smash{$D_2$}}}%
    \put(0.639055,0.24067011){\color[rgb]{0,0,0}\makebox(0,0)[b]{\smash{$B$}}}%
    \put(0.23509483,0.26160519){\color[rgb]{0,0,0}\makebox(0,0)[rb]{\smash{$\frac1{\sqrt2}$}}}%
    \put(0.23386336,0.60026089){\color[rgb]{0,0,0}\makebox(0,0)[rb]{\smash{$\frac14$}}}%
    \put(0.42477829,0.79683817){\color[rgb]{0,0,0}\makebox(0,0)[b]{\smash{$\frac1{2\sqrt2}$}}}%
    \put(0.66737894,0.79560667){\color[rgb]{0,0,0}\makebox(0,0)[b]{\smash{$\frac1{2\sqrt2}$}}}%
  \end{picture}%
\endgroup%

%% file: dynamicBrick1.tex
\begingroup%
  \makeatletter%
  \providecommand\color[2][]{%
    \errmessage{(Inkscape) Color is used for the text in Inkscape, but the package 'color.sty' is not loaded}%
    \renewcommand\color[2][]{}%
  }%
  \providecommand\transparent[1]{%
    \errmessage{(Inkscape) Transparency is used (non-zero) for the text in Inkscape, but the package 'transparent.sty' is not loaded}%
    \renewcommand\transparent[1]{}%
  }%
  \providecommand\rotatebox[2]{#2}%
  \ifx\svgwidth\undefined%
    \setlength{\unitlength}{228.1347877bp}%
    \ifx\svgscale\undefined%
      \relax%
    \else%
      \setlength{\unitlength}{\unitlength * \real{\svgscale}}%
    \fi%
  \else%
    \setlength{\unitlength}{\svgwidth}%
  \fi%
  \global\let\svgwidth\undefined%
  \global\let\svgscale\undefined%
  \makeatother%
  \begin{picture}(1,0.93485759)%
    \put(0,0){\includegraphics[width=\unitlength]{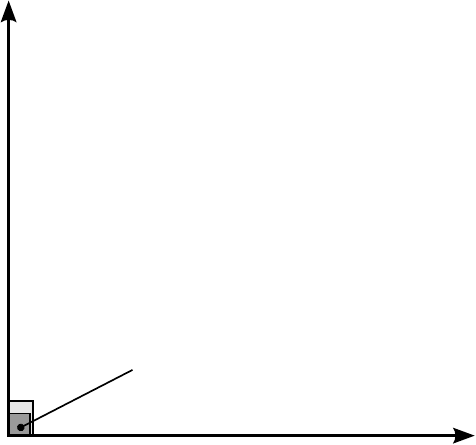}}%
    \put(0.28588711,0.16359928){\color[rgb]{0,0,0}\makebox(0,0)[lb]{\smash{$Q$}}}%
  \end{picture}%
\endgroup%

%% file: dynamicBrick2.tex
\begingroup%
  \makeatletter%
  \providecommand\color[2][]{%
    \errmessage{(Inkscape) Color is used for the text in Inkscape, but the package 'color.sty' is not loaded}%
    \renewcommand\color[2][]{}%
  }%
  \providecommand\transparent[1]{%
    \errmessage{(Inkscape) Transparency is used (non-zero) for the text in Inkscape, but the package 'transparent.sty' is not loaded}%
    \renewcommand\transparent[1]{}%
  }%
  \providecommand\rotatebox[2]{#2}%
  \ifx\svgwidth\undefined%
    \setlength{\unitlength}{228.1347877bp}%
    \ifx\svgscale\undefined%
      \relax%
    \else%
      \setlength{\unitlength}{\unitlength * \real{\svgscale}}%
    \fi%
  \else%
    \setlength{\unitlength}{\svgwidth}%
  \fi%
  \global\let\svgwidth\undefined%
  \global\let\svgscale\undefined%
  \makeatother%
  \begin{picture}(1,0.93485759)%
    \put(0,0){\includegraphics[width=\unitlength]{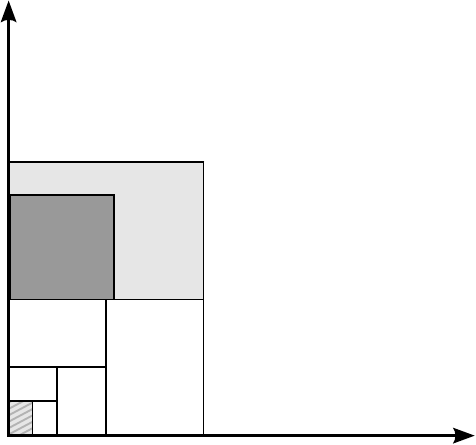}}%
    \put(0.12963012,0.3915347){\color[rgb]{0,0,0}\makebox(0,0)[b]{\smash{$Q$}}}%
  \end{picture}%
\endgroup%

%% file: dynamicBrick3.tex
\begingroup%
  \makeatletter%
  \providecommand\color[2][]{%
    \errmessage{(Inkscape) Color is used for the text in Inkscape, but the package 'color.sty' is not loaded}%
    \renewcommand\color[2][]{}%
  }%
  \providecommand\transparent[1]{%
    \errmessage{(Inkscape) Transparency is used (non-zero) for the text in Inkscape, but the package 'transparent.sty' is not loaded}%
    \renewcommand\transparent[1]{}%
  }%
  \providecommand\rotatebox[2]{#2}%
  \ifx\svgwidth\undefined%
    \setlength{\unitlength}{228.1347877bp}%
    \ifx\svgscale\undefined%
      \relax%
    \else%
      \setlength{\unitlength}{\unitlength * \real{\svgscale}}%
    \fi%
  \else%
    \setlength{\unitlength}{\svgwidth}%
  \fi%
  \global\let\svgwidth\undefined%
  \global\let\svgscale\undefined%
  \makeatother%
  \begin{picture}(1,0.93485759)%
    \put(0,0){\includegraphics[width=\unitlength]{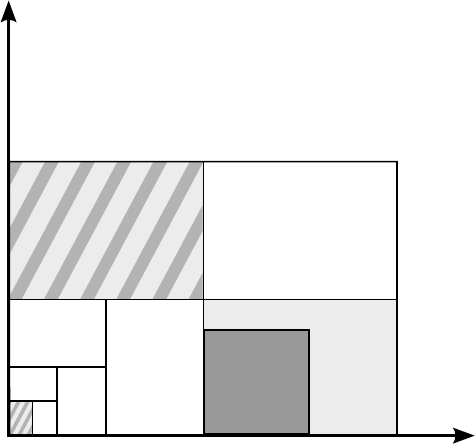}}%
    \put(0.54517393,0.10223205){\color[rgb]{0,0,0}\makebox(0,0)[b]{\smash{$Q$}}}%
  \end{picture}%
\endgroup%

%% file: dynamicBrick4.tex
\begingroup%
  \makeatletter%
  \providecommand\color[2][]{%
    \errmessage{(Inkscape) Color is used for the text in Inkscape, but the package 'color.sty' is not loaded}%
    \renewcommand\color[2][]{}%
  }%
  \providecommand\transparent[1]{%
    \errmessage{(Inkscape) Transparency is used (non-zero) for the text in Inkscape, but the package 'transparent.sty' is not loaded}%
    \renewcommand\transparent[1]{}%
  }%
  \providecommand\rotatebox[2]{#2}%
  \ifx\svgwidth\undefined%
    \setlength{\unitlength}{228.1347877bp}%
    \ifx\svgscale\undefined%
      \relax%
    \else%
      \setlength{\unitlength}{\unitlength * \real{\svgscale}}%
    \fi%
  \else%
    \setlength{\unitlength}{\svgwidth}%
  \fi%
  \global\let\svgwidth\undefined%
  \global\let\svgscale\undefined%
  \makeatother%
  \begin{picture}(1,0.93485759)%
    \put(0,0){\includegraphics[width=\unitlength]{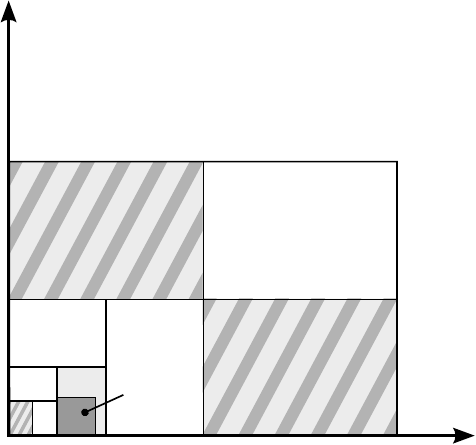}}%
    \put(0.25958687,0.10573875){\color[rgb]{0,0,0}\makebox(0,0)[lb]{\smash{$Q$}}}%
  \end{picture}%
\endgroup%

%% file: conclusion.tex
\section{Conclusion}\label{sec:conclusion}
We have presented progress on two natural variants
of packing squares into a square in an online fashion.
The most immediate open question remains the critical
packing density for a fixed container, where
the correct value may actually be less than $1/2$.
Even though we invested a considerable amount of work into establishing
a lower bound greater than $1/3$,
we believe that there are alternative schemes
that could lead to further improvement. 

Online packing into a dynamic container
remains wide open. There is still possible slack in both bounds; our feeling is
that it should be easier to improve the lower bound rather than the upper bound, 
as there is still considerable room to employ more sophisticated recursive schemes,
just like in the case of a fixed container.

There are many interesting related questions.
What is the critical density (offline and online)
for packing circles into a unit square?
This was raised by Demaine et al.~\cite{dfl-cpodi-11}. 
In an offline setting, there is a lower bound of $\pi/8=0.392\ldots$, and an 
upper bound of $\frac{2\pi}{(2+\sqrt{2})^2}=0.539\ldots$, 
which is conjectured
to be tight. Another question is to consider the critical density
as a function of the size of the largest object. 
In an offline context,
the proof by Moon and Moser provides an answer, but little is known
in an online setting.